\documentclass[12pt,draftclsnofoot,onecolumn]{IEEEtran}

\ifCLASSINFOpdf

\else

\fi

\usepackage{geometry}
\usepackage{setspace}
\usepackage[ruled,vlined]{algorithm2e}
\usepackage{cite}
\usepackage{color}
\usepackage{graphicx}
\usepackage{subfig}
\usepackage{amsfonts}
\usepackage{amsmath}
\usepackage{amsthm}
\usepackage{mathdots}
\usepackage{bm}
\usepackage{esint}
\usepackage{amssymb}

\newtheorem{definition}{Definition}

\newtheorem{proposition}{Proposition}
\newtheorem{remark}{Remark}

\makeatletter
\renewcommand{\maketag@@@}[1]{\hbox{\m@th\normalsize\normalfont#1}}%
\makeatother

\begin{document}


\begin{spacing}{1.35}

\title{{\color{black}Joint Active and} Passive Beamforming Design for Reconfigurable Intelligent Surface Enabled Integrated Sensing and Communication}

\author{
        Zhe Xing, \textit{Graduate Student Member, IEEE},
        Rui Wang, \textit{Senior Member, IEEE}, and
        Xiaojun Yuan, \textit{Senior Member, IEEE}
\thanks{Z. Xing and R. Wang are with the College of Electronics and Information Engineering, Tongji University, Shanghai 201804, China. R. Wang is also with the Shanghai Institute of Intelligent Science and Technology, Tongji University, Shanghai 201804, China (e-mail: zxing@tongji.edu.cn; ruiwang@tongji.edu.cn).

X. Yuan is with the National Key Laboratory of Science and Technology on Communications, University of Electronic Science and Technology of China, Chengdu, 610000, China (e-mail: xjyuan@uestc.edu.cn).

}
}

\markboth{}%
{\MakeLowercase{\textit{et al.}}}

\maketitle

\begin{abstract}

To exploit the potential of the reconfigurable intelligent surface (RIS) in supporting integrated sensing and communication (ISAC), this paper proposes a novel {\color{black}joint active and} passive beamforming design for RIS-enabled ISAC system in consideration of the target size.  First, the detection probability for target sensing is derived in closed-form based on the illumination power on an approximated scattering surface area of the target, and a new concept of ultimate detection resolution (UDR) is defined for the first time to measure the target detection capability. Then, an optimization problem is formulated to maximize the signal-to-noise ratio (SNR) at the user-equipment (UE) under a minimum detection probability constraint.  {\color{black}To solve the non-convex problem, a novel alternative optimization approach is developed. In this approach, the solutions of the communication and sensing beamformers are obtained by our proposed bisection-search based method. The optimal receive combining vector is derived from an equivalent Rayleigh-quotient problem. To optimize the RIS phase shifts, the Charnes-Cooper transformation is conducted to cope with the fractional objective, and a novel convexification process is proposed to convexify the detection probability constraint with matrix operations and a real-valued first-order Taylor expansion. After the convexification, a successive convex approximation (SCA) based algorithm is designed to yield a suboptimal phase-shift solution. Finally, the overall optimization algorithm is built, followed by detailed analysis on its computational complexity,  convergence behavior and problem feasibility condition.}
Extensive simulations are carried out to testify the analytical properties of the proposed beamforming design, and to reveal two important trade-offs, namely, communication vs. sensing trade-off and UDR vs. sensing-duration trade-off. In comparison with several existing benchmarks, our proposed approach is validated to be superior when detecting targets with practical sizes.

\end{abstract}


\end{spacing}

\IEEEpeerreviewmaketitle


\section{Introduction}

The forthcoming beyond fifth- (B5G) and sixth-generation (6G) mobile communications have been envisioned as pivotal enablers for many innovative applications, such as the autonomous mobility, virtual/augmented reality (VR/AR),  digital twin, and human-machine interaction, etc. \cite{6G3}. Supporting these applications requires a tight cooperation of wireless communication and environmental sensing, which have been concurrently developed with rare coordination and mutual benefit for decades \cite{JCR-1,JCR-2}. Recently, owing to their commonalities in regard of signal processing methods, hardware platforms and system architectures, etc. \cite{JCR-2}, the coexistence and merging of the two individual functionalities have attracted considerable interest, thereby promoting the emergence and development of a novel paradigm shift, termed integrated sensing and communication (ISAC) \cite{ISAC}. 

The ISAC can be performed with the aid of various key enabling technologies, including the millimeter wave (mmWave), ultra-dense network (UDN) and multiple-input-multiple-output (MIMO) radar. These technologies have been incorporated to boost the communication and sensing capabilities by improving the spectral efficiency and spatial degrees of freedom (DoF), but are still unable to adequately address several critical challenges \cite{My-TWC, My-ICCC}. For instance, the mmWave is highly susceptible to obstructions on the line-of-sight (LoS) path, and suffers from severe propagation loss in the atmosphere. To compensate for the resultant signal attenuation over the wireless channel, higher transmit power and/or antenna gain are requisite to enhance the emitted signal strength, thereby increasing the energy consumption (EC) to a large extent. 
In addition, the dense deployment of the ISAC base stations (BSs) with massive MIMO arrays brings about high hardware cost (HC), while further pushing the total EC in the network to an exorbitant level. 
To tackle these issues, recent attention has been paid to a new burgeoning concept termed reconfigurable intelligent surface (RIS), or intelligent reflecting surface (IRS) \cite{Propose-IRS-1}, which early appeared as a prototype of intelligent wall \cite{Intelligent-Wall}, and was developed by the landmark works \cite{Landmark-1, Landmark-2} three years ago.
An RIS is a near-passive reflecting metasurface with many small controllable units, which can be digitally configured to perform \textit{passive beamforming} by changing the physical properties of the impinging electromagnetic wave (such as the phase-shift), so as to create a reliable virtual LoS link and manipulate the propagation environment intelligently{\color{black}\cite{IRS-Survey-1,IRS-Survey-2, Q.Wu2020(CM)}}. Since the RIS is generally fabricated with cheap hardware components without energy-consuming radio-frequency (RF) chains \cite{L.Dai2020(Access)}, it is envisioned as a promising candidate technology in compliant with the notion of green communication in B5G and 6G.

Up to now, plenty of researches have been focusing on the performance analysis and application potential of the RIS in both communication and sensing fields. As for communication, the RIS was leveraged to transfer passive information \cite{Information-Transfer}, build index modulation scheme \cite{Index-Modulation} and achieve secure physical-layer transmission{\color{black} \cite{Secure-Transmission, Secure-Transmission-2,Secure-Transmission-3}}, owing to the flexibility of the phase-shift adjustment. More importantly, it was revealed that a sufficiently large RIS with $N$ reflecting units could yield a remarkable signal-to-noise ratio (SNR) gain by $\mathcal{O}(N^2)$ over the cascaded channel \cite{N2}. In light of this, the RIS was also widely employed to {\color{black}combat the unfavourable channel conditions by forming desired passive beams \cite{R3-Comm3-2,R3-Comm3-4}}, so as to improve the achievable data-rate \cite{My-TWC, ACR-Maximization-2}, spectral/energy efficiency \cite{Landmark-2}, outage probability \cite{Outage-Probability, Outage-Probability-2} and bit-error-rate (BER) performance \cite{BER} of the assisted wireless communication system. 
As for sensing, some prior works exploited the RIS to perform user localization{\color{black}\cite{R3-Comm3-1}} in combination with the codebook search \cite{Localization-1}, the phase-shift profile design as well as the parameter estimation{\color{black}\cite{R3-Comm3-3,Localization-2, Localization-3}}, and incorporated the RIS into the conventional radar system to assist the target detection by providing additional reception link for echoes \cite{RIS-Radar-1, RIS-Radar-2, RIS-Radar-3}. According to their results, such meta-localization and meta-radar systems were validated to be able to outperform the traditional ones without the RIS, especially when the reflection arrays were fabricated to be large.

Owing to the benefits brought by the utilization of the RIS, recent progresses have been made to incorporate the RIS into the ISAC system. For instance, Jiang \textit{et al.} \cite{RIS-ISAC-1}, first introduced the RIS to the dual-function radar and communication (DFRC) system with single user-equipment (UE) and single target, and jointly optimized the reflection matrix and the transmit precoder. The optimization problem was solved by the semidefinite relaxation (SDR) and the majorization–minimization (MM) methods. After that, Song \textit{et al.} \cite{RIS-ISAC-2}, further considered a single-user RIS-ISAC scenario with multiple targets to be detected, and proposed to maximize the minimum beampattern gain in several sensing directions under the transmit power constraint at the BS and the SNR constraint at the UE. Liu \textit{et al.} \cite{RIS-ISAC-3}, extended the beamforming design to a multi-user RIS-ISAC system, and alternatively optimized the transmit beamformer, receive filter and reflection coefficients by maximizing the sum-rate of the UE under the radar SNR constraint. Wang \textit{et al.} \cite{RIS-ISAC-4}, performed the joint waveform and phase-shift design in an RIS-aided multi-user DFRC system to minimize the multi-user interference (MUI). Tong \textit{et al.} \cite{RIS-ISAC-Huang}, exploited the RIS to assist the uplink multi-user ISAC by dividing the sensing space into several blocks, and applied the generalized approximate message passing to determine the environmental information.
Unlike these works considering the completely passive and continuous RIS, Prasobh Sankar \textit{et al.} \cite{RIS-ISAC-5}, introduced the hybrid RIS to the multi-user ISAC system, where partial elements on the RIS were designed to be active while the others remained passive. Wang \textit{et al.} \cite{RIS-ISAC-6}, considered a more practical scenario where the RIS phase-shifts were discrete, and took the Cramér-Rao bound (CRB) of the direction-of-arrival (DOA) estimation as a sensing performance metric.

In these prior works, heterogeneous RIS-aided beamforming strategies were developed to jointly fulfil the communication and sensing demands. Although these initial attempts made a big step forward in this direction, some challenging problems remained unsolved. First, the targets to be sensed are mostly treated as points. However, the cross-section area of a target is generally non-ignorable because it is essentially associated with the scattering capability. 
Second, under the assumption of point target, the beamforming design is confined to the conventional end-to-end channel model and abandons much freedom in sensing beampattern adjustment. Such freedom is crucial for further improvement of the target detection performance in an RIS-ISAC system.  Consequently, it is of significance to develop a new RIS-enabled beamforming strategy \textit{in consideration of the target size}, which motivates our research. In this paper, {\color{black}considering a joint communication and target detection problem in the ISAC field}, we take the target size into account and optimize the performances of UE communication and target detection by developing a novel joint active and passive beamforming optimization scheme. On this basis, we further introduce a new concept for sensing capability measurement and provide extensive analysis on the proposed solutions. 
Our contributions are summarized as follows.

\begin{itemize}
\item[•] \textbf{Derivation of the Detection Probability and Definition of the Ultimate Detection Resolution (UDR)}: In this work, a physics-based model is adopted to characterize the practical RIS reflection. Considering the target size, the scattering surface area of the target is approximated as a smooth surface in accordance with empirical radar cross-section (RCS) measures of practical targets. Then, based on the illumination power computed integrally over the scattering surface area, the detection probability of the target is derived in closed-form, according to which a new concept of UDR is then mathematically defined for the first time. The UDR has not been reported in any previous works to the best of our knowledge, but is essential for sensing capability measurement by characterizing the minimum size of the detectable target under a certain system setup and a given detection probability requirement. 


\item[•] \textbf{Joint Active and Passive Beamforming Design for the RIS-ISAC}: {\color{black}An optimization problem is formulated to maximize the SNR at the UE under a minimum detection probability constraint, by optimizing the communication and sensing (C\&S) beamformers, the receive combining vector, and the RIS phase shifts. 
To solve the non-convex problem, a novel alternative optimization approach is developed in this work. In specific, the C\&S beamformers are optimized by our proposed bisection-search based algorithm. The receive combining vector is optimized by transforming the original subproblem into an equivalent Rayleigh-quotient problem. Regarding the RIS phase-shift optimization, the Charnes-Cooper transformation is performed to deal with the fractional objective function. Then, a new convexification process is proposed to linearise the detection probability constraint by means of a series of matrix manipulations and a real-valued first-order Taylor expansion. After the convexification, a successive convex approximation (SCA) based algorithm is designed to iteratively find the phase-shift solution. 
Finally, the overall optimization algorithm is built, with its computational complexity, convergence behavior and problem feasibility condition being analysed in detail. It is rigorously proved that the developed algorithm is convergent, and the UDR can be determined by the problem feasibility condition. }

\item[•] \textbf{Performance Evaluation}: The analytical properties of the proposed beamforming design are testified via extensive simulations. Two important trade-offs, namely, 1) communication vs.  sensing trade-off, and 2) UDR vs. sensing-duration trade-off, are numerically investigated. Then, the proposed design is compared with several existing benchmarks, including two state-of-the-art approaches which assume point targets. Our results validate that the proposed design significantly outperforms the approach which merely considers the BS-RIS-target link, and is superior to the approach in consideration of the echo link when detecting targets with practical sizes.
\end{itemize}

The rest of this paper is organized as follows. Section II describes the system model and problem formulation. Section III proposes the joint active and passive beamforming optimization approach. Section IV carries out the simulations for performance evaluations. Section V draws the final conclusion.

\textit{Notations:} $(\cdot)^{*}$, $(\cdot)^{\mathrm{T}}$, $(\cdot)^{\mathrm{H}}$ and $(\cdot)^{-1}$ denote the conjugate, transpose, Hermitian and inversion operators, respectively. $|\cdot|$ and $\|\cdot\|_2$ denote the modulus and $\ell_2$ norm. $\mathbb{E}\{\cdot\}$, $\mathrm{rank}(\cdot)$ and $tr(\cdot)$ are the expectation, the rank and the trace. $\otimes$ denotes the Kronecker product. $\oiint \limits_{S_r}\ (\cdot)\ dS$ is the surface integral over $S_r$. $\mathrm{vec}(\cdot)$ and $\mathrm{vec}^{-1}(\cdot)$ stand for the vectorize and matrixing operators. $\mathfrak{Re}(\cdot)$ and $\mathfrak{Im}(\cdot)$ represent the real part and imaginary part. $\left\langle \mathbf{x}, \mathbf{y} \right\rangle$ denotes the inner product of $\mathbf{x}$ and $\mathbf{y}$.




\section{System Model and Problem Statement}

\begin{figure}[h]
\includegraphics[width=3.6in]{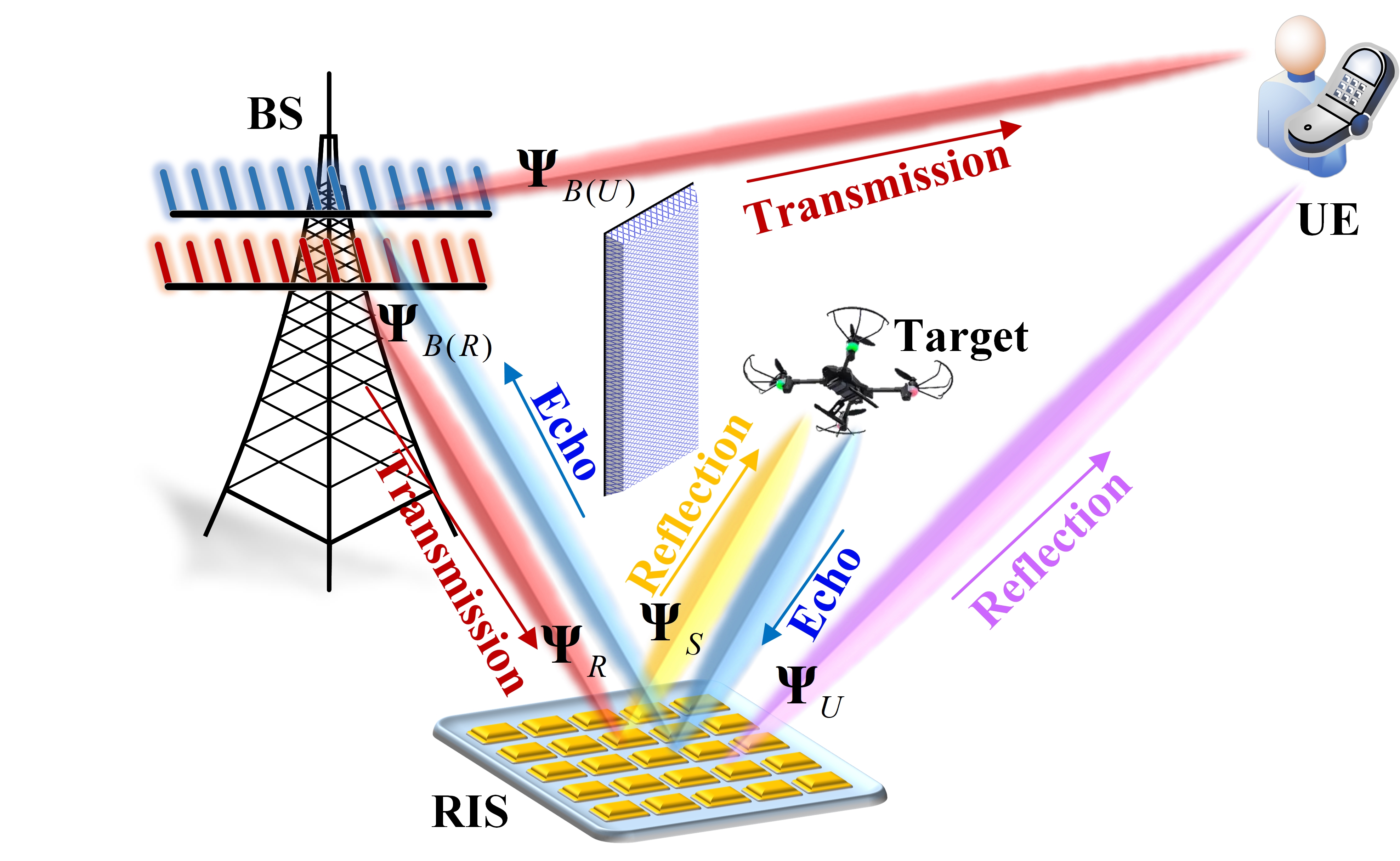}
\hfil
\centering
\caption{\color{black}The considered RIS-ISAC system, where 
the RIS is leveraged to reflect the ISAC signal from the BS to the UE and the target for communication and sensing purposes. A portion of the signal wave scattered by the target is then reflected back to the BS for target detection.}
\label{System Model}
\end{figure}

As depicted in Fig. 1, this paper considers an RIS-ISAC system composed of a BS equipped with $M$ transmitting and $M$ receiving antennas arranged in uniform linear arrays (ULAs) with antenna spacing of $d$, an RIS equipped with $N=N_xN_y$ reflecting elements in a uniform planar array (UPA) with element spacing of $d$, a target to be detected, and a single-antenna communication UE. The UE is assumed to be static. The LoS path between the BS and the target is blocked by an obstruction. 
The BS is responsible for simultaneously communicating with the UE and detecting the target. {\color{black}Since the target near or around the BS can be easily sensed by the direct-link echoes \cite{RIS-ISAC-4, RIS-ISAC-6}, this system is dedicated to the detection of the target behind the obstacle. In such a scenario, the direct BS-target-BS echo link is substantially weak and unavailable \cite{RIS-ISAC-2,Secure_ISAC}, due to double environmental scatterings of the echoes in the NLoS paths between the BS and the target}. To attain a favourable communication/sensing performance, {\color{black}we assume that the RIS has been appropriately deployed to establish a strong virtual LoS link, which can be achieved by the RIS placement optimizations \cite{R3-Comm4-1}}. By properly adjusting the phase shifts, the RIS is expected to produce considerable passive beamforming gains toward the UE and the sensing direction at the same time. When the target is illuminated by the reflective beam, it scatters the impinging electromagnetic wave to the entire propagation space. A portion of the scattered wave is then reflected back to the BS by the RIS as the echo, and is eventually harvested at the BS to determine the presence of the target. Self-interference cancellation is performed at the BS side, so that the disturbance of the transmitted signal on the echoes can be avoided \cite{RIS-ISAC-1}.

\begin{figure}[h]
\includegraphics[width=6.4in]{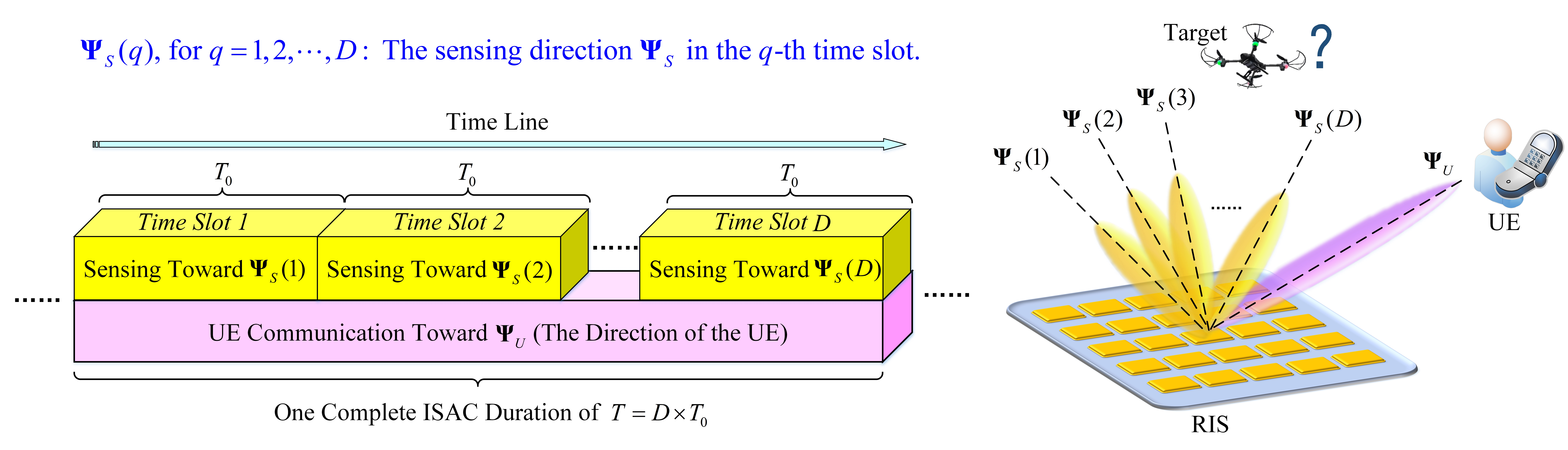}
\hfil
\centering
\caption{Illustration of the joint sensing and communication scheme over the time line.}
\label{ISAC_Working_Process}
\end{figure}

{\color{black}In our work, the exact position of the target is considered to be unknown. In this case, we aim to detect the presence of a possible target within a certain sensing region of $\mathcal{S}_{T}=\{(r,\theta,\varphi)\ |\ r\in (0,r_{\mathrm{max}}], \theta\in[\theta_1,\theta_2],\varphi\in[\varphi_1,\varphi_2]\}$ from the view of the RIS, where $r$, $\theta$ and $\varphi$ are the distance, elevation angle-of-departure (AOD) and azimuth AOD from the RIS to the possible target, respectively. }The sensing and communication tasks can be jointly performed as described in Fig. 2. Specifically, we consider one complete ISAC round as an example, which is evenly divided into $D$ transient time slots, with each time slot lasting for a time duration of $T_0$. In one time slot, the RIS updates its phase-shift variables to simultaneously generate passive beamforming gains toward the direction $\mathbf{\Psi}_U=(\theta_U,\varphi_U)$ for UE communication, and toward a sensing direction $\mathbf{\Psi}_S=(\theta_S,\varphi_S)$ for target detection, {\color{black}where $\theta_1\leq\theta_S\leq\theta_2$ and $\varphi_1\leq\varphi_S\leq\varphi_2$}. Here, $\mathbf{\Psi}_U$ is fixed owing to the assumption of motionless UE, whereas $\mathbf{\Psi}_S$ remains unchanged within one time slot but varies from time slot to time slot to scan over the neighbourhood. A sensing beam is assigned to each $\mathbf{\Psi}_S$ with a time duration of $T_0$. As such, through the environmental scanning in the entire ISAC round, the direction in which the target appears can be determined. It is noteworthy that passive beamforming needs to be optimized in each time slot, so as to change the sensing beam toward different $\mathbf{\Psi}_S$. In fact, it is sufficient to focus on the beamforming optimization in one time slot based on a specific $\mathbf{\Psi}_S$ and $\mathbf{\Psi}_U$, while those for other time slots can be done in a similar way. Therefore, we will consider a specific $\mathbf{\Psi}_S$ and $\mathbf{\Psi}_U$, followed by an elaboration of the communication and sensing performance metrics, as well as our problem formulation.


\subsection{Communication Performance Metric}
We first describe the UE communication performance metric.
{\color{black}Let the transmitted ISAC waveform at the BS be denoted by $\mathbf{x}(t)=\mathbf{w}_{c} c(t)+\mathbf{w}_{s} s(t)$ \cite{RIS-ISAC-1, RIS-ISAC-3}, where $\mathbf{w}_{c}\in \mathbb{C}^{M\times 1}$ and $\mathbf{w}_{s}\in \mathbb{C}^{M\times 1}$ are the communication and sensing beamformers; $c(t)$ and $s(t)$ are the information-carrying signal and sensing signal, satisfying $\mathbb{E}\{|c(t)|^2\}=1$ and $\mathbb{E}\{|s(t)|^2\}=1$.  Both $\mathbf{w}_{c}$ and $\mathbf{w}_{s}$ are constrained by the total transmit power limit of $tr(\mathbf{w}_{c}\mathbf{w}_{c}^\mathrm{H}) + tr(\mathbf{w}_{s}\mathbf{w}_{s}^\mathrm{H}) \leq P_{tx}$. $c(t)$ and $s(t)$ are assumed to be statistically independent and uncorrelated, i.e. $\mathbb{E}\{s(t)c^\mathrm{*}(t)\}=0$ \cite{Secure_ISAC}.}

Based on $\mathbf{x}(t)$, the signal received by the UE is expressed as {\color{black}
 \begin{equation}\label{received-communication-signal}
y_{u}(t) =\left[ \mathbf{h}_{RU}^{\mathrm{H}}\mathbf{\Omega}(\mathbf{\Psi}_{R},\mathbf{\Psi}_{U})\mathbf{H}_{BR} + \mathbf{h}_{BU}^{\mathrm{H}} \right] \mathbf{x}(t) + n_u(t),
\end{equation}\normalsize
where} $n_u(t)$ is the complex zero-mean additive white Gaussian noise at the UE with variance of $\sigma_{n,u}^2$; $\mathbf{\Psi}_{R}=\left(\theta_{R},\varphi_{R}\right)$ denotes the incident direction, with $\theta_{R}$ and $\varphi_{R}$ being the elevation and azimuth angle-of-arrivals (AOAs) at the RIS; $\mathbf{\Omega}(\mathbf{\Psi}_{R},\mathbf{\Psi}_{U})$ is the diagonal RIS response matrix with respect to $\mathbf{\Psi}_{R}$ and $\mathbf{\Psi}_{U}$;
{\color{black}$\mathbf{H}_{BR}$, $\mathbf{h}_{RU}$ and $\mathbf{h}_{BU}$ are the BS-RIS, the RIS-UE and the BS-UE channels, modelled as 
 \begin{equation}\label{H_BR and h_RU}
\mathbf{H}_{BR} = \sqrt{\frac{\rho_{BR}\kappa}{1+\kappa}} \overline{\mathbf{H}}_{BR} + \sqrt{\frac{\rho_{BR}}{1+\kappa}} \widetilde{\mathbf{H}}_{BR},
 \end{equation}
 \begin{equation}
\mathbf{h}_{XU} = \sqrt{\frac{\rho_{XU}\kappa}{1+\kappa}} \overline{\mathbf{h}}_{XU} + \sqrt{\frac{\rho_{XU}}{1+\kappa}} \widetilde{\mathbf{h}}_{XU}, \ X\in \{R,B\}
\end{equation}\normalsize
where $\kappa$ is the Rician factor; $\rho_{BR}$ and $\rho_{XU}$ denote the large-scale path losses, given by
 \begin{equation}
\rho_{BR} = \zeta_0 \left(\frac{d_{BR}}{d_0}\right)^{-\alpha_{BR}},\ \rho_{XU} = \zeta_0 \left(\frac{d_{XU}}{d_0}\right)^{-\alpha_{XU}},\ X\in \{R,B\}
\end{equation}\normalsize
where $\zeta_0$ is the path loss coefficient at $d_0=1$ m; $\alpha_{BR}$ and $\alpha_{XU}$ are the path loss exponents; $d_{BR}$ and $d_{XU}$ denote the distances between the BS and the RIS, and between the BS/RIS and the UE.
$\widetilde{\mathbf{H}}_{BR}$ and $\widetilde{\mathbf{h}}_{XU}$ are the NLoS Rayleigh fading components.
$\overline{\mathbf{H}}_{BR}$ and $\overline{\mathbf{h}}_{XU}$ are the LoS components given by
 \begin{equation}
\overline{\mathbf{H}}_{BR} = \mathbf{a}(\theta_R,\varphi_R) \mathbf{a}^\mathrm{H}(\theta_{B(R)},\varphi_{B(R)}),\ \overline{\mathbf{h}}_{RU} = \mathbf{a}(\theta_U,\varphi_U),\ \overline{\mathbf{h}}_{BU} = \mathbf{a}(\theta_{B(U)},\varphi_{B(U)}),
\end{equation} \normalsize
where} $\theta_{B(R)}$ or $\theta_{B(U)}$ and $\varphi_{B(R)}$ or $\varphi_{B(U)}$ are the elevation and azimuth AODs at the BS toward the RIS or UE; $\mathbf{a} \left(\theta_{B(R)},\varphi_{B(R)}\right)$, $\mathbf{a} \left(\theta_{B(U)},\varphi_{B(U)}\right)$, $\mathbf{a} \left(\theta_{R},\varphi_{R}\right)$ and $\mathbf{a} \left(\theta_{U},\varphi_{U}\right)$ are the steering vectors, given by
\small \begin{equation}  \nonumber
\mathbf{a} \left(\theta_{B(X)},\varphi_{B(X)}\right)
= \left( 1, e^{j\frac{2\pi d}{\lambda} \sin{\theta_{B(X)}}\sin{\varphi_{B(X)}}},\cdots, e^{j\frac{2\pi d}{\lambda} (M-1) \sin{\theta_{B(X)}}\sin{\varphi_{B(X)}}} \right)^{\mathrm{T}},\ X\in \{R,U\},
\end{equation}
\begin{equation}\small\nonumber
\begin{split}
\mathbf{a} \left(\theta_{R},\varphi_{R}\right)
= & \left( 1, e^{j\frac{2\pi d}{\lambda} \sin{ \theta_{R} }\cos{ \varphi_{R} }},\cdots, e^{j\frac{2\pi d}{\lambda} (N_x-1) \sin{ \theta_{R} }\cos{ \varphi_{R} }} \right)^{\mathrm{T}}   \!\!
 \otimes 
\left( 1, e^{j\frac{2\pi d}{\lambda} \sin{ \theta_{R} }\sin{ \varphi_{R} }},\cdots, e^{j\frac{2\pi d}{\lambda} (N_y-1) \sin{ \theta_{R} }\sin{ \varphi_{R} }} \right)^{\mathrm{T}} ,
\end{split}
\end{equation}\normalsize
\begin{equation}\small\nonumber
\begin{split}
\mathbf{a} \left(\theta_{U},\varphi_{U}\right)
= & \left( 1, e^{j\frac{2\pi d}{\lambda} \sin{ \theta_{U} }\cos{ \varphi_{U} }},\cdots, e^{j\frac{2\pi d}{\lambda} (N_x-1) \sin{ \theta_{U} }\cos{ \varphi_{U} }} \right)^{\mathrm{T}}  \!\!
  \otimes 
\left( 1, e^{j\frac{2\pi d}{\lambda} \sin{ \theta_{U} }\sin{ \varphi_{U} }},\cdots, e^{j\frac{2\pi d}{\lambda} (N_y-1) \sin{ \theta_{U} }\sin{ \varphi_{U} }} \right)^{\mathrm{T}} .
\end{split}
\end{equation}\normalsize

{\color{black}In this paper, we assume that the channel state information (CSI) of the BS-RIS, RIS-UE and BS-UE channels are known, which can be achieved by various existing channel estimation techniques proposed in e.g. \cite{Channel Estimation} for RIS-aided systems}. Besides, we use the physics-based model in \cite{Physical-Model-TCOM, Physical-Model-Yuan} to characterize the practical RIS reflection, where $\mathbf{\Omega}(\mathbf{\Psi}_{R},\mathbf{\Psi}_{U})$ is modelled as \cite{Physical-Model-TCOM}
 \begin{equation}
\mathbf{\Omega}(\mathbf{\Psi}_{R},\mathbf{\Psi}_{U}) = \frac{\sqrt{4\pi}}{\lambda} g_{uc}(\mathbf{\Psi}_{R},\mathbf{\Psi}_{U}) \mathrm{diag}(\bm{\omega}),
\end{equation}\normalsize
where $\bm{\omega}=(e^{j\beta_1},e^{j\beta_2},...,e^{j\beta_N})^\mathrm{T}$ is the unit-modulus adjustable phase-shift vector with $\beta_i$ for $i=1,2,...,N$ being the phase-shift variables, and $g_{uc}(\mathbf{\Psi}_{R},\mathbf{\Psi}_{U})$ is the inherent unit-cell response factor. Note that an expression of $g_{uc}(\mathbf{\Psi}_{R},\mathbf{\Psi}_{U})$ is derived in [45, Eq. (16)] based on the physical reflection properties. Interested readers may refer to \cite{Physical-Model-TCOM} for more details.

For the signal model in (\ref{received-communication-signal}), the SNR at the UE with respect to $\mathbf{w}_c,\mathbf{w}_s$ and $\bm{\omega}$ is {\color{black}given by
 \begin{equation}\label{Comm_Metric}
\mathrm{SNR_{UE}}(\mathbf{w}_c,\mathbf{w}_s,\bm{\omega}) = \frac{| [\mathbf{h}_{RU}^{\mathrm{H}}\mathbf{\Omega}(\mathbf{\Psi}_{R},\mathbf{\Psi}_{U})\mathbf{H}_{BR} + \mathbf{h}_{BU}^{\mathrm{H}}] \mathbf{w}_{c} |^2}{| [\mathbf{h}_{RU}^{\mathrm{H}}\mathbf{\Omega}(\mathbf{\Psi}_{R},\mathbf{\Psi}_{U})\mathbf{H}_{BR} + \mathbf{h}_{BU}^{\mathrm{H}}] \mathbf{w}_{s} |^2 + \sigma_{n,u}^2},
\end{equation}\normalsize
which is} employed as the communication performance metric.

\subsection{Sensing Performance Metric}

The target detection is performed based on the echo power depending on the target size. 
To facilitate the analysis, we consider detecting possible target in a specific sensing direction of $\mathbf{\Psi}_S=(\theta_S,\varphi_S)$ {\color{black}at a distance of $r\in (0,r_{\mathrm{max}}]$ away from the RIS}, as shown in Fig. \ref{HSSA}. Due to the randomness and irregularity of practical targets, the scattering surface area of the target is approximated as a smooth surface, in accordance with the empirical measures of the RCSs of practical targets [47, Table I].
The scattering surface area, denoted by $S(\Delta_\theta,\Delta_\varphi)$, has an elevation angle-spread of $\Delta_\theta$ and an azimuth angle-spread of $\Delta_\varphi$ from the view of the RIS. By analysing the illumination power on $S(\Delta_\theta,\Delta_\varphi)$, we will derive the detection probability of the target in closed-form. Then, based on the detection probability, we will also introduce a new concept of UDR for detection capability measurements.

\begin{figure}[h]
\includegraphics[width=2.5in]{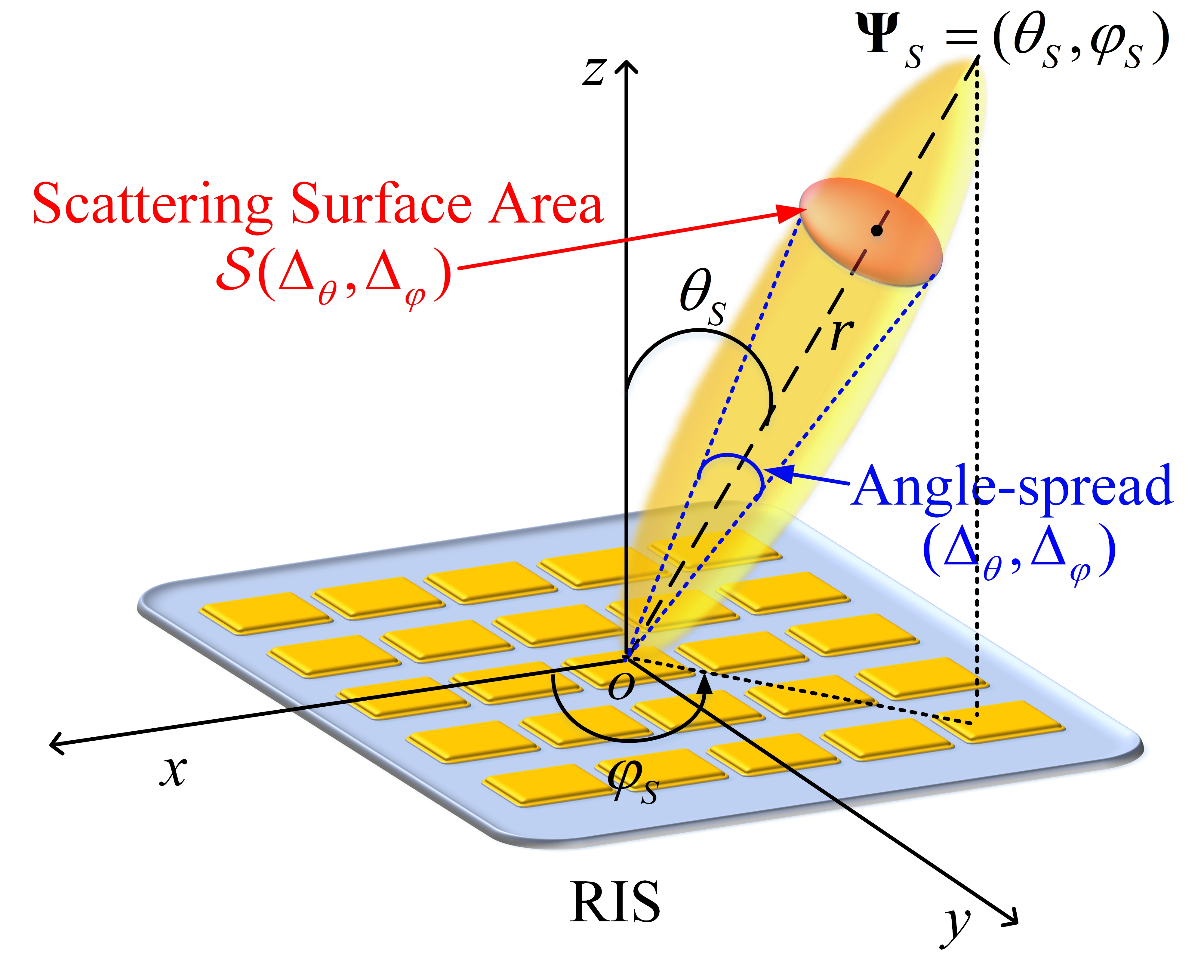}
\hfil
\centering
\caption{The visual illustration of the scattering surface area of the possible target.}
\label{HSSA}
\end{figure}

To begin with, we analyse the signal power in each transmission hop. When the RIS is reflecting the signal from the BS to an arbitrary direction $\mathbf{\Psi}=(\theta,\varphi)$, the reflective radiation power pattern toward $\mathbf{\Psi}=(\theta,\varphi)$ {\color{black}is expressed as
\begin{align}
P_R\left(\mathbf{\Psi}\right)&=
\mathbb{E}\left\{
\left|  \mathbf{a}^{\mathrm{H}}(\theta,\varphi) \mathbf{\Omega}(\mathbf{\Psi}_{R},\mathbf{\Psi}) \mathbf{H}_{BR}[\mathbf{w}_{c} c(t) + \mathbf{w}_{s} s(t)]\right|^2  \right\} \nonumber\\
&= \frac{4\pi}{\lambda^2}
 |g_{uc}(\mathbf{\Psi}_{R},\mathbf{\Psi})|^2  \mathbf{a}^{\mathrm{H}}(\theta,\varphi) \mathrm{diag}(\bm{\omega})\mathbf{H}_{BR}(\mathbf{w}_{c}\mathbf{w}_{c}^\mathrm{H} + \mathbf{w}_{s}\mathbf{w}_{s}^\mathrm{H}) \mathbf{H}_{BR}^\mathrm{H} \mathrm{diag}(\bm{\omega}^\mathrm{H}) \mathbf{a}(\theta,\varphi) ,
\end{align}\normalsize
where the expression of $\mathbf{a}(\theta,\varphi)$ is similar to that of $\mathbf{a}(\theta_U,\varphi_U)$ with the subscript $U$ being removed.
Then}, at the distance of $r$ away from the RIS, the reflective radiation power density is given by $\frac{P_R\left(\mathbf{\Psi}\right)}{4\pi r^2}$, and the illumination power on a differential surface area of $dA$ is given by $\frac{P_R\left(\mathbf{\Psi}\right)}{4\pi r^2} dA$. 


Let $\mathbb{E}[\delta_S]$ denote the average loss of the first-order scattering. It is reported in \cite{TWC-localization} that $\mathbb{E}[\delta_S]$ is around $-10$ dB in practice. Then, the differential signal power scattered by the differential surface area is expressed as
$dP_S(\mathbf{\Psi},r)=\frac{P_R\left(\mathbf{\Psi}\right)}{4\pi r^2} \mathbb{E}[\delta_S] dA$,
according to which the corresponding differential echo power receipted at the BS is derived {\color{black}as
\small \begin{align}\label{dP_rx}
dP_{rx}(\mathbf{\Psi},r)=\ & dP_S(\mathbf{\Psi},r) \left(\frac{\lambda}{4\pi r}\right)^2 
\left|   \mathbf{w}_{rx}^\mathrm{H} \mathbf{H}_{BR}^\mathrm{H} \mathbf{\Omega}(\mathbf{\Psi},\mathbf{\Psi}_{R}) \mathbf{a}(\theta,\varphi) \right|^2 \nonumber \\
=\ & \frac{dP_S(\mathbf{\Psi},r)}{4\pi r^2}
| g_{uc}(\mathbf{\Psi},\mathbf{\Psi}_R) |^2 |\mathbf{w}_{rx}^\mathrm{H} \mathbf{H}_{BR}^\mathrm{H} \mathrm{diag}(\bm{\omega}) \mathbf{a}(\theta,\varphi) |^2 \nonumber \\
=\ &  \frac{ \mathbb{E}\{\delta_S\}}{4\pi r^4 \lambda^2 }
| g_{uc}(\mathbf{\Psi}_R,\mathbf{\Psi}) |^2 
| g_{uc}(\mathbf{\Psi},\mathbf{\Psi}_R) |^2 \times
\mathbf{a}^{\mathrm{H}}(\theta,\varphi) \mathrm{diag}(\bm{\omega})\mathbf{H}_{BR}(\mathbf{w}_{c}\mathbf{w}_{c}^\mathrm{H} + \mathbf{w}_{s}\mathbf{w}_{s}^\mathrm{H}) \mathbf{H}_{BR}^\mathrm{H} \mathrm{diag}(\bm{\omega}^\mathrm{H}) \mathbf{a}(\theta,\varphi) \nonumber\\
& \times \mathbf{a}^{\mathrm{H}}(\theta,\varphi) \mathrm{diag}(\bm{\omega}^{\mathrm{H}})\mathbf{H}_{BR} \mathbf{w}_{rx} \mathbf{w}_{rx}^{\mathrm{H}} \mathbf{H}_{BR}^{\mathrm{H}} \mathrm{diag}(\bm{\omega}) \mathbf{a}(\theta,\varphi) \ 
 dA,
\end{align}\normalsize
where} $\mathbf{w}_{rx}$ is the unit-norm receive combining vector at the BS satisfying $\|\mathbf{w}_{rx}\|_2=1$.

Based on the differential echo power in (\ref{dP_rx}), the total echo power scattered from $S(\Delta_\theta,\Delta_\varphi)$ and eventually harvested at the BS can be integrally calculated {\color{black}as
\small \begin{align}\label{P_rx_detail}
P_{rx}(\mathbf{\Psi}_S,r,\Delta_\theta,\Delta_\varphi) = & \oiint \limits_{S(\Delta_\theta,\Delta_\varphi)} dP_{rx}(\mathbf{\Psi},r) \nonumber\\
\overset{(a)}{=}& \int_{\varphi_S-\frac{\Delta_{\varphi}}{2}}^{\varphi_S+\frac{\Delta_{\varphi}}{2}} \int_{\theta_S-\frac{\Delta_{\theta}}{2}}^{\theta_S+\frac{\Delta_{\theta}}{2}}
\left\{ \frac{ \mathbb{E}\{\delta_S\}}{4\pi r^2 \lambda^2 }
| g_{uc}(\mathbf{\Psi}_R,\mathbf{\Psi}) |^2 
| g_{uc}(\mathbf{\Psi},\mathbf{\Psi}_R) |^2 \right. \nonumber\\
&\ \ \ \ \ \ \ \ \ \ \ \  \times
\mathbf{a}^{\mathrm{H}}(\theta,\varphi) \mathrm{diag}(\bm{\omega})\mathbf{H}_{BR}(\mathbf{w}_{c}\mathbf{w}_{c}^\mathrm{H} + \mathbf{w}_{s}\mathbf{w}_{s}^\mathrm{H}) \mathbf{H}_{BR}^\mathrm{H} \mathrm{diag}(\bm{\omega}^\mathrm{H}) \mathbf{a}(\theta,\varphi) \nonumber\\
&\ \ \ \ \ \ \ \ \ \ \ \  \times \mathbf{a}^{\mathrm{H}}(\theta,\varphi) \mathrm{diag}(\bm{\omega}^{\mathrm{H}})\mathbf{H}_{BR} \mathbf{w}_{rx} \mathbf{w}_{rx}^{\mathrm{H}} \mathbf{H}_{BR}^{\mathrm{H}} \mathrm{diag}(\bm{\omega}) \mathbf{a}(\theta,\varphi)  \sin{\theta} \} \ d\theta\ d\varphi,
\end{align}\normalsize
where} derivation $(a)$ uses the geometric relation of $dA=r^2 \sin{\theta}\ d\theta\ d\varphi$.

To make a detection at the BS, the received echo is sampled over the duration of $T_0$ with a sampling frequency of $f_s$, and is dedicated to a hypothesis test.
Then, in accordance with the Neyman-Pearson criterion, the detection probability is expressed as [49, Eq. (3.8)]
 \begin{equation}\label{Pd_metric}
\mathbb{P}_d (\mathbf{\Psi}_S,r,\Delta_\theta,\Delta_\varphi)   
= Q \left( Q^{-1}(\mathbb{P}_f) - \frac{\sqrt{T_0 f_s P_{rx}(\mathbf{\Psi}_S,r,\Delta_\theta,\Delta_\varphi) }}{\sigma_n} \right),
\end{equation}\normalsize
where {\color{black} $Q(\cdot)$ and $Q^{-1}(\cdot)$ are the Q-function and the inverse Q-function, respectively, given by $Q(x)=\frac{1}{\sqrt{2\pi}}\int_{x}^{\infty} \mathrm{exp}\left(-\frac{u^2}{2}\right)du$ \cite{Q-func} and $Q^{-1}(x)=\sqrt{2} \mathrm{erf}^{-1}(1-2x)$, with $\mathrm{erf}^{-1}(.)$ being the inverse error function}; 
$\mathbb{P}_f= Q\left(\frac{\eta}{\sqrt{ \sigma_n^2/(T_0 f_s) }}\right)$ represents the false alarm rate, with $\eta$ being the decision threshold.  

{\color{black}In accordance with (\ref{Pd_metric}), we note that $\mathbb{P}_d (\mathbf{\Psi}_S,r,\Delta_\theta,\Delta_\varphi)$ decreases if $r$ increases, implying that the detection performance will be degraded when the target moves away from the RIS. Thus, without loss of generality, we consider the worst-case detection probability with respect to $r=r_{\mathrm{max}}$, i.e. $\mathbb{P}_d (\mathbf{\Psi}_S,r_{\mathrm{max}},\Delta_\theta,\Delta_\varphi)$, as the sensing performance metric to be optimized.  Moreover,}
we note that under a certain system setup (i.e. given $P_{tx}$, $\mathbf{H}_{BR}$, etc.), the detection probability is positively related to $S(\Delta_\theta,\Delta_\varphi)$ associated with the target size, owing to the integral operator $ \oiint \limits_{S(\Delta_\theta,\Delta_\varphi)} (\cdot)$. If the detection probability is required to satisfy $\mathbb{P}_d (\mathbf{\Psi}_S,r_{\mathrm{max}},\Delta_\theta,\Delta_\varphi)\geq \gamma_{\mathbb{P}_d^{(\mathrm{min})}}$ with $\gamma_{\mathbb{P}_d^{(\mathrm{min})}}\in(0,1)$ being a constant, the target size should be larger than a certain minimum. 
Based on the above analysis, we provide the following definition.
\begin{definition}
Under a certain system setup, a target at the distance of $r_{\mathrm{max}}$ in the sensing direction of $\mathbf{\Psi}_S=(\theta_S,\varphi_S)$ is called to be “$\gamma_{\mathbb{P}_d^{(\mathrm{min})}}$-detectable”, if its scattering surface area, i.e. $S(\Delta_\theta,\Delta_\varphi)$, can make $\mathbb{P}_d (\mathbf{\Psi}_S,r_{\mathrm{max}},\Delta_\theta,\Delta_\varphi)\geq \gamma_{\mathbb{P}_d^{(\mathrm{min})}}$ hold. Then, the scattering surface area of the smallest $\gamma_{\mathbb{P}_d^{(\mathrm{min})}}$-detectable target is defined as the {ultimate detection resolution (UDR)}, mathematically given by
 \begin{equation}\label{Detectable_Resolution}{\color{black}
\mathrm{UDR}(\Delta_\theta,\Delta_\varphi,\gamma_{\mathbb{P}_d^{(\mathrm{min})}}) 
\overset{\triangle}{=}  
\min_{\Delta_\theta,\Delta_\varphi} \{S(\Delta_\theta,\Delta_\varphi)\}, \ \mathrm{subject\ to} \ \mathbb{P}_d (\mathbf{\Psi}_S,r_{\mathrm{max}},\Delta_\theta,\Delta_\varphi)\geq \gamma_{\mathbb{P}_d^{(\mathrm{min})}},
}\end{equation}\normalsize
{\color{black}which represents the minimum $S(\Delta_\theta,\Delta_\varphi)$ with respect to $\Delta_\theta$ and $\Delta_\varphi$ that makes $\mathbb{P}_d (\mathbf{\Psi}_S,r_{\mathrm{max}},\Delta_\theta,\Delta_\varphi)\geq \gamma_{\mathbb{P}_d^{(\mathrm{min})}}$ hold.}
\end{definition}
The UDR defined in (\ref{Detectable_Resolution}) characterizes the capability of the target detection from the perspective of the size of the detectable target.
In Section III, we will show that the UDR can be explicitly determined via the feasibility condition of the optimization problem.





\subsection{Problem Statement}

Based on $\mathrm{SNR_{UE}}(\mathbf{w}_c,\mathbf{w}_s,\bm{\omega})$ and $\mathbb{P}_d (\mathbf{\Psi}_S,r_\mathrm{max},\Delta_\theta,\Delta_\varphi)$, our objective is to improve the UE communication performance as much as possible, while ensuring the minimum detection probability requirement.  The overall optimization problem is then formulated {\color{black}as
 \begin{subequations}\label{Overall_problem} 
\begin{align}
(\mathrm{P1}):\ \ \mathop{\mathrm{maximize}}\limits_{\mathbf{w}_{c}, \mathbf{w}_{s},\mathbf{w}_{rx},\bm{\omega}}\ \  & \mathrm{SNR_{UE}}(\mathbf{w}_c,\mathbf{w}_s,\bm{\omega}), \\
\mathrm{subject\ to}\ \  & |[\bm{\omega}]_\ell|=1,  \ell=1,2,...,N, \\
&  tr(\mathbf{w}_{c}\mathbf{w}_{c}^\mathrm{H}) + tr(\mathbf{w}_{s}\mathbf{w}_{s}^\mathrm{H}) \leq P_{tx}, \\
& \|\mathbf{w}_{rx}\|_2=1,\\
& \mathbb{P}_d (\mathbf{\Psi}_S,r_\mathrm{max},\Delta_\theta,\Delta_\varphi) \geq \gamma_{\mathbb{P}_d^{(\mathrm{min})}},
\end{align}
\end{subequations}\normalsize
where constraint (\ref{Overall_problem}b) comes from the unit-modulus property of the adjustable phase-shift variables; constraint (\ref{Overall_problem}c) is the total transmit power limit; constraint (\ref{Overall_problem}d) comes from the unit-norm receive combining vector at the BS; constraint (\ref{Overall_problem}e) means that the detection probability is guaranteed to be not lower than a threshold $\gamma_{\mathbb{P}_d^{(\mathrm{min})}}$.
As $(\mathrm{P1})$ is non-convex, we will propose a novel alternative optimization approach to solve $(\mathrm{P1})$ in the next section.}

\section{Joint Active and Passive Beamforming Optimization}

This section is dedicated to solving Problem $(\mathrm{P1})$ by {\color{black} developing a novel alternative optimization approach to optimize $\mathbf{w}_{c}$, $\mathbf{w}_{s}$, $\mathbf{w}_{rx}$ and $\bm{\omega}$. First, with fixed $\mathbf{w}_{rx}$ and $\bm{\omega}$, we optimize $\mathbf{w}_{c}$ and $\mathbf{w}_{s}$ by proposing a bisection-search based method. Then, with fixed $\mathbf{w}_{c}$, $\mathbf{w}_{s}$ and $\bm{\omega}$, we optimize $\mathbf{w}_{rx}$ by transforming the original subproblem into an equivalent Rayleigh-quotient problem. Subsequently, when optimizing $\bm{\omega}$ with fixed $\mathbf{w}_{c}$, $\mathbf{w}_{s}$ and $\mathbf{w}_{rx}$, we adopt the Charnes-Cooper transformation to deal with the fractional objective function, and propose a novel convexification process as well as an SCA-based algorithm to iteratively find the solution. Finally, we design the overall algorithm, and present essential analyses on the computational complexity, the algorithm convergence and the problem feasibility condition.

\subsection{Optimization of $\mathbf{w}_{c}$, $\mathbf{w}_{s}$ with Fixed $\mathbf{w}_{rx}$, $\bm{\omega}$}}

{\color{black}
Considering $\mathbf{w}_{rx}$ and $\bm{\omega}$ to be given, we optimize $\mathbf{w}_{c}$ and $\mathbf{w}_{s}$ by solving the following subproblem:
 \begin{subequations}\label{P2} 
\begin{align}
(\mathrm{P2}):\ \ \mathop{\mathrm{maximize}}\limits_{\mathbf{w}_{c}, \mathbf{w}_{s}}\ \  & \mathrm{SNR_{UE}}(\mathbf{w}_c,\mathbf{w}_s,\bm{\omega}), \\
\mathrm{subject\ to}\ \  & \mathrm{ (\ref{Overall_problem}c), (\ref{Overall_problem}e)}.
\end{align}
\end{subequations}\normalsize
After defining $\mathbf{W}_c=\mathbf{w}_c\mathbf{w}_c^\mathrm{H}$, $\mathbf{W}_s=\mathbf{w}_s\mathbf{w}_s^\mathrm{H}$, 
$\mathbf{\Xi} = [\mathbf{H}_{BR}^\mathrm{H} \mathbf{\Omega}^\mathrm{H}(\mathbf{\Psi}_{R},\mathbf{\Psi}_{U}) \mathbf{h}_{RU} + \mathbf{h}_{BU}] [\mathbf{h}_{RU}^{\mathrm{H}}\mathbf{\Omega}(\mathbf{\Psi}_{R},\mathbf{\Psi}_{U})\mathbf{H}_{BR} + \mathbf{h}_{BU}^{\mathrm{H}}]$, and 
 \begin{align}
\mathbf{M}_r=\int_{\varphi_S-\frac{\Delta_{\varphi}}{2}}^{\varphi_S+\frac{\Delta_{\varphi}}{2}} \int_{\theta_S-\frac{\Delta_{\theta}}{2}}^{\theta_S+\frac{\Delta_{\theta}}{2}} &
\{ | g_{uc}(\mathbf{\Psi}_R,\mathbf{\Psi}) |^2 
| g_{uc}(\mathbf{\Psi},\mathbf{\Psi}_R) |^2 
 \mathbf{H}_{BR}^\mathrm{H} \mathrm{diag}(\bm{\omega}^\mathrm{H}) \mathbf{a}(\theta,\varphi) 
 \mathbf{a}^{\mathrm{H}}(\theta,\varphi) \mathrm{diag}(\bm{\omega}^{\mathrm{H}})\mathbf{H}_{BR} \mathbf{w}_{rx} \nonumber\\
&\times \mathbf{w}_{rx}^{\mathrm{H}} \mathbf{H}_{BR}^{\mathrm{H}} \mathrm{diag}(\bm{\omega}) \mathbf{a}(\theta,\varphi)  \mathbf{a}^{\mathrm{H}}(\theta,\varphi) \mathrm{diag}(\bm{\omega})\mathbf{H}_{BR}  \sin{\theta}\}\ d\theta\  d\varphi, \nonumber
 \end{align}\normalsize
the above $\mathrm{(P2)}$ can be recast as
 \begin{subequations}\label{P3} 
\begin{align}
(\mathrm{P3}):\ \ \mathop{\mathrm{maximize}}\limits_{\mathbf{W}_{c}\succeq\mathbf{0}, \mathbf{W}_{s}\succeq\mathbf{0}}\ \  & \frac{tr(\mathbf{\Xi}\mathbf{W}_c)}{tr(\mathbf{\Xi}\mathbf{W}_s)+\sigma_{n,u}^2}, \\
\mathrm{subject\ to}\ \  & tr(\mathbf{W}_c)+tr(\mathbf{W}_s) \leq P_{tx},\\
& tr(\mathbf{M}_r (\mathbf{W}_c + \mathbf{W}_s)) \geq \mathcal{G},\\
&\mathrm{rank}(\mathbf{W}_c)=1, \mathrm{rank}(\mathbf{W}_s)=1,
\end{align}
\end{subequations}\normalsize
where $ \mathcal{G}$ is a constant given by
 \begin{equation}\label{mathcal_G}
\mathcal{G} = \frac{4\pi r_{\mathrm{max}}^2 \lambda^2 \sigma_n^2}{\mathbb{E}[\delta_S] T_0 f_s} \left(Q^{-1}(\mathbb{P}_f) - Q^{-1}(\gamma_{\mathbb{P}_d^{(\mathrm{min})}}) \right)^{2}.
\end{equation}\normalsize

To solve $(\mathrm{P3})$, we introduce an auxiliary variable $t\geq 0$ and transform $(\mathrm{P3})$ into 
 \begin{subequations}\label{P4} 
\begin{align}
(\mathrm{P4}):\ \ \mathop{\mathrm{maximize}}\limits_{\mathbf{W}_{c}\succeq\mathbf{0}, \mathbf{W}_{s}\succeq\mathbf{0},t\geq 0}\ \  & t, \\
\mathrm{subject\ to}\ \  & \frac{tr(\mathbf{\Xi}\mathbf{W}_c)}{tr(\mathbf{\Xi}\mathbf{W}_s)+\sigma_{n,u}^2}\geq t,\\
& \mathrm{(\ref{P3}b),(\ref{P3}c),(\ref{P3}d)}.
\end{align}
\end{subequations}\normalsize
Since $\mathbf{W}_{s}$ and $t$ are coupled in (\ref{P4}b) making $(\mathrm{P4})$ non-convex, here we propose a new bisection-search based method to yield the solution. Concretely, we first treat $t$ as a constant and consider the following feasibility check problem, described as
 \begin{subequations}\label{P5} 
\begin{align}
(\mathrm{P5}):\ \ \mathop{\mathrm{maximize}}\limits_{\mathbf{W}_{c}\succeq\mathbf{0}, \mathbf{W}_{s}\succeq\mathbf{0}}\ \  & t, \\
\mathrm{subject\ to}\ \  
& \mathrm{(\ref{P4}b),(\ref{P3}b),(\ref{P3}c),(\ref{P3}d)}.
\end{align}
\end{subequations} \normalsize
To decide the maximum of $t$, we choose a certain $t=t_{\mathrm{min}}$ such that $(\mathrm{P5})$ has feasible solutions, and choose a certain $t=t_{\mathrm{max}}>t_{\mathrm{min}}$ that makes $(\mathrm{P5})$ infeasible. Then, we
successively branch on the initial search interval of $t\in [t_{\mathrm{min}},t_{\mathrm{max}}]$, and examine the feasibility of $(\mathrm{P5})$ by using the SDR to solve it after each branch. This procedure repeats, until the length of the search interval is cut to be smaller than a predetermined value. The details are summarized in Algorithm 1.
\begin{algorithm}
\color{black}
\caption{The proposed bisection-search based algorithm for solving $(\mathrm{P4})$.}
\LinesNumbered
{\bf Initialization:} $t_{\mathrm{min}}$, $t_{\mathrm{max}}$ ($t_{\mathrm{max}}>t_{\mathrm{min}}>0$), and $\epsilon_1>0$\;

{\bf Repeat:} \\
\ \ \ 1) Compute $t_{\mathrm{new}}=\frac{t_{\mathrm{min}}+t_{\mathrm{max}}}{2}$\;
\ \ \ 2) Use SDR to solve $(\mathrm{P5})$ with $t=t_{\mathrm{new}}$ by dropping the rank-one constraint (\ref{P3}d)\;
\ \ \ 3) {\bf If} Problem $(\mathrm{P5})$ with $t=t_{\mathrm{new}}$ has feasible solutions {\bf then} \\ \ \ \ \ \ \ \ \ \ $t_{\mathrm{min}}\leftarrow t_{\mathrm{new}}$,  and record the optimal solutions of $\mathbf{W}_c$ and $\mathbf{W}_s$ as $\mathbf{W}_c^{(\star)}$ and $\mathbf{W}_s^{(\star)}$; \\
\ \ \ \ \ \ {\bf Else} $t_{\mathrm{max}}\leftarrow t_{\mathrm{new}}$\;
{\bf Until:} $t_{\mathrm{max}} - t_{\mathrm{min}} \leq \epsilon_1$\;


Perform eigenvalue decomposition for $\mathbf{W}_c^{(\star)}$ and $\mathbf{W}_s^{(\star)}$ to obtain $\mathbf{w}_c^{(\star)}$ and $\mathbf{w}_s^{(\star)}$\;

{\bf Output:} $\mathbf{w}_c^{(\star)}$ and $\mathbf{w}_s^{(\star)}$\;

\end{algorithm}

\begin{remark}
Let $k=2$ denote the number of optimization variables in $(\mathrm{P5})$, and $m=3$ denote the number of constraints except for the rank-one constraint (\ref{P3}d). In accordance with \cite{SDR-SPM}, because $m\leq k+1$ holds, the SDR of $(\mathrm{P5})$ is tight and there exist rank-one solutions $\mathbf{W}_c^{(\star)}$ and $\mathbf{W}_s^{(\star)}$ that can be efficiently found. Then, the optimal $\mathbf{w}_c$ and $\mathbf{w}_s$, denoted by $\mathbf{w}_c^{(\star)}$ and $\mathbf{w}_s^{(\star)}$, can be obtained through the eigenvalue decompositions of $\mathbf{W}_c^{(\star)}$ and $\mathbf{W}_s^{(\star)}$.

\end{remark}

\subsection{Optimization of $\mathbf{w}_{rx}$ with Fixed $\mathbf{w}_{c}$, $\mathbf{w}_{s}$ and $\bm{\omega}$}

Then, considering $\mathbf{w}_{c}$, $\mathbf{w}_{s}$ and $\bm{\omega}$ to be given, we optimize the receive combining vector $\mathbf{w}_{rx}$. As $\mathbf{w}_{rx}$ only appears in (\ref{Overall_problem}d) and (\ref{Overall_problem}e) to maintain the required target sensing performance, we can update $\mathbf{w}_{rx}$ by maximizing the detection probability via
 \begin{subequations}\label{P6} 
\begin{align}
(\mathrm{P6}):\ \ \mathop{\mathrm{maximize}}\limits_{\mathbf{w}_{rx}}\ \  & \mathbb{P}_d (\mathbf{\Psi}_S,r_\mathrm{max},\Delta_\theta,\Delta_\varphi), \\
\mathrm{subject\ to}\ \  & \| \mathbf{w}_{rx} \|_2=1.
\end{align}
\end{subequations}\normalsize
According to (\ref{Pd_metric}), the maximization of $\mathbb{P}_d (\mathbf{\Psi}_S,r_\mathrm{max},\Delta_\theta,\Delta_\varphi)$ is equivalent to the maximization of $P_{rx}(\mathbf{\Psi}_S,r_\mathrm{max},\Delta_\theta,\Delta_\varphi)$ in (\ref{P_rx_detail}). Hence, by defining 
\small \begin{align}
\mathbf{M}_{cs} = \int_{\varphi_S-\frac{\Delta_{\varphi}}{2}}^{\varphi_S+\frac{\Delta_{\varphi}}{2}} \int_{\theta_S-\frac{\Delta_{\theta}}{2}}^{\theta_S+\frac{\Delta_{\theta}}{2}} 
& \left\{ \frac{ \mathbb{E}\{\delta_S\}}{4\pi r_{\mathrm{max}}^2 \lambda^2 }   | g_{uc}(\mathbf{\Psi}_R,\mathbf{\Psi}) |^2 
| g_{uc}(\mathbf{\Psi},\mathbf{\Psi}_R) |^2 
 \mathbf{H}_{BR}^\mathrm{H} \mathrm{diag}(\bm{\omega}) \mathbf{a}(\theta,\varphi) 
 \mathbf{a}^{\mathrm{H}}(\theta,\varphi) \mathrm{diag}(\bm{\omega})\mathbf{H}_{BR} ( \mathbf{w}_{c} \mathbf{w}_{c}^{\mathrm{H}} \right. 
\nonumber\\ 
& + \mathbf{w}_{s} \mathbf{w}_{s}^{\mathrm{H}} ) \mathbf{H}_{BR}^{\mathrm{H}} \mathrm{diag}(\bm{\omega}^{\mathrm{H}}) \mathbf{a}(\theta,\varphi)  \mathbf{a}^{\mathrm{H}}(\theta,\varphi) \mathrm{diag}(\bm{\omega}^{\mathrm{H}})\mathbf{H}_{BR}  \sin{\theta}\} \ d\theta \ d\varphi,
 \end{align}\normalsize
 the above $(\mathrm{P6})$ can be recast as
 \begin{subequations}\label{P7} 
\begin{align}
(\mathrm{P7}):\ \ \mathop{\mathrm{maximize}}\limits_{\mathbf{w}_{rx}}\ \  & \mathbf{w}_{rx}^\mathrm{H} \mathbf{M}_{cs} \mathbf{w}_{rx}, \\
\mathrm{subject\ to}\ \  & \| \mathbf{w}_{rx} \|_2=1.
\end{align}
\end{subequations}\normalsize

Note that when constrained by (\ref{P7}b), the objective in (\ref{P7}a) is equivalent to $\mathop{\mathrm{maximize}}\limits_{\mathbf{w}_{rx}} \frac{\mathbf{w}_{rx}^\mathrm{H} \mathbf{M}_{cs} \mathbf{w}_{rx}} {\mathbf{w}_{rx}^\mathrm{H} \mathbf{w}_{rx}}$ since $\mathbf{w}_{rx}^\mathrm{H} \mathbf{w}_{rx}=1$ holds, which implies that Problem $(\mathrm{P7})$ is a typical Rayleigh-quotient problem. As a result, the solution of $(\mathrm{P7})$ can be readily attained as
 \begin{equation}\label{wrx_opt}
\mathbf{w}_{rx} = \mathbf{w}_{rx}^{(\star)} =  \mathbf{e}_{\mathrm{max}}(\mathbf{M}_{cs}),
\end{equation}\normalsize
where $\mathbf{e}_{\mathrm{max}}(\mathbf{M}_{cs})$ denotes the eigenvector of $\mathbf{M}_{cs}$ corresponding to the largest eigenvalue of $\mathbf{M}_{cs}$.

}

\subsection{Optimization of $\bm{\omega}$ with Fixed $\mathbf{w}_{c}$, $\mathbf{w}_{s}$ and $\mathbf{w}_{rx}$}

Since the optimal $\mathbf{w}_{c}$, $\mathbf{w}_{s}$ and $\mathbf{w}_{rx}$ have been settled, we are now ready to optimize $\bm{\omega}$. With given $\mathbf{w}_{c}$, $\mathbf{w}_{s}$ and $\mathbf{w}_{rx}$, we begin our analysis from the following subproblem with respect to $\bm{\omega}$, given by
 \begin{subequations}\label{P8}
\begin{align}
(\mathrm{P8}):\ \ \mathop{\mathrm{maximize}}\limits_{\bm{\omega}}\ \  & \mathrm{SNR_{UE}}(\mathbf{w}_c,\mathbf{w}_s,\bm{\omega}), \\
\mathrm{subject\ to}\ \  & \mathrm{(\ref{Overall_problem}b), (\ref{Overall_problem}e)}.
\end{align}
\end{subequations}\normalsize
{\color{black}As $(\mathrm{P8})$ is complicatedly intractable due to the fractional objective function and the non-convex constraints, we first use the Charnes-Cooper transformation to convert the objective into a linear function}. Then, we propose to convexify constraint (\ref{Overall_problem}e) into a linear constraint with matrix operations and a real-valued first-order Taylor {\color{black}expansion}. Finally, we use the SDR to relax the problem and design an SCA-based algorithm to acquire a local optimal phase-shift solution. 

{\color{black}

\subsubsection{\textbf{Charnes-Cooper Transformation for $(\mathrm{P8})$}} 

To deal with the fractional objective in $(\mathrm{P8})$, we first define $\mathbf{q}=[\bm{\omega}^\mathrm{T},1]^\mathrm{T}$ and $\mathbf{Q}=\mathbf{q}\mathbf{q}^\mathrm{H}$, such that the objective function can be rewritten as
 \begin{equation}
\mathrm{SNR_{UE}} = \frac{\mathbf{q}^\mathrm{H} \mathbf{Z}_c \mathbf{q}}{\mathbf{q}^\mathrm{H} \mathbf{Z}_s \mathbf{q} + \sigma_{n,u}^2} = \frac{tr(\mathbf{Z}_c \mathbf{Q})}{tr(\mathbf{Z}_s \mathbf{Q})+ \sigma_{n,u}^2},
\end{equation}\normalsize
where $\mathbf{Z}_c$ and $\mathbf{Z}_s$ are, respectively, expressed as
\small \begin{equation}\nonumber
\mathbf{Z}_c = 
\left[\begin{matrix}
\frac{4\pi}{\lambda^2} |g_{uc}(\mathbf{\Psi}_{R},\mathbf{\Psi}_{U})|^2 \mathrm{diag}(\mathbf{w}_c^\mathrm{H} \mathbf{H}_{BR}^\mathrm{H}) \mathbf{h}_{RU}\mathbf{h}_{RU}^\mathrm{H} \mathrm{diag}( \mathbf{H}_{BR} \mathbf{w}_c),
& \frac{\sqrt{4\pi}}{\lambda} g_{uc}^*(\mathbf{\Psi}_{R},\mathbf{\Psi}_{U}) \mathrm{diag}(\mathbf{w}_c^\mathrm{H} \mathbf{H}_{BR}^\mathrm{H}) \mathbf{h}_{RU}\mathbf{h}_{BU}^\mathrm{H} \mathbf{w}_c \\
\frac{\sqrt{4\pi}}{\lambda} g_{uc}(\mathbf{\Psi}_{R},\mathbf{\Psi}_{U}) \mathbf{w}_c^\mathrm{H} \mathbf{h}_{BU} \mathbf{h}_{RU}^\mathrm{H} \mathrm{diag}( \mathbf{H}_{BR} \mathbf{w}_c),
&  \mathbf{w}_c^\mathrm{H} \mathbf{h}_{BU} \mathbf{h}_{BU}^\mathrm{H} \mathbf{w}_c
\end{matrix}
\right],
\end{equation}\normalsize
\small \begin{equation}\nonumber
\mathbf{Z}_s = 
\left[\begin{matrix}
\frac{4\pi}{\lambda^2} |g_{uc}(\mathbf{\Psi}_{R},\mathbf{\Psi}_{U})|^2 \mathrm{diag}(\mathbf{w}_s^\mathrm{H} \mathbf{H}_{BR}^\mathrm{H}) \mathbf{h}_{RU}\mathbf{h}_{RU}^\mathrm{H} \mathrm{diag}( \mathbf{H}_{BR} \mathbf{w}_s),
& \frac{\sqrt{4\pi}}{\lambda} g_{uc}^*(\mathbf{\Psi}_{R},\mathbf{\Psi}_{U}) \mathrm{diag}(\mathbf{w}_s^\mathrm{H} \mathbf{H}_{BR}^\mathrm{H}) \mathbf{h}_{RU}\mathbf{h}_{BU}^\mathrm{H} \mathbf{w}_s \\
\frac{\sqrt{4\pi}}{\lambda} g_{uc}(\mathbf{\Psi}_{R},\mathbf{\Psi}_{U}) \mathbf{w}_s^\mathrm{H} \mathbf{h}_{BU} \mathbf{h}_{RU}^\mathrm{H} \mathrm{diag}( \mathbf{H}_{BR} \mathbf{w}_s),
&  \mathbf{w}_s^\mathrm{H} \mathbf{h}_{BU} \mathbf{h}_{BU}^\mathrm{H} \mathbf{w}_s
\end{matrix}
\right].
\end{equation}\normalsize
Then, Problem $\mathrm{(P8)}$ is transformed into
 \begin{subequations}\label{P9}
\begin{align}
(\mathrm{P9}):\ \ \mathop{\mathrm{maximize}}\limits_{\mathbf{Q}\succeq \mathbf{0}}\ \  & \frac{tr(\mathbf{Z}_c \mathbf{Q})}{tr(\mathbf{Z}_s \mathbf{Q})+ \sigma_{n,u}^2}, \\
\mathrm{subject\ to}\ \  & tr\left(\mathbf{E}_\ell \mathbf{Q}\right)=1,  \ell=1,2,...,N+1,  \\ 
& \mathrm{rank}(\mathbf{Q})=1 ,\\
& \mathrm{(\ref{Overall_problem}e)},
\end{align}
\end{subequations}\normalsize
where $\mathbf{E}_\ell$ is a selecting matrix with each element satisfying
 \begin{equation}
[\mathbf{E}_\ell]_{(m,n)} = 
\left\{
\begin{matrix}
1,\ \ m=n=\ell,\\
0,\ \ \ \mathrm{otherwise}.
\end{matrix}
\right.
\end{equation}\normalsize

As the objective function in $(\mathrm{P9})$ is non-convex with respect to $\mathbf{Q}$, we apply the Charnes-Cooper transformation to acquire an equivalent linear expression. In specific, we introduce $\mu=\frac{1}{tr(\mathbf{Z}_s \mathbf{Q})+ \sigma_{n,u}^2}$ and $\mathbf{X}=\mu\mathbf{Q}$, so as to transform $(\mathrm{P9})$ into
 \begin{subequations}\label{P10}
\begin{align}
(\mathrm{P10}):\ \ \mathop{\mathrm{maximize}}\limits_{\mathbf{X}\succeq \mathbf{0},\mu>0}\ \  & tr(\mathbf{Z}_c \mathbf{X}), \\
\mathrm{subject\ to}\ \  & tr\left(\mathbf{E}_\ell \mathbf{X}\right)=\mu,  \ell=1,2,...,N+1,  \\ 
& \mathrm{rank}(\mathbf{X})=1 ,\\
& tr(\mathbf{Z}_s \mathbf{X})+ \mu \sigma_{n,u}^2 = 1, \\
& \mathrm{(\ref{Overall_problem}e)}.
\end{align}
\end{subequations}\normalsize
In Problem $(\mathrm{P10})$, the constraints and the objective function are convex except for (\ref{P10}c) and (\ref{Overall_problem}e). Thereupon, we will further focus on the convexification of (\ref{Overall_problem}e) to make $(\mathrm{P10})$ more tractable.

}

\subsubsection{\textbf{Convexification of Constraint (\ref{Overall_problem}e)}}

To convexify (\ref{Overall_problem}e), we need to first convert (\ref{Overall_problem}e) into a constraint with respect to variable $\mathbf{X}$. To begin with, we rewrite $\mathbb{P}_d (\mathbf{\Psi}_S,r_{\mathrm{max}},\Delta_\theta,\Delta_\varphi) \geq \gamma_{\mathbb{P}_d^{(\mathrm{min})}}$ as
 \begin{equation}
Q \left( Q^{-1}(\mathbb{P}_f) - \frac{\sqrt{T_0 f_s P_{rx}(\mathbf{\Psi}_S,r_{\mathrm{max}},\Delta_\theta,\Delta_\varphi) }}{\sigma_n} \right) \geq \gamma_{\mathbb{P}_d^{(\mathrm{min})}},
\end{equation}\normalsize
yielding
 \begin{equation}\label{P_rx_inequal}
P_{rx}(\mathbf{\Psi}_S,r_{\mathrm{max}},\Delta_\theta,\Delta_\varphi)
\geq \frac{\sigma_n^2}{T_0 f_s} \left(   Q^{-1}(\mathbb{P}_f)  -  Q^{-1} (\gamma_{\mathbb{P}_d^{(\mathrm{min})}}) \right)^2.
\end{equation}\normalsize
Then, by defining {\color{black} a function $f(\bm{\omega})$ in relation to $\bm{\omega}$ as
\small \begin{align}\nonumber
f(\bm{\omega})=\int_{\varphi_S-\frac{\Delta_{\varphi}}{2}}^{\varphi_S+\frac{\Delta_{\varphi}}{2}} \int_{\theta_S-\frac{\Delta_{\theta}}{2}}^{\theta_S+\frac{\Delta_{\theta}}{2}}
& \left\{ 
| g_{uc}(\mathbf{\Psi}_R,\mathbf{\Psi}) |^2 
| g_{uc}(\mathbf{\Psi},\mathbf{\Psi}_R) |^2 \right. 
\mathbf{a}^{\mathrm{H}}(\theta,\varphi) \mathrm{diag}(\bm{\omega})\mathbf{H}_{BR}(\mathbf{w}_{c}\mathbf{w}_{c}^\mathrm{H} + \mathbf{w}_{s}\mathbf{w}_{s}^\mathrm{H}) \mathbf{H}_{BR}^\mathrm{H} \mathrm{diag}(\bm{\omega}^\mathrm{H}) \mathbf{a}(\theta,\varphi) \nonumber\\
&\  \times \mathbf{a}^{\mathrm{H}}(\theta,\varphi) \mathrm{diag}(\bm{\omega}^{\mathrm{H}})\mathbf{H}_{BR} \mathbf{w}_{rx} \mathbf{w}_{rx}^{\mathrm{H}} \mathbf{H}_{BR}^{\mathrm{H}} \mathrm{diag}(\bm{\omega}) \mathbf{a}(\theta,\varphi)  \sin{\theta} \}  d\theta\ d\varphi, \nonumber
\end{align}\normalsize
the} above (\ref{P_rx_inequal}) can be equivalently written as
 \begin{equation}
f(\bm{\omega}) \geq \mathcal{G},
\end{equation}\normalsize
where $ \mathcal{G}$ is given in (\ref{mathcal_G}).

To simplify the expression, we define two auxiliary matrices with respect to $\theta$ and $\varphi$ as {\color{black}
 \begin{equation}
\mathbf{A}(\theta,\varphi) = |g_{uc}(\mathbf{\Psi}_R,\mathbf{\Psi})|^2 \left[ \mathrm{diag}\{\mathbf{a}^{\mathrm{H}}(\theta,\varphi)\}\mathbf{H}_{BR}(\mathbf{w}_{c}  \mathbf{w}_{c}^\mathrm{H} + \mathbf{w}_{s}  \mathbf{w}_{s}^\mathrm{H})\mathbf{H}_{BR}^\mathrm{H} \mathrm{diag}\{\mathbf{a}(\theta,\varphi)\}  \right]^{\mathrm{T}},
\end{equation}
\begin{equation}
\mathbf{B}(\theta,\varphi) = | g_{uc}(\mathbf{\Psi},\mathbf{\Psi}_R)|^2
\mathrm{diag}\{\mathbf{a}^\mathrm{H}(\theta,\varphi)\} \mathbf{H}_{BR} \mathbf{w}_{rx}  \mathbf{w}_{rx}^\mathrm{H} \mathbf{H}_{BR}^\mathrm{H} \mathrm{diag}\{\mathbf{a}(\theta,\varphi)\} \sin{\theta}.
\end{equation}\normalsize
Then}, we have
 \begin{align}\label{Simplify-29}
f(\bm{\omega})
= & \int_{\varphi_S-\frac{\Delta_{\varphi}}{2}}^{\varphi_S+\frac{\Delta_{\varphi}}{2}} \int_{\theta_S-\frac{\Delta_{\theta}}{2}}^{\theta_S+\frac{\Delta_{\theta}}{2}}
\bm{\omega}^\mathrm{H} \mathbf{A}(\theta,\varphi) \bm{\omega}
\bm{\omega}^\mathrm{H} \mathbf{B}(\theta,\varphi) \bm{\omega}
\ d\theta\ d\varphi  \nonumber\\
\overset{(a)}{=} & \int_{\varphi_S-\frac{\Delta_{\varphi}}{2}}^{\varphi_S+\frac{\Delta_{\varphi}}{2}} \int_{\theta_S-\frac{\Delta_{\theta}}{2}}^{\theta_S+\frac{\Delta_{\theta}}{2}}
\mathbf{q}^\mathrm{H} \widetilde{\mathbf{A}}(\theta,\varphi) \mathbf{q}
\mathbf{q}^\mathrm{H} \widetilde{\mathbf{B}}(\theta,\varphi) \mathbf{q}
\ d\theta\ d\varphi \nonumber\\
 = & f(\mathbf{q}),
\end{align}\normalsize
where $\widetilde{\mathbf{A}}(\theta,\varphi)$ and $\widetilde{\mathbf{B}}(\theta,\varphi)$ are given by
 \begin{equation}
\widetilde{\mathbf{A}}(\theta,\varphi)=
\left[\begin{matrix} \mathbf{A}(\theta,\varphi) & \mathbf{0}_{N\times 1} \\ \mathbf{0}_{1\times N} & 0
\end{matrix} \right],\ 
\widetilde{\mathbf{B}}(\theta,\varphi)=
\left[\begin{matrix} \mathbf{B}(\theta,\varphi) & \mathbf{0}_{N\times 1} \\ \mathbf{0}_{1\times N} & 0
\end{matrix} \right],
\end{equation}\normalsize
and derivation $(a)$ uses the property of $\bm{\omega}^\mathrm{H} \mathbf{A}(\theta,\varphi) \bm{\omega}=\mathbf{q}^\mathrm{H} \widetilde{\mathbf{A}}(\theta,\varphi) \mathbf{q}$ and $\bm{\omega}^\mathrm{H} \mathbf{B}(\theta,\varphi) \bm{\omega}=\mathbf{q}^\mathrm{H} \widetilde{\mathbf{B}}(\theta,\varphi) \mathbf{q}$.

In (\ref{Simplify-29}), $f(\mathbf{q})$ is composed of a product of double quadratic forms involved in an integral, which is still highly non-convex. To tackle this issue, we derive the following proposition.

\begin{proposition}
By introducing $\mathbf{Q}=\mathbf{q}\mathbf{q}^\mathrm{H}$, $f(\mathbf{q})$ in (\ref{Simplify-29}) can be transformed into
 \begin{equation}\label{g_Q}
f(\mathbf{q}) = \|\mathbf{v}(\mathbf{Q})\|_2^2,
\end{equation}\normalsize
where $\mathbf{v}(\mathbf{Q})\in \mathbb{C}^{(N+1)^2 \times 1}$ is a vector with respect to $\mathbf{Q}$, given by 
 \begin{equation}
\mathbf{v}(\mathbf{Q}) = \left[tr(\mathbf{S}_1 \mathbf{Q}),tr(\mathbf{S}_2 \mathbf{Q}),\cdots,tr(\mathbf{S}_{(N+1)^2} \mathbf{Q})\right]^\mathrm{T},
\end{equation}\normalsize
where $\mathbf{S}_i$ for $i=1,2,\cdots,(N+1)^2$ are specified in (\ref{S_i}) in Appendix A, and are independent of $\mathbf{Q}$.

\end{proposition}
\begin{proof}
The proof is relegated to Appendix A.
\end{proof}

Based on Proposition 1, constraint (\ref{Overall_problem}e) can be equivalently written as
 \begin{equation}\label{g_cons1}
\|\mathbf{v}(\mathbf{Q})\|_2 \geq 
\sqrt{\mathcal{G} }.
\end{equation}\normalsize
By denoting $\mathbf{X}=\mu\mathbf{Q}$ as introduced previously, constraint (\ref{Overall_problem}e) can be further expressed as
 \begin{equation}\label{g_cons}
 \|\mathbf{v}(\mathbf{X})\|_2 - 
\mu \sqrt{\mathcal{G}  } \geq 0,
\end{equation}\normalsize
so that Problem $(\mathrm{P10})$ can be recast as
 \begin{subequations}\label{P11}
\begin{align}
(\mathrm{P11}):\ \ \mathop{\mathrm{maximize}}\limits_{\mathbf{X}\succeq \mathbf{0},\mu>0}\ \  & tr(\mathbf{Z}_c \mathbf{X}), \\
\mathrm{subject\ to}\ \  & \mathrm{(\ref{P10}b),(\ref{P10}c),(\ref{P10}d), (\ref{g_cons})}.
\end{align}
\end{subequations}\normalsize

\begin{remark}
Since $tr(\cdot)$ is linear and $\|\cdot\|_2$ is convex, $\|\mathbf{v}(\mathbf{X})\|_2$ is convex with respect to $\mathbf{X}$. 
However, (\ref{g_cons}) is still a non-convex constraint, due to the inequality of “convex function”$\geq$“constant”. 
\end{remark}

To further convexify (\ref{g_cons}), we next propose to linearise $\|\mathbf{v}(\mathbf{X})\|_2$ by a real-valued first-order {\color{black}Taylor expansion}. 
Let $\mathcal{T}_{\mathbf{z}_k}\left(h(\mathbf{z})\right)$ denote the first-order {\color{black}Taylor expansion} for a function $h(\mathbf{z})$ around a given point $\mathbf{z}_k$, expressed as
 \begin{equation}\label{FApprox}
\mathcal{T}_{\mathbf{z}_k}\left(h(\mathbf{z})\right)=h(\mathbf{z}_k) + \left\langle \left.\frac{\partial h(\mathbf{z})}{\partial \mathbf{z}}\right|_{\mathbf{z}=\mathbf{z}_k}, \left(\mathbf{z}-\mathbf{z}_k\right) \right\rangle,
\end{equation}\normalsize
where $\mathbf{z}_k$ can be, e.g. a local optimum obtained in the $k$-th iteration in a successive optimization process. Then, using the definition in (\ref{FApprox}), the first-order {\color{black}Taylor expansion} of $\|\mathbf{v}(\mathbf{X})\|_2$ around $\mathbf{v}(\mathbf{X}_k)$ can be calculated by
 \begin{equation}
\mathcal{T}_{\mathbf{v}(\mathbf{X}_k)}\left(\|\mathbf{v}(\mathbf{X})\|_2\right)=\|\mathbf{v}(\mathbf{X}_k)\|_2 + \left\langle \left.\frac{\partial \|\mathbf{v}(\mathbf{X})\|_2}{\partial\mathbf{v}(\mathbf{X})}\right|_{ \mathbf{v}(\mathbf{X})=\mathbf{v}(\mathbf{X}_k) }, \left[\mathbf{v}(\mathbf{X})-\mathbf{v}(\mathbf{X}_k)\right] \right\rangle.
\end{equation}\normalsize
{\color{black}It is remarkable that because $\|\mathbf{v}(\mathbf{X})\|_2$ is convex, $\mathcal{T}_{ \mathbf{v}(\mathbf{X}_k) }\left( \|\mathbf{v}(\mathbf{X})\|_2 \right)$ is a lower-bound of $\|\mathbf{v}(\mathbf{X})\|_2$ around $\forall \mathbf{v}(\mathbf{X}_k)$, i.e. $\|\mathbf{v}(\mathbf{X})\|_2 \geq \mathcal{T}_{ \mathbf{v}(\mathbf{X}_k) }\left( \|\mathbf{v}(\mathbf{X})\|_2 \right)$,
which implies that (\ref{g_cons}) is guaranteed if 
 \begin{equation}\label{g_cons_Taylor}
\mathcal{T}_{ \mathbf{v}(\mathbf{X}_k) }\left( \|\mathbf{v}(\mathbf{X})\|_2 \right) - 
\mu \sqrt{\mathcal{G}  } \geq 0
\end{equation} \normalsize
holds. 
Thus, $\mathcal{T}_{ \mathbf{v}(\mathbf{X}_k) }\left( \|\mathbf{v}(\mathbf{X})\|_2 \right)$ can be used as a surrogate function of $\|\mathbf{v}(\mathbf{X})\|_2$}.

To derive $\mathcal{T}_{ \mathbf{v}(\mathbf{X}_k) }\left( \|\mathbf{v}(\mathbf{X})\|_2 \right)$, we should first derive $\left.\frac{\partial \|\mathbf{v}(\mathbf{X})\|_2}{\partial\mathbf{v}(\mathbf{X})}\right|_{ \mathbf{v}(\mathbf{X})=\mathbf{v}(\mathbf{X}_k) } $. Nevertheless, we find that the inner product $\left\langle \left.\frac{\partial \|\mathbf{v}(\mathbf{X})\|_2}{\partial\mathbf{v}(\mathbf{X})}\right|_{ \mathbf{v}(\mathbf{X})=\mathbf{v}(\mathbf{X}_k) }, \left[\mathbf{v}(\mathbf{X})-\mathbf{v}(\mathbf{X}_k)\right] \right\rangle$ is generally a complex value, making (\ref{g_cons_Taylor}) invalid (i.e. “complex value” $\geq$ “real value”) during the optimization process. Fortunately, we note that the optimization for a complex variable can be intrinsically treated as the optimization for its real and imaginary parts. As such, we propose to transform $\mathbf{v}(\mathbf{X})$ by combining its real and imaginary parts into a new vector, in order to find an equivalent real representation of $\mathcal{T}_{ \mathbf{v}(\mathbf{X}_k) }\left( \|\mathbf{v}(\mathbf{X})\|_2 \right)$.

\begin{proposition}
A real representation of $\mathcal{T}_{ \mathbf{v}(\mathbf{X}_k) }\left( \|\mathbf{v}(\mathbf{X})\|_2 \right)$ is derived as
 \begin{equation}\label{real_representation}
\mathcal{T}_{ \widetilde{\mathbf{v}}(\mathbf{X}_k) }\left( \|\widetilde{\mathbf{v}}(\mathbf{X})\|_2 \right)
= tr\left(\mathbf{\Upsilon}_{\mathbf{X}_k}  \mathbf{X}\right),
\end{equation}\normalsize
where $\widetilde{\mathbf{v}}(\mathbf{X})$ is a real-value vector expressed as (\ref{widetilde_vQ_A}) in Appendix B.
$\mathbf{\Upsilon}_{\mathbf{X}_k}$ is a matrix with respect to $\mathbf{X}_k$ but is independent of $\mathbf{X}$, given by (\ref{Upsilon_Qk}) in Appendix B.
\end{proposition}
\begin{proof}
The proof is relegated to Appendix B.
\end{proof}

In accordance with the proof, $\mathcal{T}_{ \mathbf{v}(\mathbf{X}_k) }\left( \|\mathbf{v}(\mathbf{X})\|_2 \right) - \mu\sqrt{\mathcal{G}}\geq 0$ is equivalent to $\mathcal{T}_{ \widetilde{\mathbf{v}}(\mathbf{X}_k) }\left( \|\widetilde{\mathbf{v}}(\mathbf{X})\|_2 \right) - \mu\sqrt{\mathcal{G}} \geq 0$, such that $\mathcal{T}_{ \widetilde{\mathbf{v}}(\mathbf{X}_k) }\left( \|\widetilde{\mathbf{v}}(\mathbf{X})\|_2 \right)$ can be used as a surrogate function of $\|\mathbf{v}(\mathbf{X})\|_2$. Consequently, 
the constraint $\|\mathbf{v}(\mathbf{X})\|_2 - \mu \sqrt{\mathcal{G}} \geq 0$ can be convexified as
 \begin{equation}\label{Convexification-19c}
tr\left(\mathbf{\Upsilon}_{\mathbf{X}_k}  \mathbf{X}\right) - \mu \sqrt{\mathcal{G}} \geq 0.
\end{equation}\normalsize
Based on (\ref{Convexification-19c}), Problem $(\mathrm{P11})$ can be transformed into
 \begin{subequations}\label{P12}
\begin{align}
(\mathrm{P12}):\ \ \mathop{\mathrm{maximize}}\limits_{\mathbf{X}\succeq \mathbf{0},\mu>0}\ \  & tr(\mathbf{Z}_c \mathbf{X}), \\
\mathrm{subject\ to}\ \  & \mathrm{(\ref{P10}b),(\ref{P10}c),(\ref{P10}d), (\ref{Convexification-19c})}.
\end{align}
\end{subequations}\normalsize

It is worth mentioning that we aim to acquire the phase-shift solution of $(\mathrm{P11})$ by solving $(\mathrm{P12})$. However, given a certain point $\mathbf{X}_k$, the solution of $(\mathrm{P12})$ is only a local optimum around $\mathbf{X}_k$ instead of being the true optimum of $(\mathrm{P11})$. Therefore, we will design an SCA-based algorithm which iteratively solves $(\mathrm{P12})$ to approach the solution of $(\mathrm{P11})$. 

\subsubsection{\textbf{SCA-based Algorithm Design}}

The developed SCA-based algorithm is detailed in Algorithm 2, which can be summarized as follows.
First, we initialize the parameters including the initial point $\mathbf{X}_1$, iteration index $k$ and a terminating condition $\epsilon_2>0$.
Then, we successively solve $(\mathrm{P12})$ until the gap between the objective values in adjacent two iterations is smaller than $\epsilon_2$. Due to the non-convexity of $\mathrm{rank}(\mathbf{X})=1$, the SDR can be applied to drop the rank-one constraint. However, the SDR for $(\mathrm{P12})$ is not strictly guaranteed to be tight. If a rank-one solution is obtained, we can perform eigenvalue decomposition to acquire the optimal $\mathbf{q}$ and $\bm{\omega}$. Otherwise, we need to apply some convex relaxation techniques for phase-only beamforming, e.g. the SDP concave-convex procedure (SDP-CCP) \cite{SDP-CCP}, to iteratively acquire an approximate rank-one solution before eigenvalue decomposition. The SDP-CCP is detailed in [52, Algorithm 1].

\begin{algorithm}
\caption{The developed SCA-based algorithm.}
\LinesNumbered
{\bf Initialization:} \\
\ \ Initialize the first point $\mathbf{X}_1 \succeq \mathbf{0}$ satisfying $\mathrm{rank}(\mathbf{X}_1)=1$; Initialize the iteration time: $k\leftarrow 1$\;
\ \ Set a small constant $\epsilon_2 > 0$\;

{\bf Repeat:} \\
\ \ 1) Compute $\mathbf{\Upsilon}_{\mathbf{X}_k}$ according to (\ref{Upsilon_Qk})\;
\ \ 2) Solve $(\mathrm{P12})$ around point $\mathbf{X}_k$ with the aid of the SDR or SDP-CCP, and obtain the optimal solutions of $\mathbf{X}$ and $\mu$, denoted by $\mathbf{X}^{(\star)}$ and $\mu^{(\star)}$\;
\ \ 3) $\mathbf{X}_{k+1}\leftarrow \mathbf{X}^{(\star)}$\;
\ \ 4) Compute $\mathrm{Gap}_{k}^{k+1}=tr(\mathbf{Z}_c \mathbf{X}_{k+1}) - tr(\mathbf{Z}_c \mathbf{X}_k)$, compare $\mathrm{Gap}_{k}^{k+1}$ and $\epsilon_2$\;
\ \ 5) $k\leftarrow k+1$\;
{\bf Until:} $\mathrm{Gap}_{k}^{k+1} \leq \epsilon_2$\;

\ \ Obtain $\mathbf{Q}^{(\star)}=\mathbf{X}^{(\star)}/ \mu^{(\star)}$, and obtain $\mathbf{q}^{(\star)}$ from the eigenvalue decomposition of $\mathbf{Q}^{(\star)}$\;

{\bf Output:} $\bm{\omega}^{(\star)}=\mathbf{q}^{(\star)}[1:N]$\;

\end{algorithm}


\begin{proposition}
The proposed Algorithm 2 is convergent.
\end{proposition}
\begin{proof}
According to the proof of Proposition 2, we have $\|\mathbf{v}(\mathbf{X})\|_2=\|\widetilde{\mathbf{v}}(\mathbf{X})\|_2$. Then, because $\|\widetilde{\mathbf{v}}(\mathbf{X})\|_2$ is convex and $tr\left(\mathbf{\Upsilon}_{\mathbf{X}_k}  \mathbf{X}\right)$ is the first-order {\color{black}Taylor expansion} of $\|\widetilde{\mathbf{v}}(\mathbf{X})\|_2$, there should be $\|\mathbf{v}(\mathbf{X})\|_2 =\|\widetilde{\mathbf{v}}(\mathbf{X})\|_2 \geq tr\left(\mathbf{\Upsilon}_{\mathbf{X}_k}  \mathbf{X}\right)$ and $\|\mathbf{v}(\mathbf{X}_k)\|_2 =\|\widetilde{\mathbf{v}}(\mathbf{X}_k)\|_2 = tr\left(\mathbf{\Upsilon}_{\mathbf{X}_k}  \mathbf{X}_k\right)$ for $\forall \mathbf{X}_k$. 
As a result, we have $tr\left(\mathbf{\Upsilon}_{\mathbf{X}_{k+1}}  \mathbf{X}_{k+1}\right) = \|\widetilde{\mathbf{v}}(\mathbf{X}_{k+1})\|_2 \geq tr\left(\mathbf{\Upsilon}_{\mathbf{X}_{k}}  \mathbf{X}_{k+1}\right)$ for $\forall k\geq 1$. Since $\mathbf{X}_{k+1}$ is the optimal solution of $(\mathrm{P12})$ around point $\mathbf{X}_{k}$, there is $tr\left(\mathbf{\Upsilon}_{\mathbf{X}_{k}}  \mathbf{X}_{k+1}\right) - \mu \sqrt{\mathcal{G} } \geq 0$. Thus, we have $tr\left(\mathbf{\Upsilon}_{\mathbf{X}_{k+1}}  \mathbf{X}_{k+1}\right) - \mu \sqrt{\mathcal{G} } \geq 0$ because $tr\left(\mathbf{\Upsilon}_{\mathbf{X}_{k+1}}  \mathbf{X}_{k+1}\right) \geq tr\left(\mathbf{\Upsilon}_{\mathbf{X}_{k}}  \mathbf{X}_{k+1}\right)$, implying that $\mathbf{X}_{k+1}$ is a possible solution of $(\mathrm{P12})$ around point $\mathbf{X}_{k+1}$. 
If $\mathbf{X}_{k+1}$ is the optimal solution around point $\mathbf{X}_{k+1}$, we have $\mathbf{X}_{k+2}=\mathbf{X}_{k+1}$, resulting in $tr(\mathbf{Z}_c \mathbf{X}_{k+2}) = tr(\mathbf{Z}_c \mathbf{X}_{k+1})$. 
Otherwise, there must exist another solution $\mathbf{X}_{k+2}$ that is better than $\mathbf{X}_{k+1}$, such that $tr(\mathbf{Z}_c \mathbf{X}_{k+2}) > tr(\mathbf{Z}_c \mathbf{X}_{k+1})$ holds owing to the maximization of the objective. Consequently, we have $tr(\mathbf{Z}_c \mathbf{X}_{k+2}) \geq tr(\mathbf{Z}_c \mathbf{X}_{k+1})$ for $\forall k\geq 1$, implying that the objective value is monotonically increasing during the iteration process. Additionally, the objective function $tr(\mathbf{Z}_c \mathbf{X})$ in $(\mathrm{P12})$ is upper-bounded owing to the original unit-modulus phase-shift constraint. Therefore, we conclude that Algorithm 2 is convergent.
\end{proof}

\subsection{\color{black} Overall Algorithm Design for Solving $(\mathrm{P1})$}
{\color{black}
Based on the proposed solutions in Section III-A, B, C, we design the overall optimization algorithm, i.e. Algorithm 3, to solve $(\mathrm{P1})$ and acquire the solutions of $\mathbf{w}_{c}, \mathbf{w}_{s},\mathbf{w}_{rx}$ and $\bm{\omega}$. In Algorithm 3, we alternatively optimize $(\mathbf{w}_{c}, \mathbf{w}_{s})$, $\mathbf{w}_{rx}$ and $\bm{\omega}$, until the convergence criterion is eventually met. The configuration of the initial points $\mathbf{w}_{c,1},\mathbf{w}_{s,1},\mathbf{w}_{rx,1},\bm{\omega}_1$ will be given in Remark 4 in the following Section III-E, in accordance with the feasibility condition of the original optimization problem.

\begin{algorithm}
\color{black}
\caption{The overall algorithm for solving $(\mathrm{P1})$.}
\LinesNumbered
{\bf Initialization and input:} \\
\ \ Initialize $\mathbf{w}_{c,1},\mathbf{w}_{s,1},\mathbf{w}_{rx,1},\bm{\omega}_1$\;
\ \ Initialize the iteration time $ T \leftarrow 1$, terminating condition $\epsilon_3>0$\;

{\bf Repeat:} \\
\ \ 1) Input $\mathbf{w}_{rx,T},\bm{\omega}_T$ to Algorithm 1. 

\ \ \ \ \ Call Algorithm 1 to optimize 
$\mathbf{w}_{c}$, $\mathbf{w}_{s}$, and obtain $\mathbf{w}_{c}^{(\star)}$ and $\mathbf{w}_{s}^{(\star)}$\;
\ \ \ \ \   $\mathbf{w}_{c,T+1}\leftarrow\mathbf{w}_{c}^{(\star)}$, $\mathbf{w}_{s,T+1}\leftarrow\mathbf{w}_{s}^{(\star)}$\;
\ \ 2) Obtain $\mathbf{w}_{rx}^{(\star)}$ according to (\ref{wrx_opt}); $\mathbf{w}_{rx,T+1} \leftarrow \mathbf{w}_{rx}^{(\star)}$\;

\ \ 3) Input $\mathbf{w}_{c,T+1},\mathbf{w}_{s,T+1},\mathbf{w}_{rx,T+1}, \bm{\omega}_T$ to Algorithm 2, and use $\bm{\omega}_T$ to initialize the first point as

\ \ \ \ \ $\mathbf{X}_1=\frac{1}{tr(\mathbf{Z}_s \mathbf{Q}_1)+ \sigma_{n,u}^2} \mathbf{Q}_1$, where $\mathbf{Q}_1 = \mathbf{q}_1 \mathbf{q}_1^\mathrm{H}$ and $\mathbf{q}_1=[\bm{\omega}_T^\mathrm{T},1]^\mathrm{T}$;

\ \ \ \ \ Call Algorithm 2 to optimize 
$\bm{\omega}$ and obtain $\bm{\omega}^{(\star)}$; $\bm{\omega}_{T+1}\leftarrow \bm{\omega}^{(\star)}$ \;

\ \ 4) Compute 
$\mathrm{Gap}_{T}^{T+1}=\mathrm{SNR_{UE}}(\mathbf{w}_{c,T+1},\mathbf{w}_{s,T+1}, \bm{\omega}_{T+1}) - \mathrm{SNR_{UE}}(\mathbf{w}_{c,T},\mathbf{w}_{s,T}, \bm{\omega}_{T})$\;
\ \ 5) $T\leftarrow T+1$\;
{\bf Until:} $\mathrm{Gap}_{T}^{T+1} \leq \epsilon_3$\;


{\bf Output:} $\mathbf{w}_{c,T+1},\mathbf{w}_{s,T+1},\mathbf{w}_{rx,T+1}, \bm{\omega}_{T+1}$ as the final solutions\;

\end{algorithm}

\subsubsection{\textbf{Convergence Analysis}}
Since Algorithm 3 includes an alternative optimization process, it is necessary to analyse its convergence.

\begin{proposition}
The developed Algorithm 3 is convergent.
\end{proposition}
\begin{proof}
Owing to the maximizations of the objectives in $(\mathrm{P2})$ and $(\mathrm{P8})$, we have $\mathrm{SNR_{UE}}(\mathbf{w}_{c,T+1},\mathbf{w}_{s,T+1}, \bm{\omega}_{T}) \geq \mathrm{SNR_{UE}}(\mathbf{w}_{c,T},\mathbf{w}_{s,T}, \bm{\omega}_{T})$ and $\mathrm{SNR_{UE}}(\mathbf{w}_{c,T+1},\mathbf{w}_{s,T+1}, \bm{\omega}_{T+1}) \geq \mathrm{SNR_{UE}}(\mathbf{w}_{c,T+1},\mathbf{w}_{s,T+1}, \bm{\omega}_{T})$, yielding
\begin{equation}\label{Overall Convergence}
\mathrm{SNR_{UE}}(\mathbf{w}_{c,T+1},\mathbf{w}_{s,T+1}, \bm{\omega}_{T+1}) \geq \mathrm{SNR_{UE}}(\mathbf{w}_{c,T},\mathbf{w}_{s,T}, \bm{\omega}_{T}).
\end{equation}
The inequality (\ref{Overall Convergence}) implies that the objective value in the $(T+1)$-th iteration is not lower than that in the $T$-th iteration. Additionally, due to the total transmit power limit in (\ref{Overall_problem}c) and the unit-modulus phase shift constraint in (\ref{Overall_problem}b), the objective function $\mathrm{SNR_{UE}}(\mathbf{w}_{c},\mathbf{w}_{s}, \bm{\omega})$ is upper-bounded. Consequently, we prove that Algorithm 3 is convergent.
\end{proof}

\subsubsection{\textbf{Computational Complexity Analysis}}
In addition to the convergence behavior, the computational complexity is also an important property to be analysed for a new algorithm. For each iteration in Algorithm 3, the complexities of the optimizations of $(\mathbf{w}_c,\mathbf{w}_s)$, $\mathbf{w}_{rx}$ and $\bm{\omega}$ are detailed as follows:
\begin{enumerate}
\item[•] Optimization of $(\mathbf{w}_c,\mathbf{w}_s)$: Since an SDP problem (i.e. $(\mathrm{P5})$ without the rank-one constraint) needs to be solved after each branch during the bisection search, the computational complexity of this part is $\log_2\left(\frac{t_{\mathrm{max}}-t_{\mathrm{min}}}{\epsilon_1} \right) \times \mathcal{O}(M^{4.5}\log(1/\mathfrak{u}))$, where $\log_2\left(\frac{t_{\mathrm{max}}-t_{\mathrm{min}}}{\epsilon_1} \right)$ is the number of execution times of the bisection search, and $\mathcal{O}(M^{4.5}\log(1/\mathfrak{u}))$ is the complexity of using SDR to solve $(\mathrm{P5})$ in each branch, with $\mathfrak{u}$ being the solution accuracy \cite{SDR-SPM}.

\item[•] Optimization of $\mathbf{w}_{rx}$: According to (\ref{wrx_opt}), one time of eigenvalue decomposition is required to obtain the solution of $\mathbf{w}_{rx}$, whose complexity is $\mathcal{O}(M^3)$.

\item[•] Optimization of $\bm{\omega}$: The phase-shift variables are optimized by executing Algorithm 2 to successively solve $(\mathrm{P12})$, so that the complexity of this part is $K\times\mathcal{O}\left((N+3)^{4}N^{\frac{1}{2}}\log(1/\mathfrak{u})\right)$, where $K$ is the total iteration times required for Algorithm 2 to converge, which will be numerically investigated via simulations in Section IV.
\end{enumerate}

As a result, the computational complexity of each iteration in Algorithm 3 is $\log_2\left(\frac{t_{\mathrm{max}}-t_{\mathrm{min}}}{\epsilon_1} \right) \times \mathcal{O}(M^{4.5}\log(1/\mathfrak{u})) 
+ \mathcal{O}(M^3) + K\times\mathcal{O}\left((N+3)^{4}N^{\frac{1}{2}}\log(1/\mathfrak{u})\right)$.

}

\subsection{Problem Feasibility Condition and UDR}

On account of the total transmit power limit and the unit-modulus phase-shift constraint, the detection probability is restricted below a certain value $U\in(0,1)$, implying that there exists an upper bound $U\in(0,1)$ such that $\mathbb{P}_d (\mathbf{\Psi}_S,r_{\mathrm{max}},\Delta_\theta,\Delta_\varphi)\leq U$ holds for any $ \mathbf{w}_{c}, \mathbf{w}_{s},\mathbf{w}_{rx},\bm{\omega}$ satisfying (\ref{Overall_problem}b), (\ref{Overall_problem}c), (\ref{Overall_problem}d). 
If $U < \gamma_{\mathbb{P}_d^{(\mathrm{min})}}$, Problem $(\mathrm{P1})$ will become infeasible because we cannot find a group of $\mathbf{w}_{c}, \mathbf{w}_{s},\mathbf{w}_{rx},\bm{\omega}$ making constraint (\ref{Overall_problem}e) hold in any case.  Therefore, to keep $(\mathrm{P1})$ feasible, the following condition should be satisfied.

\begin{remark}\color{black}
Problem $(\mathrm{P1})$ is feasible when 
$U=\mathop{\mathrm{max}}\limits_{\mathbf{w}_{c}, \mathbf{w}_{s},\mathbf{w}_{rx},\bm{\omega}}\ \mathbb{P}_d (\mathbf{\Psi}_S,r_\mathrm{max},\Delta_\theta,\Delta_\varphi) \geq \gamma_{\mathbb{P}_d^{(\mathrm{min})}}$ holds, where $U$ can be derived by solving the following problem:
 \begin{subequations}\label{P13} 
\begin{align}
(\mathrm{P13}):\ \ \mathop{\mathrm{maximize}}\limits_{\mathbf{w}_{c}, \mathbf{w}_{s},\mathbf{w}_{rx},\bm{\omega}}\ \  & \mathbb{P}_d (\mathbf{\Psi}_S,r_\mathrm{max},\Delta_\theta,\Delta_\varphi), \\
\mathrm{subject\ to}\ \  & \mathrm{(\ref{Overall_problem}b), (\ref{Overall_problem}c), (\ref{Overall_problem}d)}.
\end{align}
\end{subequations}\normalsize
This is because if $U=\mathop{\mathrm{max}}\limits_{\mathbf{w}_{c}, \mathbf{w}_{s},\mathbf{w}_{rx},\bm{\omega}}\ \mathbb{P}_d (\mathbf{\Psi}_S,r_\mathrm{max},\Delta_\theta,\Delta_\varphi) \geq \gamma_{\mathbb{P}_d^{(\mathrm{min})}}$ holds, there must exist a group of $\mathbf{w}_{c}, \mathbf{w}_{s},\mathbf{w}_{rx}$, $\bm{\omega}$ that can be found to satisfy constraint (\ref{Overall_problem}e).
\end{remark}
Problem $(\mathrm{P13})$ can be solved by alternatively optimizing $(\mathbf{w}_{c}, \mathbf{w}_{s})$, $\mathbf{w}_{rx}$ and $\bm{\omega}$. To optimize $(\mathbf{w}_{c}, \mathbf{w}_{s})$, the objective function in (\ref{P13}a) can be transformed into $tr(\mathbf{M}_r (\mathbf{W}_c + \mathbf{W}_s))$; to optimize $\mathbf{w}_{rx}$, the objective function in (\ref{P13}a) can be recast as $\mathbf{w}_{rx}^\mathrm{H} \mathbf{M}_{cs} \mathbf{w}_{rx}$; to optimize $\bm{\omega}$, the objective function in (\ref{P13}a) can be converted into $ \| \mathbf{v}(\mathbf{Q}) \|_2$. Then, $(\mathbf{w}_{c}, \mathbf{w}_{s})$, $\mathbf{w}_{rx}$ and $\bm{\omega}$ can be updated alternatively by solving the following three subproblems:
\begin{align}\label{P14}
(\mathrm{P14}):\ \ \mathop{\mathrm{maximize}}\limits_{\mathbf{W}_{c}\succeq\mathbf{0}, \mathbf{W}_{s}\succeq\mathbf{0}}\ \   tr(\mathbf{M}_r (\mathbf{W}_c + \mathbf{W}_s)), \ \ 
\mathrm{subject\ to}\ \   \mathrm{(\ref{P3}b),(\ref{P3}d)},
\end{align}\normalsize 
\begin{align}\label{P15}
(\mathrm{P15}):\ \ \mathop{\mathrm{maximize}}\limits_{\mathbf{w}_{rx}}\ \  \mathbf{w}_{rx}^\mathrm{H} \mathbf{M}_{cs} \mathbf{w}_{rx}, \ \ 
\mathrm{subject\ to}\ \  \mathrm{(\ref{P7}b)},
\end{align}\normalsize
\begin{align}\label{P16}
(\mathrm{P16}):\ \ \mathop{\mathrm{maximize}}\limits_{\mathbf{Q}\succeq \mathbf{0}}\ \  \| \mathbf{v}(\mathbf{Q}) \|_2, \ \ 
\mathrm{subject\ to}\ \  \mathrm{(\ref{P9}b),(\ref{P9}c)}.
\end{align}\normalsize

The solutions of $(\mathrm{P14})$-$(\mathrm{P16})$ can be found using the SDR, eigenvalue decomposition and SCA, which can be done in similar ways as described in Section III. A, B, C. The details are skipped here for conciseness.
Moreover, we note that $\mathbb{P}_d (\mathbf{\Psi}_S,r_\mathrm{max},\Delta_\theta,\Delta_\varphi)$ is positively related to $\Delta_\theta$ and $\Delta_\varphi$, because the two terms affect the integral in the expression of the received echo power $P_{rx}(\mathbf{\Psi}_S,r_\mathrm{max},\Delta_\theta,\Delta_\varphi)$ in $\mathbb{P}_d (\mathbf{\Psi}_S,r_\mathrm{max},\Delta_\theta,\Delta_\varphi)$. This implies that if $\Delta_\theta$ and $\Delta_\varphi$ are reduced, $\mathbb{P}_d (\mathbf{\Psi}_S,r_\mathrm{max},\Delta_\theta,\Delta_\varphi)$ will decrease concomitantly on account of the reduction of $P_{rx}(\mathbf{\Psi}_S,r_\mathrm{max},\Delta_\theta,\Delta_\varphi)$. Therewith, 
we can determine $\mathrm{UDR}(\Delta_\theta,\Delta_\varphi,\gamma_{\mathbb{P}_d^{(\mathrm{min})}})$ introduced in Section II in accordance {\color{black}with
\begin{equation}\label{Detail_UDR}
\mathrm{UDR}(\Delta_\theta,\Delta_\varphi,\gamma_{\mathbb{P}_d^{(\mathrm{min})}}) \in \{S(\Delta_\theta,\Delta_\varphi)\ | \ \mathop{\mathrm{max}}\limits_{\mathbf{w}_{c}, \mathbf{w}_{s},\mathbf{w}_{rx},\bm{\omega}}\ \mathbb{P}_d (\mathbf{\Psi}_S,r_\mathrm{max},\Delta_\theta,\Delta_\varphi) = \gamma_{\mathbb{P}_d^{(\mathrm{min})}} \},
\end{equation}
which means that $\mathrm{UDR}(\Delta_\theta,\Delta_\varphi,\gamma_{\mathbb{P}_d^{(\mathrm{min})}})$ belongs to a class of $S(\Delta_\theta,\Delta_\varphi)$, whose corresponding $\Delta_\theta$ and $\Delta_\varphi$ satisfy $\mathop{\mathrm{max}}\limits_{\mathbf{w}_{c}, \mathbf{w}_{s},\mathbf{w}_{rx},\bm{\omega}}\ \mathbb{P}_d (\mathbf{\Psi}_S,r_\mathrm{max},\Delta_\theta,\Delta_\varphi) = \gamma_{\mathbb{P}_d^{(\mathrm{min})}}$. On this basis, if $\Delta_\theta$ and $\Delta_\varphi$ are further scaled down, there will be $\mathop{\mathrm{max}}\limits_{\mathbf{w}_{c}, \mathbf{w}_{s},\mathbf{w}_{rx},\bm{\omega}}\ \mathbb{P}_d (\mathbf{\Psi}_S,r_\mathrm{max},\Delta_\theta,\Delta_\varphi) < \gamma_{\mathbb{P}_d^{(\mathrm{min})}}$ such that $(\mathrm{P1})$ becomes infeasible. As a result, the above (\ref{Detail_UDR}) characterizes the scattering surface area of the smallest $\gamma_{\mathbb{P}_d^{(\mathrm{min})}}$-detectable target.}


Based on the problem feasibility condition in Remark 3, the initial points in Algorithm 3 can be configured as follows.

\begin{remark}\color{black}

Let $ \mathbf{w}_{c}^{(\ddagger)}, \mathbf{w}_{s}^{(\ddagger)},\mathbf{w}_{rx}^{(\ddagger)},\bm{\omega}^{(\ddagger)}$ be the solutions of $(\mathrm{P13})$.
Then, if the feasibility condition $\mathop{\mathrm{max}}\limits_{\mathbf{w}_{c}, \mathbf{w}_{s},\mathbf{w}_{rx},\bm{\omega}}\ \mathbb{P}_d (\mathbf{\Psi}_S,r_\mathrm{max},\Delta_\theta,\Delta_\varphi) \geq \gamma_{\mathbb{P}_d^{(\mathrm{min})}}$ is satisfied, we can use $ \mathbf{w}_{c}^{(\ddagger)}, \mathbf{w}_{s}^{(\ddagger)},\mathbf{w}_{rx}^{(\ddagger)},\bm{\omega}^{(\ddagger)}$ as the initial points of Algorithm 3, i.e. $(\mathbf{w}_{c,1},\mathbf{w}_{s,1},\mathbf{w}_{rx,1},\bm{\omega}_1)= (\mathbf{w}_{c}^{(\ddagger)}, \mathbf{w}_{s}^{(\ddagger)},\mathbf{w}_{rx}^{(\ddagger)},\bm{\omega}^{(\ddagger)})$, such that $(\mathrm{P1})$ has feasible solutions during the optimization process.

\end{remark}



\section{Simulation Results}

This section carries out the simulations to testify the analytical properties of the proposed algorithms, and to evaluate the performance of our beamforming design through the comparisons with several existing benchmarks, including two recently developed approaches considering point targets.


In the simulations, the positions of the BS, the RIS and the UE are {\color{black} $(0,0,18)$, $(2,10,12)$ and $(-30,80,25)$ in meters}. Unless stated otherwise, $\mathbf{\Psi}_S$ and $r_{\mathrm{max}}$ are set as {\color{black}$\mathbf{\Psi}_S=(0.38\pi, 0.44\pi)$ and $r_{\mathrm{max}}=8$ m. The elevation and azimuth angle-spreads of the scattering surface area are $\Delta_\theta=\Delta_\varphi=\frac{\pi}{16}$}. The noise powers are $\sigma_{n,u}^2=-80$ dBm and $\sigma_n^2=-90$ dBm; the carrier frequency is $f_c=2.5$ GHz, resulting in the signal wavelength of $\lambda=11.99$ cm; the average scattering loss is $\mathbb{E}[\delta_S]=-10$ dB \cite{TWC-localization}; the antenna/element spacing is half-wavelength; the terminating conditions of Algorithm 1, 2 and 3 are $\epsilon_1=0.0002$, $\epsilon_2=0.002$ and $\epsilon_3=0.002$; the duration of the sensing time slot is $T_0=0.1$ s; the sampling frequency at the BS is $f_s=1$ kHz; the decision threshold is $\eta=10^{-4.5}$; the Rician factor is $\kappa=10$; the path loss exponents for BS-RIS and RIS-UE links are $\alpha_{BR}=\alpha_{RU}=2.2$, while that for BS-UE link is $\alpha_{BU}=3.5$; the path loss coefficient is $\zeta_0=-30$ dB; the total transmit power limit is $P_{tx}=30$ dBm.
The other parameters including $M$, $N$, $\gamma_{\mathbb{P}_d^{(\mathrm{min})}}$ will vary for diverse observations.

The proposed algorithms are programmed by MATLAB and executed on a personal computer with Intel Core i7 and 32 GB RAM. The optimization problems are solved by CVX Toolbox with Sedumi solver. When computing the integral, we apply the trapezoidal method to approximate the integral value with an accuracy of 100 divisions over the variable range.

\subsection{Convergence Behavior}

\begin{figure}[h]
\centering
\subfloat[]{\includegraphics[width=3.3in]{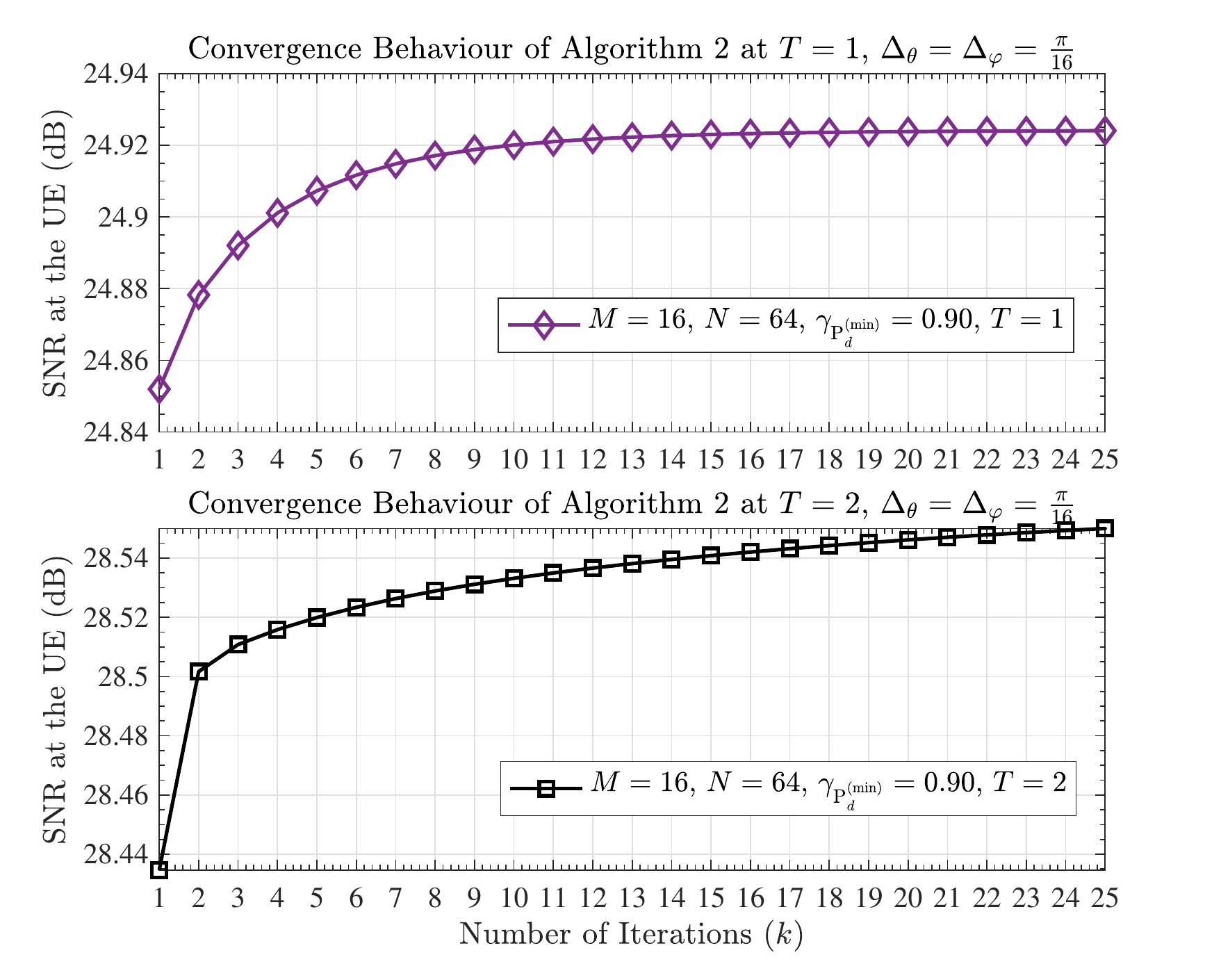}}
\subfloat[]{\includegraphics[width=3.3in]{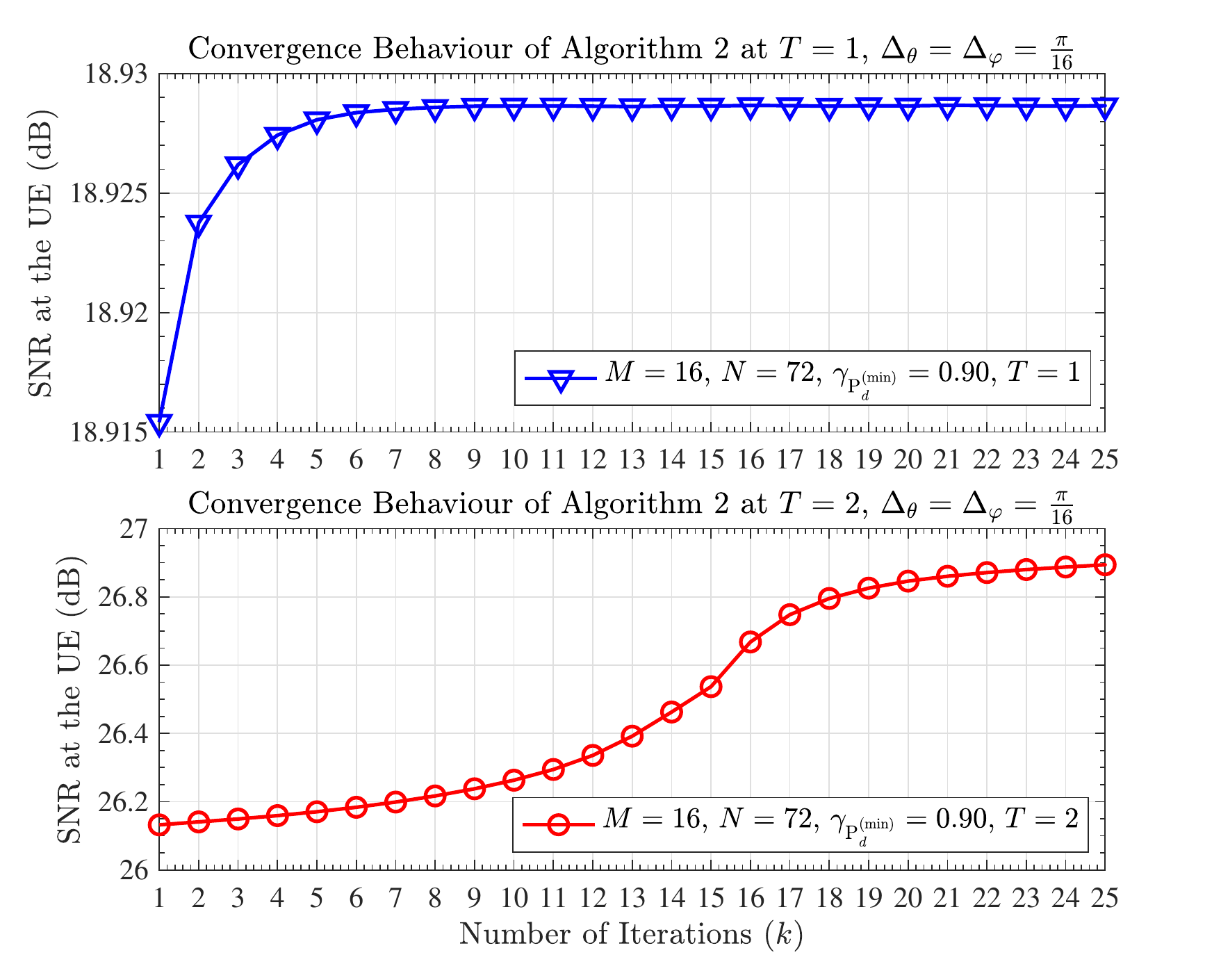}}
\hfil
\caption{\color{black}The convergence behaviors of Algorithm 2 at $T=1$ and $T=2$, when the numbers of BS antennas and RIS reflectors are: (a) $M=16,N=64$, and (b) $M=16,N=72$.}
\label{SCA-Convergence}
\end{figure}

{\color{black}Based on the above parameter configurations, we first numerically investigate the convergence behaviors of the proposed algorithms. Fig. \ref{SCA-Convergence} depicts the convergence curves of Algorithm 2 pertaining to $1\sim 25$ iterations at $T=1$ and $T=2$. From Fig. \ref{SCA-Convergence}, we observe that the SNR at the UE monotonically increases during the iteration process.
Besides, the convergence rates of Algorithm 2 are distinct under different simulation settings, without indicating an explicit regularity with respect to ether $N$ or $T$. The number of iterations required for convergence is shown to be diverse. However, it is worth noting that after sufficient times of iterations, all the curves presented in the figure can eventually converge to implicit upper limits, which testifies the proof of Proposition 3.}

\begin{figure}[h]
\centering
\begin{minipage}[t]{3.4in}
\includegraphics[width=3.3in]{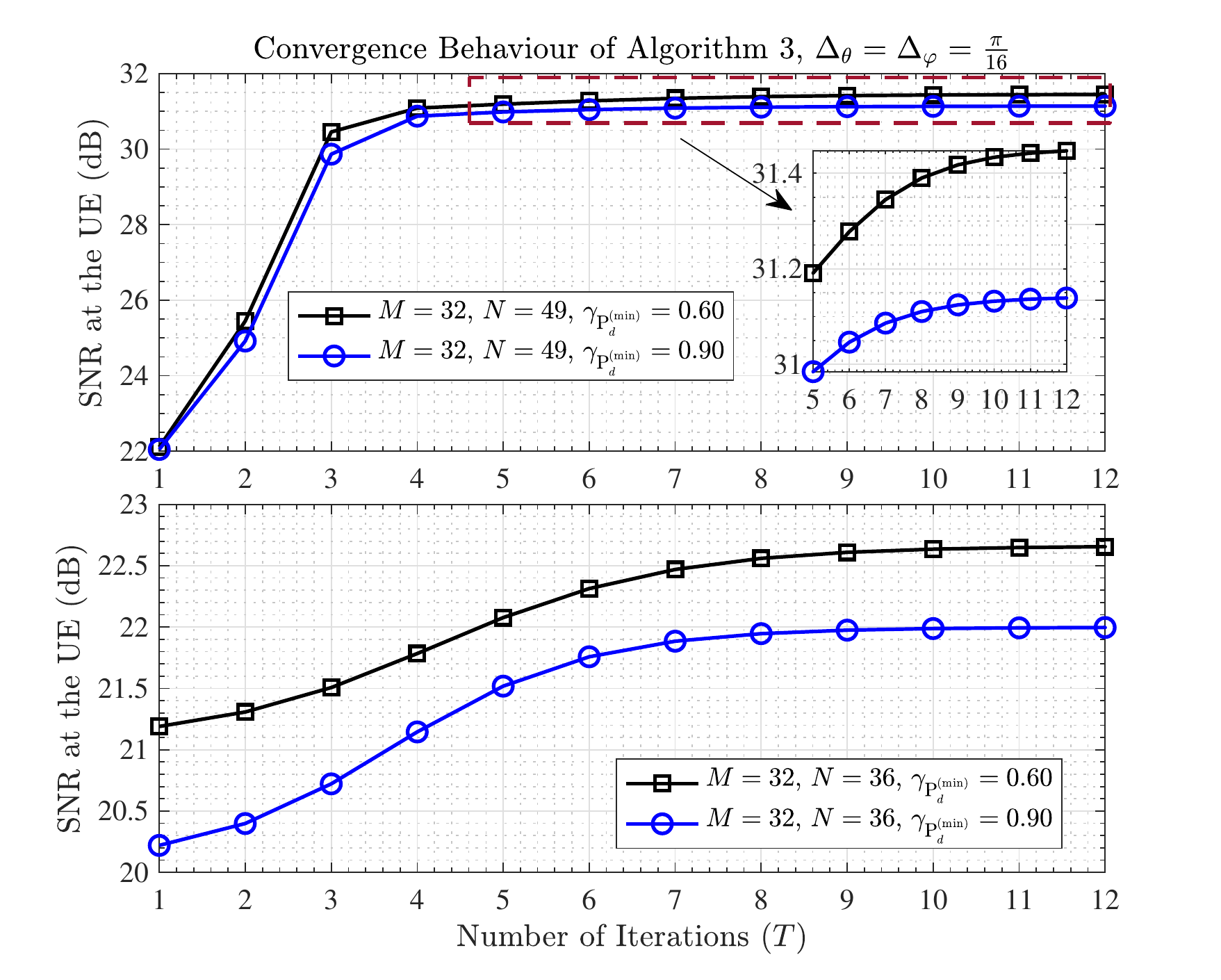}
\hfil
\centering
\caption{\color{black}Convergence behavior of Algorithm 3 at $M=32$. }
\label{Overall-Convergence}
\end{minipage}
\begin{minipage}[t]{3.4in}
\includegraphics[width=3.3in]{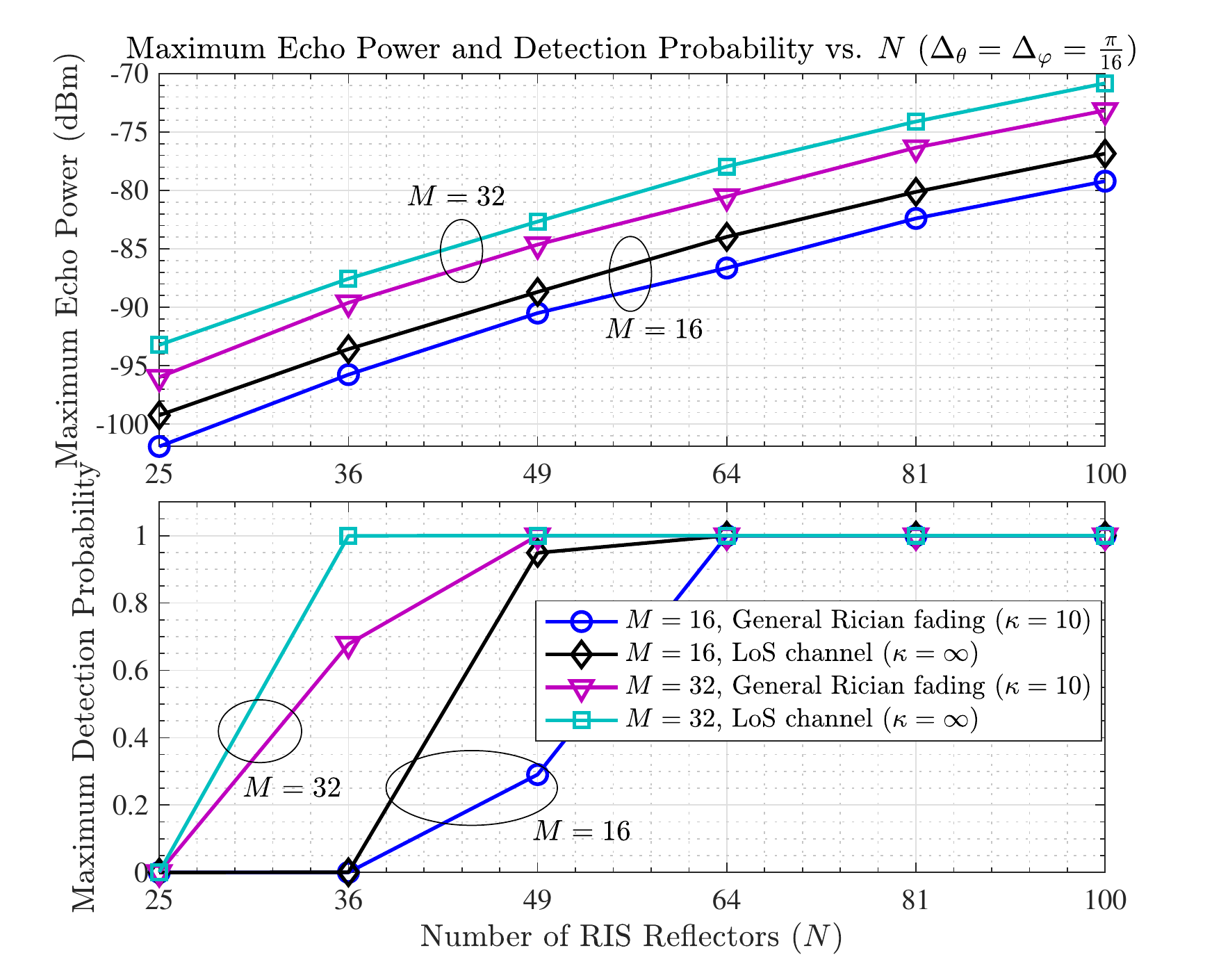}
\hfil
\centering
\caption{\color{black}Maximum echo power and detection probability with respect to $N$ at $P_{tx}=30$ dBm.}
\label{maxPrx_maxPd_N}
\end{minipage}
\end{figure}

{\color{black}Fig. \ref{Overall-Convergence} plots the convergence curves of Algorithm 3 pertaining to $1\sim 12$ iterations at $M=32$.  It is demonstrated that all the curves in the figure show an uptrend and converge to certain upper limits (within about $10$ iterations). This phenomenon is consistent with the proof of Proposition 4. Moreover, the SNRs at the UE with $N=49$ can converge to higher levels compared to those with $N=36$, due to the higher passive beamforming gain generated by larger RIS.}

\subsection{Problem Feasibility Condition}

{\color{black}Then, we simulate the problem feasibility condition, by presenting the maximum echo power and the maximum detection probability obtained by solving $(\mathrm{P13})$. The results with respect to $N$ are provided in Fig. \ref{maxPrx_maxPd_N}. 
It is demonstrated that when $N$ grows from $25$ to $100$, both the maximum echo power and the maximum detection probability are improved concomitantly, owing to the increase of the passive beamforming gain produced by the RIS.  It is noteworthy that when a small number of RIS reflectors (e.g., $N=25$) are equipped, the received echo is significantly weak with an echo power below $-93$ dBm. This is because the echo inherently experiences double reflections before reaching the BS, thereby suffering from severe path loss when the passive beamforming gain is not sufficiently high. Nevertheless, it can be observed that as $N$ further grows to $100$, the echo strength is enhanced by $20$ dB or so, while the detection probability is improved by more than $95\%$. This result is generalized for both pure LoS and Rician fading ($\kappa=10$) scenarios. As such, we can conclude that adding more reflectors to the RIS is an effective way to improve the potential detection performance.

Here, we also investigate the required minimum transmit power for Problem $(\mathrm{P1})$ to be feasible. The results are presented in Fig. \ref{Problem_feasibility-N} by varying $N$. Fig. \ref{Problem_feasibility-N} indicates that}
increasing the number of RIS reflectors can prominently reduce the required minimum transmit power. From this perspective, it is shown that the enlargement of the RIS can be a viable solution to energy conservation without posing much performance loss.

\begin{figure}[h]
\centering
\begin{minipage}[t]{3.4in}
\includegraphics[width=3.3in]{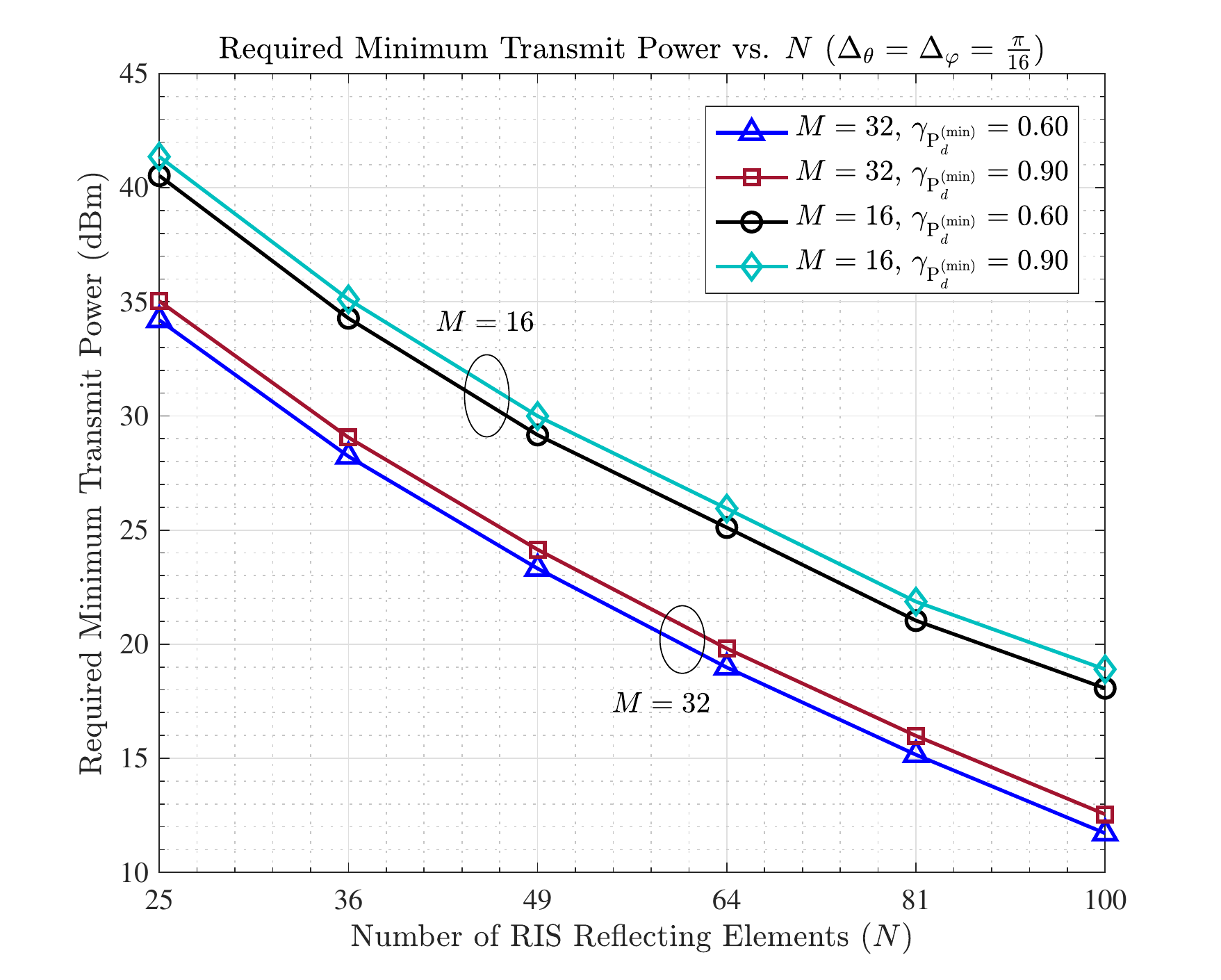}
\hfil
\centering
\caption{\color{black}Required minimum transmit power in dBm with respect to $N$. }
\label{Problem_feasibility-N}
\end{minipage}
\begin{minipage}[t]{3.4in}
\includegraphics[width=3.3in]{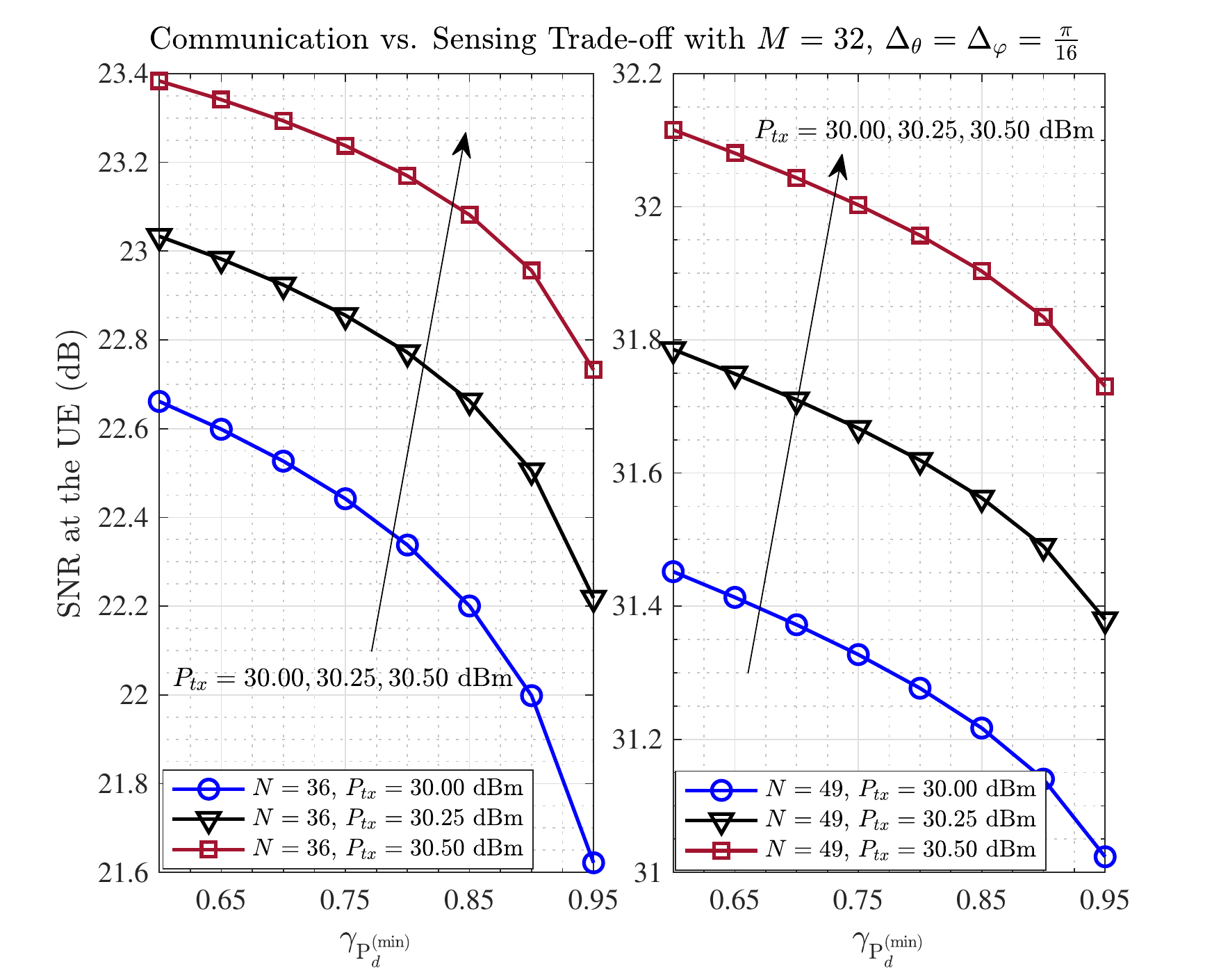}
\hfil
\centering
\caption{\color{black}SNR at the UE with respect to $\gamma_{\mathbb{P}_d^{(\mathrm{min})}}$.}
\label{Communication-sensing}
\end{minipage}
\end{figure}

\subsection{Communication vs.  Sensing Trade-off}

According to our problem formulation, an inherent trade-off exists between the communication and sensing performances. We identify this trade-off by presenting the joint variations of the optimized SNR at the UE and the minimum detection probability requirement in Fig. \ref{Communication-sensing}.
It is shown that the increase of the minimum detection probability requirement results in an SNR degradation. The rationale behind is that when $\mathbf{\Psi}_S$ deviates from $\mathbf{\Psi}_U$, the optimized passive beamforming gains toward the UE and the sensing direction are generally competing paradoxes, i.e. one's enhancement requires another's compromise. This phenomenon occurs because of the limitations of the system power and the transmission/reflection array scales. Therefore, a balance between sensing and communication performances deserves to be pursued in an RIS-enabled ISAC scenario.

\subsection{UDR vs. Sensing-duration Trade-off}

Note that in (\ref{Detail_UDR}), a trade-off between the UDR and $T_0$ is implied when $\gamma_{\mathbb{P}_d^{(\mathrm{min})}}$ is given. Specifically, the equality $\mathop{\mathrm{max}}\limits_{\mathbf{w}_{c}, \mathbf{w}_{s},\mathbf{w}_{rx},\bm{\omega}}\ \mathbb{P}_d (\mathbf{\Psi}_S,r_\mathrm{max},\Delta_\theta,\Delta_\varphi) = \gamma_{\mathbb{P}_d^{(\mathrm{min})}}$ in (\ref{Detail_UDR}) can be rewritten as
 \begin{equation}\nonumber
T_0 = \frac{\sigma_n^2}{ f_s \mathop{\mathrm{max}}\limits_{\mathbf{w}_{c}, \mathbf{w}_{s},\mathbf{w}_{rx},\bm{\omega}} P_{rx} (\mathbf{\Psi}_S,r_\mathrm{max},\Delta_\theta,\Delta_\varphi) } \left(   Q^{-1}(\mathbb{P}_f)  -  Q^{-1} (\gamma_{\mathbb{P}_d^{(\mathrm{min})}}) \right)^2,
 \end{equation}
indicating that $T_0$ is inversely proportional to $\mathop{\mathrm{max}}\limits_{\mathbf{w}_{c}, \mathbf{w}_{s},\mathbf{w}_{rx},\bm{\omega}} P_{rx} (\mathbf{\Psi}_S,r_\mathrm{max},\Delta_\theta,\Delta_\varphi)$. If the angle-spread $\Delta_\theta$ or $\Delta_\varphi$ is scaled down, the value of $\mathop{\mathrm{max}}\limits_{\mathbf{w}_{c}, \mathbf{w}_{s},\mathbf{w}_{rx},\bm{\omega}} P_{rx} (\mathbf{\Psi}_S,r_\mathrm{max},\Delta_\theta,\Delta_\varphi)$ will be reduced on account of the decrease of the integral value inside, thereby leading to an increase of $T_0$.

\begin{figure}[h]
\centering
\begin{minipage}[t]{3.4in}
\includegraphics[width=3.3in]{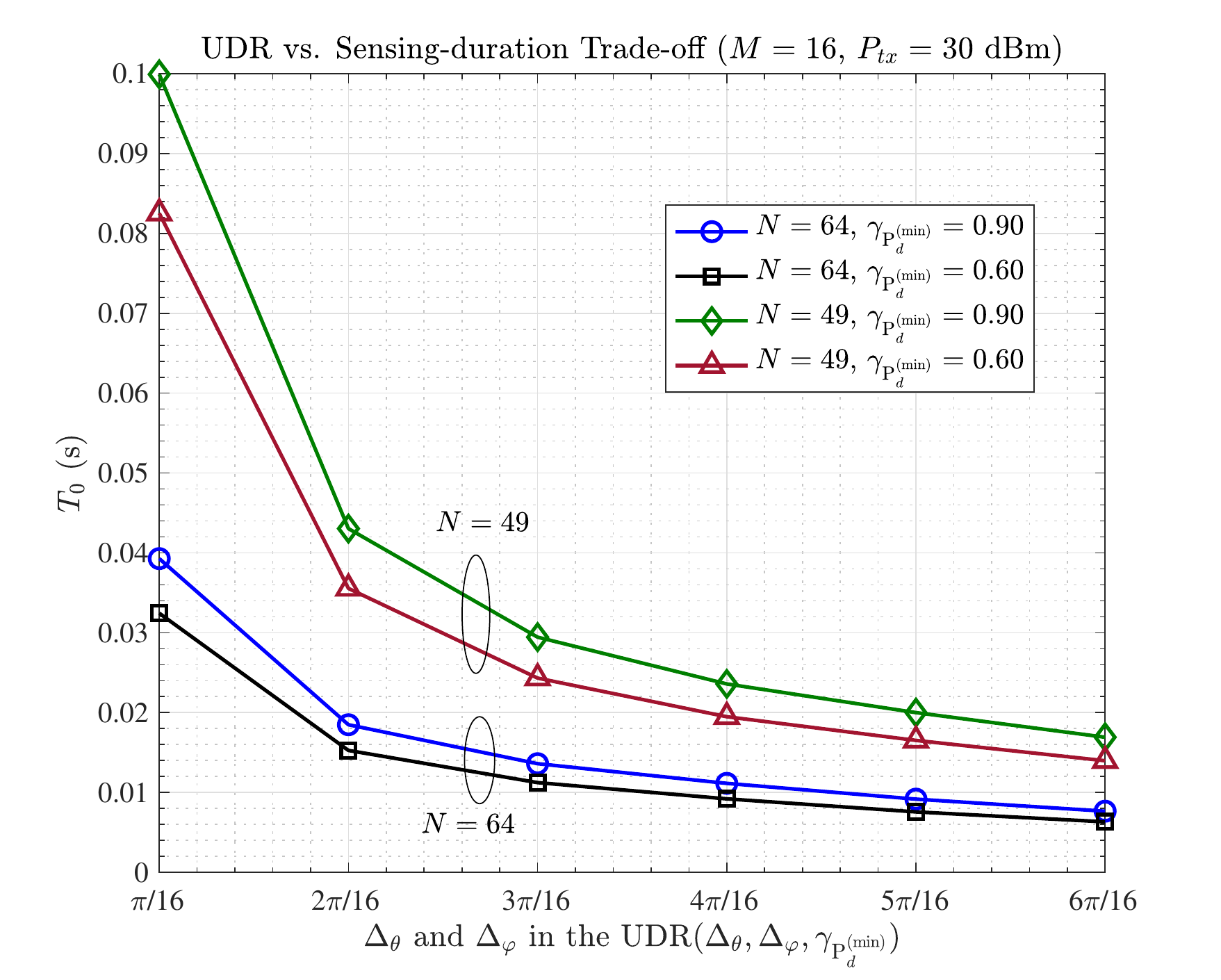}
\hfil
\centering
\caption{\color{black}$T_0$ with respect to $\Delta_\theta$ and $\Delta_\varphi$ in the $\mathrm{UDR}(\Delta_\theta,\Delta_\varphi,\gamma_{\mathbb{P}_d^{(\mathrm{min})}})$. }
\label{Trade-off}
\end{minipage}
\begin{minipage}[t]{3.4in}
\includegraphics[width=3.3in]{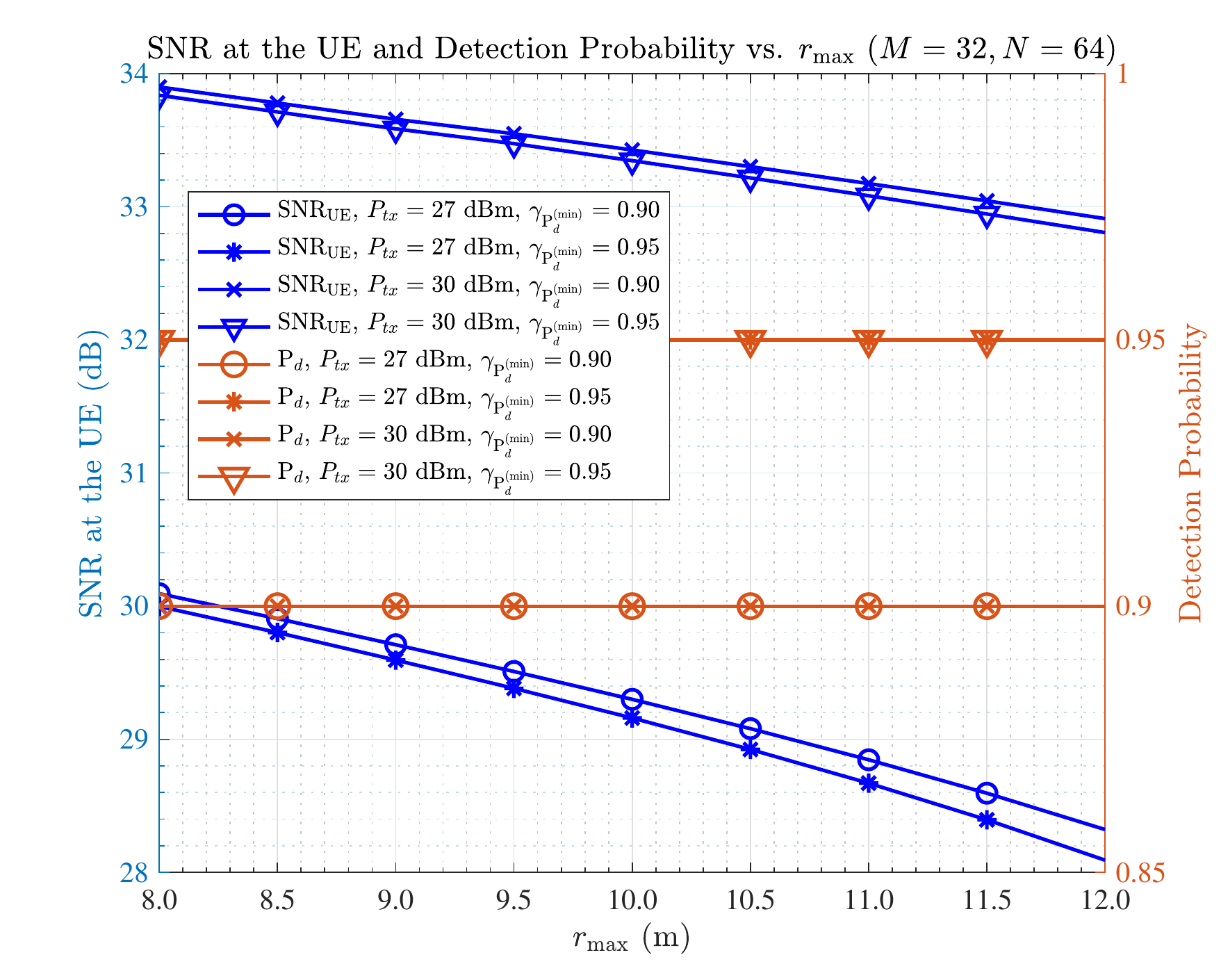}
\hfil
\centering
\caption{\color{black}SNR at the UE and detection probability with respect to $r_\mathrm{max}$.}
\label{With_respect_to_rmax}
\end{minipage}
\end{figure}

Considering $\Delta_\theta=\Delta_\varphi$ for simplicity, we show this trade-off in Fig. \ref{Trade-off} by depicting $T_0$ with respect to $\Delta_\theta$ and $\Delta_\varphi$ in the $\mathrm{UDR}(\Delta_\theta,\Delta_\varphi,\gamma_{\mathbb{P}_d^{(\mathrm{min})}})$. We observe that longer $T_0$ is required for detecting the target with smaller $\Delta_\theta$ and $\Delta_\varphi$. The rationale behind is that according to (\ref{Pd_metric}), the detection probability is proportional to the echo energy-to-noise ratio (ENR) at the BS, denoted by $\mathrm{ENR}=\frac{T_0 f_s P_{rx} (\mathbf{\Psi}_S,r_\mathrm{max},\Delta_\theta,\Delta_\varphi) }{\sigma_n^2}$. As smaller $\Delta_\theta$ and $\Delta_\varphi$ result in lower $P_{rx} (\mathbf{\Psi}_S,r_\mathrm{max},\Delta_\theta,\Delta_\varphi)$, $T_0$ should be prolonged to compensate for the reduction of $P_{rx} (\mathbf{\Psi}_S,r_\mathrm{max},\Delta_\theta,\Delta_\varphi)$ in order to maintain the ENR.
Besides, Fig. \ref{Trade-off} also illustrates that if $T_0$ remains, a smaller target can be detected with a minimum detection probability of $\gamma_{\mathbb{P}_d^{(\mathrm{min})}}$ when $N$ grows. In this regard, we conclude that increasing the number of RIS reflectors can enhance the target detection capability.

{\color{black} \subsection{Performance with Respect to $r_\mathrm{max}$}}

{\color{black}
Subsequently, we examine the optimization performance by varying $r_\mathrm{max}$ to change the length of the BS-RIS-target-RIS-BS link. The optimization results of the SNR at the UE and the detection probability are depicted in Fig. \ref{With_respect_to_rmax}. From Fig. \ref{With_respect_to_rmax}, we observe that when $r_\mathrm{max}$ increases, the SNR at the UE descends whereas the detection probability remains at $\gamma_{\mathbb{P}_d^{(\mathrm{min})}}$. This is because the overall optimization problem is formulated to maximize the SNR while ensuring the detection probability to be not lower than $\gamma_{\mathbb{P}_d^{(\mathrm{min})}}$. When the target moves far away from the RIS, the echo will experience severer path loss. Then, the passive beamforming gain toward the sensing direction should be strengthened to maintain the required detection performance, while the gain toward the UE is compromised in view of the communication vs. sensing trade-off.
}

\subsection{Performance Comparisons with Existing Benchmarks}

Finally, we compare our method with several existing benchmarks, including: 
{\color{black}
\textbf{Benchmark 1-random phase shifts}: The RIS phase shifts are randomly generated and follow uniform distribution on $[0,2\pi]$. The SNR and echo power are averaged over 300 Monte Carlo trials.
\textbf{Benchmark 2-without sensing constraint:} The optimization result obtained by solving $(\mathrm{P1})$ in the absence of the constraint (\ref{Overall_problem}e).
\textbf{Benchmark 3-directional phase shifts:} Replacing the phase-shift solution of Benchmark 2 with $\bm{\omega} = \mathbf{a}(\theta_U,\varphi_U) \odot \mathbf{a}^*(\theta_R,\varphi_R)$.}
\textbf{Benchmark 4:} The approach of maximizing the beamforming gain toward the sensing direction $\mathbf{\Psi}_S$ \textit{without the consideration of the echoes}, while satisfying the minimum SNR requirement of the UE. Similar ideas were adopted in some existing works, such as \cite{RIS-ISAC-2}. 
\textbf{Benchmark 5:}
The approach of maximizing the beamforming gain over the cascaded BS-RIS-target-RIS-BS link in consideration of the echoes, \textit{which however, treats the targets as points}. Similar ideas were adopted in e.g. \cite{RIS-ISAC-1}.  
For comparisons with Benchmark 4 and 5, we conduct the following steps:
1) Execute the proposed Algorithm 3 under certain configurations of $\gamma_{\mathbb{P}_d^{(\mathrm{min})}}$, $\Delta_\theta$ and $\Delta_\varphi$ to acquire the solutions of $\mathbf{w}_{c}, \mathbf{w}_{s},\mathbf{w}_{rx},\bm{\omega}$, and record the corresponding SNR at the UE, the echo power and the detection probability for $S(\Delta_\theta,\Delta_\varphi)$ as $\mathrm{SNR_{UE,(proposed)}}$, $P_{rx,\mathrm{(proposed)}}$ and $\mathbb{P}_{d,\mathrm{(proposed)}}$.
2) Set the minimum SNR requirements of Benchmark 4 and 5 as $\mathrm{SNR_{UE,(proposed)}}$, and perform the two benchmarks to acquire their beamforming solutions.
3) Use the beamforming solutions obtained in Step 2) to compute the detection probabilities of Benchmark 4 and 5 for $S(\Delta_\theta,\Delta_\varphi)$, denoted by $\mathbb{P}_{d,\mathrm{(B4)}}$ and $\mathbb{P}_{d,\mathrm{(B5)}}$, respectively.
4) Compare $\mathbb{P}_{d,\mathrm{(B4)}}$, $\mathbb{P}_{d,\mathrm{(B5)}}$ with $\mathbb{P}_{d,\mathrm{(proposed)}}$.

\begin{figure}[h]
\includegraphics[width=3.3in]{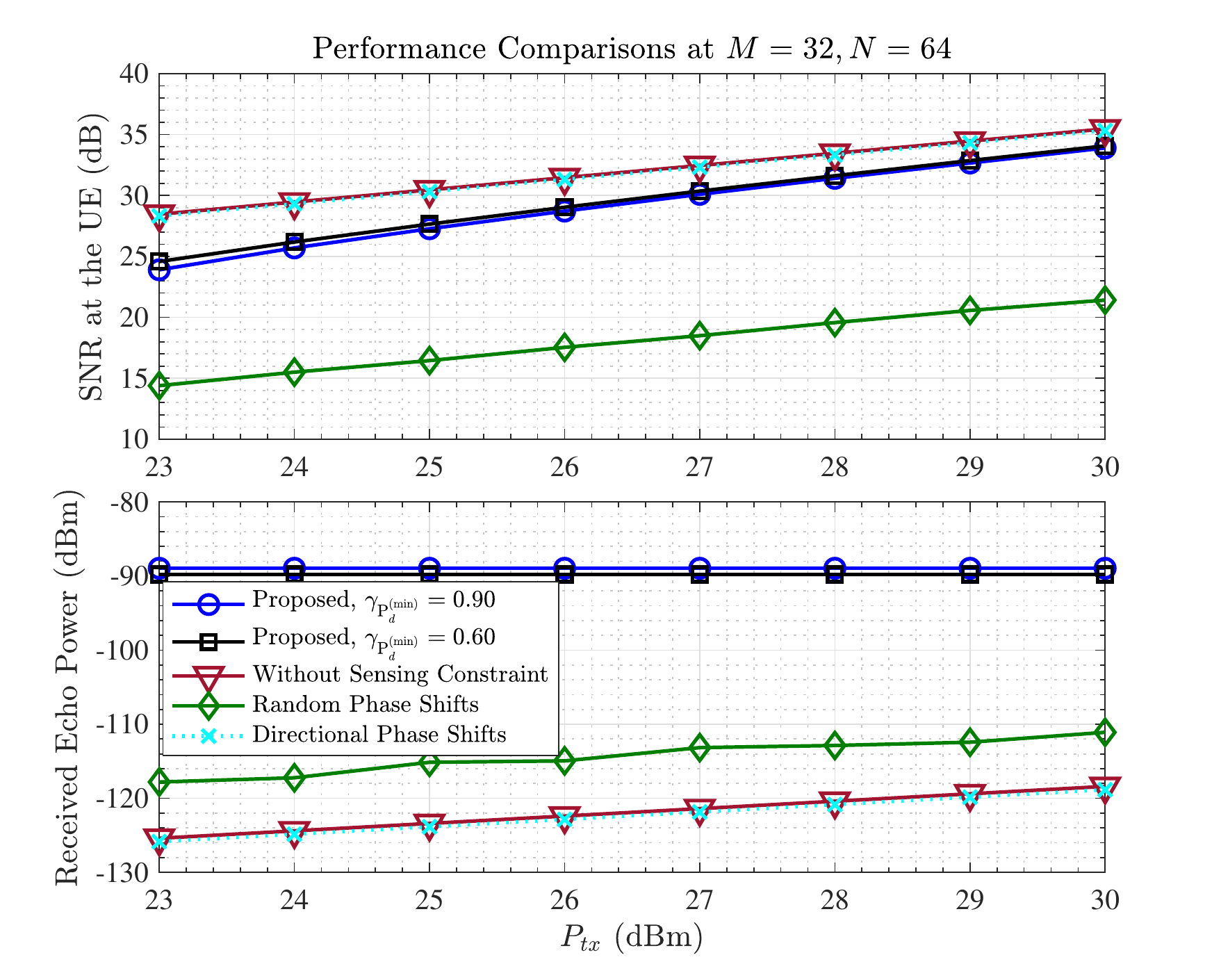}
\hfil
\centering
\caption{\color{black}Performance comparisons with Benchmark 1, 2 and 3.}
\label{Comparison_Benchmarks}
\end{figure}

\begin{spacing}{1.36}
{\color{black} Fig. \ref{Comparison_Benchmarks} illustrates the performance comparisons with Benchmark 1, 2 and 3 in terms of SNR at the UE and the received echo power.  It can be seen that the proposed optimization method outperforms Benchmark 1 to a large extend, since Benchmark 1 merely randomizes the RIS phase shifts instead of optimizing them. Additionally, the SNR of the proposed approach is relatively lower than those of Benchmark 2 and 3, whereas the received echo power shows the opposite. This is because both Benchmark 2 and 3 completely focus on maximizing the SNR at the UE without considering the target sensing, so that their passive beamforming gains toward the sensing direction are significantly weak. As a result, the proposed beamforming design is validated to be advantageous in simultaneously supporting UE communication and target detection.}
\end{spacing}

\vspace*{5pt}
\begin{table}[h]
\caption{Performance comparisons with Benchmark 4 and 5 at $M=32$, $N=64$ and $P_{tx} = 30$ dBm.}
\label{Table_with_B34}
\color{black}
\centering
\begin{small}
\begin{tabular}{ccccc}
\hline
\hline
& $\Delta_{\theta}$ and $\Delta_{\varphi}$ $(\Delta_{\theta}=\Delta_{\varphi})$ & SNR at the UE & Received Echo Power & Detection Probability\\
\hline
\hline
Proposed & $\pi/16$ & $33.8993$ dB & $-88.9526$ dBm & $\mathbb{P}_{d,\mathrm{(proposed)}} = 0.9000$ \\ 
Benchmark 4 & $\pi/16$ & $33.8993$ dB & $-105.8100$ dBm & $\mathbb{P}_{d,\mathrm{(B4)}} = 2.6450\times10^{-17}$ \\
Benchmark 5 & $\pi/16$ & $33.8993$ dB & $-89.0670$ dBm & $\mathbb{P}_{d,\mathrm{(B5)}} =  0.8716$ \\
\hline
Proposed & $2\pi/16$ & $34.5037$ dB & $-88.9526$ dBm & $\mathbb{P}_{d,\mathrm{(proposed)}} = 0.9000$ \\ 
Benchmark 4 & $2\pi/16$ & $34.5037$ dB & $-99.6171$ dBm & $\mathbb{P}_{d,\mathrm{(B4)}} = 1.0769\times10^{-11}$ \\
Benchmark 5 & $2\pi/16$ & $34.5037$ dB & $-89.1962$ dBm & $\mathbb{P}_{d,\mathrm{(B5)}} =  0.8339$ \\
\hline
Proposed & $3\pi/16$ & $34.6705$ dB & $-88.9526$ dBm & $\mathbb{P}_{d,\mathrm{(proposed)}} = 0.9000$ \\ 
Benchmark 4 & $3\pi/16$ & $34.6705$ dB & $-96.8008$ dBm & $\mathbb{P}_{d,\mathrm{(B4)}} = 2.8253\times10^{-8}$ \\
Benchmark 5 & $3\pi/16$ & $34.6705$ dB & $-89.4402$ dBm & $\mathbb{P}_{d,\mathrm{(B5)}} =  0.7472$ \\
\hline
Proposed & $4\pi/16$ & $35.1372$ dB & $-88.9526$ dBm & $\mathbb{P}_{d,\mathrm{(proposed)}} = 0.9000$ \\ 
Benchmark 4 & $4\pi/16$ & $35.1372$ dB & $-96.6249$ dBm & $\mathbb{P}_{d,\mathrm{(B4)}} = 4.7506\times10^{-8}$ \\
Benchmark 5 & $4\pi/16$ & $35.1372$ dB & $-94.1107$ dBm & $\mathbb{P}_{d,\mathrm{(B5)}} =  8.1513\times 10^{-5}$ \\
\hline
\hline
\end{tabular}
\end{small}
\end{table}

Table \ref{Table_with_B34} lists the performances of the proposed algorithm, Benchmark 4 and Benchmark 5 with respect to different $\Delta_{\theta}$ and $\Delta_{\varphi}$. The results illustrate that our proposed beamforming strategy performs prominently better than Benchmark 4 in terms of the echo power and the detection probability, under the premise of the same SNRs at the UE. This is because Benchmark 4 simply considers the BS-RIS-target link, without providing additional passive beamforming gain for the target-RIS-BS echo link.
More importantly, in comparison with Benchmark 5, we observe that the echo power and the detection probability of Benchmark 5 are close to those of the proposed algorithm when $\Delta_{\varphi}=\Delta_{\theta}=\pi/16$. Nevertheless, as $\Delta_{\theta}$ and $\Delta_{\varphi}$ increase from $\pi/16$ to $4\pi/16$, the performance gap emerges and widens. 
In specific, the detection probability of the proposed strategy remains at $90\%$, whereas that of Benchmark 5 descends in a monotonic way. 
This phenomenon occurs because when detecting a large possible target, the Benchmark 5 aims to maximize the beamforming gain toward a point, such that the sensing beam generated by Benchmark 5 may only cover a portion of the scattering surface area.
By contrast, the sensing beam generated by the proposed strategy is optimized in consideration of the target size, such that it can illuminate the entire scattering surface area of the target as much as possible. 



\section{Conclusions and Prospects}


In our work, the RIS was leveraged to assist the joint communication and target detection in the ISAC, and a novel joint active and passive beamforming optimization approach was developed to optimize the communication and sensing performances. To this end, the detection probability was derived based on the illumination power on an approximated scattering surface area of the target, and the UDR was defined to measure the target detection capability. Then, an optimization problem was formulated to maximize the SNR at the UE under a minimum detection probability constraint. {\color{black}To solve this non-convex problem, the communication and sensing beamformers, the receive combining vector and the RIS phase shifts are alternatively optimized and updated. }The results validated that our proposed beamforming design was superior to the investigated benchmarks when detecting targets with practical sizes, so that it would be promising in contributing to the development of the RIS-enabled ISAC.
{\color{black}In the future, the proposed design in this paper is worth to be extended to a more general multi-user case with multiple targets to be simultaneously detected.}

\appendices

\section{Proof of Proposition 1}

We prove by first considering the transformation of $\mathbf{q}^\mathrm{H} \widetilde{\mathbf{A}}(\theta,\varphi) \mathbf{q} \mathbf{q}^\mathrm{H} \widetilde{\mathbf{B}}(\theta,\varphi) \mathbf{q}$ inside the integral in $f(\mathbf{q})$. In specific, by introducing $\mathbf{Q}=\mathbf{q}\mathbf{q}^\mathrm{H}$, we have
\small \begin{equation}
\mathbf{q}^\mathrm{H} \widetilde{\mathbf{A}}(\theta,\varphi) \mathbf{q} \mathbf{q}^\mathrm{H} \widetilde{\mathbf{B}}(\theta,\varphi) \mathbf{q}=tr\left\{ \widetilde{\mathbf{B}}(\theta,\varphi) \mathbf{Q} \widetilde{\mathbf{A}}(\theta,\varphi)  \mathbf{Q} \right\}
= \left[\mathrm{vec}(\mathbf{Q}^\mathrm{T})\right]^\mathrm{T}
\left[\widetilde{\mathbf{A}}^\mathrm{T}(\theta,\varphi) \otimes \widetilde{\mathbf{B}}(\theta,\varphi)\right] \mathrm{vec}(\mathbf{Q}).
\end{equation}\normalsize
Since $\left[\mathrm{vec}(\mathbf{Q}^\mathrm{T})\right]^\mathrm{T}=\left[\mathrm{vec}(\mathbf{q}^*\mathbf{q}^\mathrm{T} )\right]^\mathrm{T}=\left(\mathbf{q}\otimes\mathbf{q}^*  \right)^\mathrm{T}=\mathbf{q}^\mathrm{T} \otimes \mathbf{q}^\mathrm{H}$ and $\mathrm{vec}(\mathbf{Q})=\mathrm{vec}(\mathbf{q}\mathbf{q}^\mathrm{H})=\mathbf{q}^* \otimes \mathbf{q}$ hold, we have
\small \begin{equation}\label{Kro}
 \left[\mathrm{vec}(\mathbf{Q}^\mathrm{T})\right]^\mathrm{T}
\left[\widetilde{\mathbf{A}}^\mathrm{T}(\theta,\varphi) \otimes \widetilde{\mathbf{B}}(\theta,\varphi)\right] \mathrm{vec}(\mathbf{Q})= 
\left(\mathbf{q}^\mathrm{T} \otimes \mathbf{q}^\mathrm{H}\right) \left[\widetilde{\mathbf{A}}^\mathrm{T}(\theta,\varphi) \otimes \widetilde{\mathbf{B}}(\theta,\varphi)\right] \left( \mathbf{q}^* \otimes \mathbf{q} \right).
\end{equation}\normalsize
In order to use (\ref{Kro}) to simplify $f(\mathbf{q})$, we define two auxiliary matrices, i.e. $\mathbf{C}(\theta,\varphi)$ and $\mathbf{K}$, as 
\small \begin{equation}
\mathbf{C}(\theta,\varphi)=\widetilde{\mathbf{A}}^\mathrm{T}(\theta,\varphi) \otimes \widetilde{\mathbf{B}}(\theta,\varphi),\ \ \ 
\mathbf{K}=\int_{\varphi_S-\frac{\Delta_{\varphi}}{2}}^{\varphi_S+\frac{\Delta_{\varphi}}{2}} \int_{\theta_S-\frac{\Delta_{\theta}}{2}}^{\theta_S+\frac{\Delta_{\theta}}{2}} \mathbf{C}(\theta,\varphi)\ d\theta\  d\varphi.
\end{equation}\normalsize
Then, $f(\mathbf{q})$ can be recast as
\small \begin{equation}\label{f_omega_K}
f(\mathbf{q})=\left(\mathbf{q}^\mathrm{T} \otimes \mathbf{q}^\mathrm{H}\right) \mathbf{K} \left( \mathbf{q}^* \otimes \mathbf{q} \right).
\end{equation}\normalsize

It is noted that because: 1) $\mathbf{C}(\theta,\varphi)\succeq \mathbf{0}$ for $\forall\varphi\in [\varphi_S-\frac{\Delta_{\varphi}}{2},\varphi_S+\frac{\Delta_{\varphi}}{2}]$ and $\forall\theta\in [\theta_S-\frac{\Delta_{\theta}}{2},\theta_S+\frac{\Delta_{\theta}}{2}]$, and 2) $\mathbf{C}^{\mathrm{H}}(\theta,\varphi)=\widetilde{\mathbf{A}}^\mathrm{*}(\theta,\varphi) \otimes \widetilde{\mathbf{B}}^\mathrm{H}(\theta,\varphi)=\widetilde{\mathbf{A}}^\mathrm{T}(\theta,\varphi) \otimes \widetilde{\mathbf{B}}(\theta,\varphi)=\mathbf{C}(\theta,\varphi)$, we have $\mathbf{K}\succeq \mathbf{0}$ and $\mathbf{K}^\mathrm{H}=\mathbf{K}$, implying that $\mathbf{K}$ is a positive semidefinite Hermitian matrix. Therefore, $\mathbf{K}$ can be decomposed into $\mathbf{K}=\mathbf{U}\mathbf{\Lambda}\mathbf{U}^{\mathrm{H}}$ via the eigenvalue decomposition, where $\mathbf{U}$ is a unitary matrix, and $\mathbf{\Lambda}$ is a diagonal eigenvalue matrix, whose diagonal elements are non-negative. 

Substituting $\mathbf{K}=\mathbf{U}\mathbf{\Lambda}\mathbf{U}^{\mathrm{H}}$ into (\ref{f_omega_K}), we obtain
\small \begin{align}
f(\mathbf{q})=&\left(\mathbf{q}^\mathrm{T} \otimes \mathbf{q}^\mathrm{H}\right) \mathbf{U}\mathbf{\Lambda}\mathbf{U}^{\mathrm{H}} \left( \mathbf{q}^* \otimes \mathbf{q} \right) 
= \left(\mathbf{q}^\mathrm{T} \otimes \mathbf{q}^\mathrm{H}\right) \mathbf{U}\mathbf{\Lambda}^{\frac{1}{2}}\mathbf{\Lambda}^{\frac{1}{2}}\mathbf{U}^{\mathrm{H}} \left( \mathbf{q}^* \otimes \mathbf{q} \right)
= \left\|\left(\mathbf{q}^\mathrm{T} \otimes \mathbf{q}^\mathrm{H}\right) \mathbf{U}\mathbf{\Lambda}^{\frac{1}{2}}\right\|_2^2.
\end{align}\normalsize

By rewriting $\mathbf{U}\mathbf{\Lambda}^{\frac{1}{2}}$ as $\mathbf{U}\mathbf{\Lambda}^{\frac{1}{2}}=[\bm{\sigma}_1,\bm{\sigma}_2,\cdots,\bm{\sigma}_{(N+1)^2}]$, we have
\small \begin{align}
f(\mathbf{q})=& \left\|\left(\mathbf{q}^\mathrm{T} \otimes \mathbf{q}^\mathrm{H}\right) [\bm{\sigma}_1,\bm{\sigma}_2,\cdots,\bm{\sigma}_{(N+1)^2}] \right\|_2^2 \nonumber\\
=& \left\| \left[\left(\mathbf{q}^\mathrm{T} \otimes \mathbf{q}^\mathrm{H}\right)\bm{\sigma}_1,\left(\mathbf{q}^\mathrm{T} \otimes \mathbf{q}^\mathrm{H}\right)\bm{\sigma}_2,\cdots,\left(\mathbf{q}^\mathrm{T} \otimes \mathbf{q}^\mathrm{H}\right)\bm{\sigma}_{(N+1)^2}\right] \right\|_2^2 \nonumber\\
=& \left\| \left[\mathbf{q}^\mathrm{H} (\mathrm{vec}^{-1}(\bm{\sigma}_1)) \mathbf{q},\mathbf{q}^\mathrm{H} (\mathrm{vec}^{-1}(\bm{\sigma}_2)) \mathbf{q},\cdots,\mathbf{q}^\mathrm{H} (\mathrm{vec}^{-1}(\bm{\sigma}_{(N+1)^2})) \mathbf{q}\right] \right\|_2^2 \nonumber\\
=& \left\| \left[tr\left\{ (\mathrm{vec}^{-1}(\bm{\sigma}_1)) \mathbf{q} \mathbf{q}^\mathrm{H}\right\},tr\left\{ (\mathrm{vec}^{-1}(\bm{\sigma}_2)) \mathbf{q} \mathbf{q}^\mathrm{H}\right\},\cdots,tr\left\{ (\mathrm{vec}^{-1}(\bm{\sigma}_{(N+1)^2})) \mathbf{q} \mathbf{q}^\mathrm{H}\right\}\right] \right\|_2^2 \nonumber\\
=& \left\| \left[tr\left\{ (\mathrm{vec}^{-1}(\bm{\sigma}_1)) \mathbf{Q}\right\},tr\left\{ (\mathrm{vec}^{-1}(\bm{\sigma}_2)) \mathbf{Q}\right\},\cdots,tr\left\{ (\mathrm{vec}^{-1}(\bm{\sigma}_{(N+1)^2})) \mathbf{Q}\right\}\right] \right\|_2^2,
\end{align}\normalsize
where $\bm{\sigma}_i$ is the $i$-th column of $\mathbf{U}\mathbf{\Lambda}^{\frac{1}{2}}$. Finally, by denoting $\mathrm{vec}^{-1}(\bm{\sigma}_{i})$ as
\small \begin{equation}\label{S_i}
\mathbf{S}_i = \mathrm{vec}^{-1}(\bm{\sigma}_{i}), \ \ \mathrm{for}\ \  i=1,2,...,(N+1)^2,
\end{equation}\normalsize
we complete the proof of Proposition 1.

\section{Proof of Proposition 2}
Here, we aim to derive the real representation of $\mathcal{T}_{ \mathbf{v}(\mathbf{X}_k) }\left( \|\mathbf{v}(\mathbf{X})\|_2 \right)$. 
We rewrite the vector $\mathbf{v}(\mathbf{X})$ as
\small \begin{equation}\label{vQ}
\mathbf{v}(\mathbf{X}) = \mathbf{v}_\mathrm{R}(\mathbf{X}) + j \mathbf{v}_\mathrm{I}(\mathbf{X}),
\end{equation}\normalsize
where $ \mathbf{v}_\mathrm{R}(\mathbf{X})$ and $ \mathbf{v}_\mathrm{I}(\mathbf{X})$ are, respectively, the real part and imaginary part of $\mathbf{v}(\mathbf{X})$, given by
\small \begin{align}
\mathbf{v}_\mathrm{R}(\mathbf{X}) =& \left[\mathfrak{Re}\{tr(\mathbf{S}_1 \mathbf{X})\}, \mathfrak{Re}\{tr(\mathbf{S}_2 \mathbf{X})\},\cdots,\mathfrak{Re}\{tr(\mathbf{S}_{(N+1)^2} \mathbf{X})\}\right]^\mathrm{T}  \nonumber\\
=& \left[\frac{1}{2}tr\left( (\mathbf{S}_1 + \mathbf{S}_1^\mathrm{H}) \mathbf{X}\right), \frac{1}{2}tr\left( (\mathbf{S}_2 + \mathbf{S}_2^\mathrm{H}) \mathbf{X}\right),\cdots,\frac{1}{2}tr\left( (\mathbf{S}_{(N+1)^2} + \mathbf{S}_{(N+1)^2}^\mathrm{H}) \mathbf{X}\right)\right]^\mathrm{T},
\end{align}\normalsize
\small \begin{align}
\mathbf{v}_\mathrm{I}(\mathbf{X}) =& \left[\mathfrak{Im}\{tr(\mathbf{S}_1 \mathbf{X})\}, \mathfrak{Im}\{tr(\mathbf{S}_2 \mathbf{X})\},\cdots,\mathfrak{Im}\{tr(\mathbf{S}_{(N+1)^2} \mathbf{X})\}\right]^\mathrm{T}  \nonumber\\
=& \left[-\frac{j}{2}tr\left( (\mathbf{S}_1 - \mathbf{S}_1^\mathrm{H}) \mathbf{X}\right), -\frac{j}{2}tr\left( (\mathbf{S}_2 - \mathbf{S}_2^\mathrm{H}) \mathbf{X}\right),\cdots,-\frac{j}{2}tr\left( (\mathbf{S}_{(N+1)^2} - \mathbf{S}_{(N+1)^2}^\mathrm{H}) \mathbf{X}\right)\right]^\mathrm{T}.
\end{align}\normalsize
Then, we combine $ \mathbf{v}_\mathrm{R}(\mathbf{X})$ and $ \mathbf{v}_\mathrm{I}(\mathbf{X})$ into a new vector $\widetilde{\mathbf{v}}(\mathbf{X})$, formed by
\small \begin{equation}\label{widetilde_vQ_A}
\begin{split}
\widetilde{\mathbf{v}}(\mathbf{X}) = \left[\begin{matrix}
\mathbf{v}_\mathrm{R}(\mathbf{X}) \\
\mathbf{v}_\mathrm{I}(\mathbf{X})
\end{matrix}\right] 
\in \mathbb{R}^{2(N+1)^2\times 1}.
\end{split}
\end{equation}\normalsize
Hence, according to (\ref{vQ}) and (\ref{widetilde_vQ_A}), we have
\small \begin{equation}
\|\widetilde{\mathbf{v}}(\mathbf{X})\|_2 = \|\mathbf{v}(\mathbf{X})\|_2 = 
\sqrt{   \mathbf{v}_\mathrm{R}^\mathrm{T}(\mathbf{X}) \mathbf{v}_\mathrm{R}(\mathbf{X}) +  \mathbf{v}_\mathrm{I}^\mathrm{T}(\mathbf{X}) \mathbf{v}_\mathrm{I}(\mathbf{X})  },
\end{equation}\normalsize
implying that the constraints $ \|\mathbf{v}(\mathbf{X})\|_2 - 
\mu \sqrt{\mathcal{G}  } \geq 0$ and $ \|\widetilde{\mathbf{v}}(\mathbf{X})\|_2 - 
\mu \sqrt{\mathcal{G}  } \geq 0$ are equivalent. 
Note that according to (\ref{widetilde_vQ_A}), $\widetilde{\mathbf{v}}(\mathbf{X})$ is a real-value vector, such that the partial derivatives of $\|\widetilde{\mathbf{v}}(\mathbf{X})\|_2$ with respect to $ \mathbf{v}_\mathrm{R}(\mathbf{X})$ and $ \mathbf{v}_\mathrm{I}(\mathbf{X})$ are completely real. 

After some manipulations, we obtain 
\small \begin{equation}
\frac{\partial \|\widetilde{\mathbf{v}}(\mathbf{X})\|_2}{\partial \mathbf{v}_\mathrm{R}(\mathbf{X})} = \frac{\mathbf{v}_\mathrm{R}(\mathbf{X})}{\|\widetilde{\mathbf{v}}(\mathbf{X})\|_2} , \ \ 
\frac{\partial \|\widetilde{\mathbf{v}}(\mathbf{X})\|_2}{\partial \mathbf{v}_\mathrm{I}(\mathbf{X})} = \frac{\mathbf{v}_\mathrm{I}(\mathbf{X})}{\|\widetilde{\mathbf{v}}(\mathbf{X})\|_2}.
\end{equation}\normalsize
Then, we have
\small \begin{equation}\label{Derivative_k}
\left.\frac{\partial \|\widetilde{\mathbf{v}}(\mathbf{X})\|_2}{\partial \widetilde{\mathbf{v}}(\mathbf{X})}\right|_{ \widetilde{\mathbf{v}}(\mathbf{X})=\widetilde{\mathbf{v}}(\mathbf{X}_k) } = 
\left[\begin{matrix}
\left.\frac{\partial \|\widetilde{\mathbf{v}}(\mathbf{X})\|_2}{\partial \mathbf{v}_\mathrm{R}(\mathbf{X})}\right|_{ \widetilde{\mathbf{v}}(\mathbf{X})=\widetilde{\mathbf{v}}(\mathbf{X}_k) } \\
\left.\frac{\partial \|\widetilde{\mathbf{v}}(\mathbf{X})\|_2}{\partial \mathbf{v}_\mathrm{I}(\mathbf{X})}\right|_{ \widetilde{\mathbf{v}}(\mathbf{X})=\widetilde{\mathbf{v}}(\mathbf{X}_k) }
\end{matrix} \right] =
\left[\begin{matrix}
\frac{\mathbf{v}_\mathrm{R}(\mathbf{X}_k)}{\|\widetilde{\mathbf{v}}(\mathbf{X}_k)\|_2} \\
\frac{\mathbf{v}_\mathrm{I}(\mathbf{X}_k)}{\|\widetilde{\mathbf{v}}(\mathbf{X}_k)\|_2}
\end{matrix}\right]
= \frac{\widetilde{\mathbf{v}}(\mathbf{X}_k)}{\|\widetilde{\mathbf{v}}(\mathbf{X}_k)\|_2}.
\end{equation}\normalsize

Based on (\ref{Derivative_k}), we obtain
\small \begin{align}\label{67}
\mathcal{T}_{ \widetilde{\mathbf{v}}(\mathbf{X}_k) }\left( \|\widetilde{\mathbf{v}}(\mathbf{X})\|_2 \right)
=\ &
\|\widetilde{\mathbf{v}}(\mathbf{X}_k)\|_2  +  \left\langle \left.\frac{\partial \|\widetilde{\mathbf{v}}(\mathbf{X})\|_2}{\partial\widetilde{\mathbf{v}}(\mathbf{X})}\right|_{ \widetilde{\mathbf{v}}(\mathbf{X})=\widetilde{\mathbf{v}}(\mathbf{X}_k) }, \left[\widetilde{\mathbf{v}}(\mathbf{X})-\widetilde{\mathbf{v}}(\mathbf{X}_k)\right] \right\rangle \nonumber\\
=\ &
\|\widetilde{\mathbf{v}}(\mathbf{X}_k)\|_2 + \left(\frac{\widetilde{\mathbf{v}}(\mathbf{X}_k)}{\|\widetilde{\mathbf{v}}(\mathbf{X}_k)\|_2}\right)^{\mathrm{T}}   \left[\widetilde{\mathbf{v}}(\mathbf{X})-\widetilde{\mathbf{v}}(\mathbf{X}_k)\right] \nonumber\\
=\ & \|\widetilde{\mathbf{v}}(\mathbf{X}_k)\|_2^{-1} \widetilde{\mathbf{v}}^{\mathrm{T}}(\mathbf{X}_k) \widetilde{\mathbf{v}}(\mathbf{X}) \nonumber\\
=\ & \|\widetilde{\mathbf{v}}(\mathbf{X}_k)\|_2^{-1}
\sum_{i=1}^{(N+1)^2}   tr\left( \frac{1}{2}(\mathbf{S}_{i} + \mathbf{S}_{i}^\mathrm{H}) \mathbf{X}_k\right)  tr\left( \frac{1}{2}(\mathbf{S}_{i} + \mathbf{S}_{i}^\mathrm{H}) \mathbf{X}\right)   \nonumber\\
&\ \ \ \ \ \ + \|\widetilde{\mathbf{v}}(\mathbf{X}_k)\|_2^{-1}
\sum_{i=1}^{(N+1)^2}   tr\left( -\frac{j}{2}(\mathbf{S}_{i} - \mathbf{S}_{i}^\mathrm{H}) \mathbf{X}_k\right)  tr\left( -\frac{j}{2}(\mathbf{S}_{i} - \mathbf{S}_{i}^\mathrm{H}) \mathbf{X}\right)  \nonumber\\
=\ & \sum_{i=1}^{(N+1)^2} tr\left( \|\widetilde{\mathbf{v}}(\mathbf{X}_k)\|_2^{-1}  tr\left( \frac{1}{4}(\mathbf{S}_{i} + \mathbf{S}_{i}^\mathrm{H}) \mathbf{X}_k\right) (\mathbf{S}_{i} + \mathbf{S}_{i}^\mathrm{H}) \mathbf{X}\right) \nonumber\\
&\ \ \ \ \ \  + \sum_{i=1}^{(N+1)^2} tr\left( \|\widetilde{\mathbf{v}}(\mathbf{X}_k)\|_2^{-1}  tr\left( -\frac{1}{4}(\mathbf{S}_{i} - \mathbf{S}_{i}^\mathrm{H}) \mathbf{X}_k\right) (\mathbf{S}_{i} - \mathbf{S}_{i}^\mathrm{H}) \mathbf{X}\right) 
=\  tr\left(\mathbf{\Upsilon}_{\mathbf{X}_k}  \mathbf{X}\right),
\end{align}\normalsize
where $\mathbf{\Upsilon}_{\mathbf{X}_k}$ is given by
\small \begin{equation}\label{Upsilon_Qk}
\begin{split}
\mathbf{\Upsilon}_{\mathbf{X}_k}
= \frac{1}{4 \|\widetilde{\mathbf{v}}(\mathbf{X}_k)\|_2}  \left\{
\sum_{i=1}^{(N+1)^2}   tr\left( (\mathbf{S}_{i} + \mathbf{S}_{i}^\mathrm{H}) \mathbf{X}_k\right) (\mathbf{S}_{i} + \mathbf{S}_{i}^\mathrm{H}) -    \sum_{i=1}^{(N+1)^2}  tr\left( (\mathbf{S}_{i} - \mathbf{S}_{i}^\mathrm{H}) \mathbf{X}_k\right) (\mathbf{S}_{i} - \mathbf{S}_{i}^\mathrm{H})
\right\}.
\end{split}
\end{equation}\normalsize

Because $\|\widetilde{\mathbf{v}}(\mathbf{X}_k)\|_2$, $\left.\frac{\partial \|\widetilde{\mathbf{v}}(\mathbf{X})\|_2}{\partial \widetilde{\mathbf{v}}(\mathbf{X})}\right|_{ \widetilde{\mathbf{v}}(\mathbf{X})=\widetilde{\mathbf{v}}(\mathbf{X}_k) }$ and $\left[\widetilde{\mathbf{v}}(\mathbf{X})-\widetilde{\mathbf{v}}(\mathbf{X}_k)\right]$ are all real, $\mathcal{T}_{ \widetilde{\mathbf{v}}(\mathbf{X}_k) }\left( \|\widetilde{\mathbf{v}}(\mathbf{X})\|_2 \right)$ is real as well. In addition, because $\|\mathbf{v}(\mathbf{X})\|_2 - \mu \sqrt{\mathcal{G}}\geq 0$ and $\|\widetilde{\mathbf{v}}(\mathbf{X})\|_2 -\mu \sqrt{\mathcal{G}} \geq 0$ are equivalent, $\mathcal{T}_{ \mathbf{v}(\mathbf{X}_k) }\left( \|\mathbf{v}(\mathbf{X})\|_2 \right) - \mu \sqrt{\mathcal{G}}\geq 0$ and $\mathcal{T}_{ \widetilde{\mathbf{v}}(\mathbf{X}_k) }\left( \|\widetilde{\mathbf{v}}(\mathbf{X})\|_2 \right) - \mu \sqrt{\mathcal{G}}\geq 0$ are also equivalent. Consequently, we prove that $\mathcal{T}_{ \widetilde{\mathbf{v}}(\mathbf{X}_k) }\left( \|\widetilde{\mathbf{v}}(\mathbf{X})\|_2 \right)$ is a real representation of $\mathcal{T}_{ \mathbf{v}(\mathbf{X}_k) }\left( \|\mathbf{v}(\mathbf{X})\|_2 \right)$.



\begin{spacing}{1.17}

\ifCLASSOPTIONcaptionsoff
  \newpage
\fi

\end{spacing}



\begin{thebibliography}{1}

\bibitem{6G3}
X. You, \textit{et al.}, "Towards 6G wireless communication networks: Vision, enabling technologies, and new paradigm shifts," \emph{SCIENCE CHINA Information Sciences}, vol. 64, no. 1, pp. 1-74, Jan. 2021.

\bibitem{JCR-1}
J. A. Zhang, \textit{et al.}, "Enabling joint communication and radar sensing in mobile networks - A survey," \emph{IEEE
Communications Surveys \& Tutorials}, vol. 24, no. 1, pp. 306-345, Firstquarter 2022.

\bibitem{JCR-2}
J. A. Zhang, \textit{et al.}, "An overview of signal processing techniques for joint communication and radar sensing," \emph{IEEE Journal of Selected Topics in Signal Processing}, vol. 15, no. 6, pp. 1295-1315, Nov. 2021.

\bibitem{ISAC}
F. Liu, \textit{et al.}, "Integrated sensing and communications: Towards dual-functional wireless networks for 6G and beyond," \emph{IEEE Journal on Selected Areas in Communications}, Early Access, Mar. 2022, DOI: 10.1109/JSAC.2022.3156632.

\bibitem{My-TWC}
Z. Xing, R. Wang, J. Wu and E. Liu, "Achievable rate analysis and phase shift optimization on intelligent reflecting surface with hardware impairments," \emph{IEEE Transactions on Wireless Communications}, vol. 20, no. 9, pp. 5514-5530, Sept. 2021.

\bibitem{My-ICCC}
Z. Xing, \textit{et al.}, "Location-aware beamforming design for reconfigurable intelligent surface aided communication system," in \emph{Proc. IEEE/CIC International Conference on Communications in China (ICCC)}, Xiamen, China, Jul. 2021, pp. 201-206.

\bibitem{Propose-IRS-1}
C. Liaskos, \textit{et al}., “A new wireless communication paradigm through software-controlled metasurfaces,” \emph{IEEE Communications Magazine}, vol. 56, no. 9, pp. 162-169, Sept. 2018.

\bibitem{Intelligent-Wall}
L. Subrt and P. Pechac, “Intelligent walls as autonomous parts of smart indoor environments,” \emph{IET Communications}, vol. 6, no. 8, pp. 1004-1010, May 2012.

\bibitem{Landmark-1}
Q. Wu and R. Zhang, “Intelligent reflecting surface enhanced wireless network via joint active and passive beamforming,” \emph{IEEE Transactions on Wireless Communications}, vol. 18, no. 11, pp. 5394-5409, Nov. 2019.

\bibitem{Landmark-2}
C. Huang, \textit{et al}., “Reconfigurable intelligent surfaces for energy efficiency in wireless communication,” \emph{IEEE Transactions on Wireless Communications}, vol. 18, no. 8, pp. 4157-4170, Aug. 2019.

{\color{black}
\bibitem{IRS-Survey-1}
M. Jian, \textit{et al}., “Reconfigurable intelligent surfaces for wireless communications: Overview of hardware designs, channel models, and estimation techniques,” \emph{Intelligent and Converged Networks}, vol. 3, no. 1, pp. 1-32, Mar. 2022.

\bibitem{IRS-Survey-2}
J. Xu, \textit{et al}., “Reconfiguring wireless environment via intelligent surfaces for 6G: Reflection, modulation, and security,” \emph{arXiv:2208.10931v1}, Aug. 2022, [Online]. Available: https://arxiv.org/abs/2208.10931.}

\bibitem{Q.Wu2020(CM)}
Q. Wu and R. Zhang, “Towards smart and reconfigurable environment: Intelligent reflecting surface aided wireless network,” \emph{IEEE Communications Magazine}, vol. 58, no. 1, pp. 106-112, Jan. 2020.

\bibitem{L.Dai2020(Access)}
L. Dai, \textit{et al}., “Reconfigurable intelligent surface-based wireless communications: Antenna design, prototyping, and experimental results,” \emph{IEEE Access}, vol. 8, pp. 45913-45923, Mar. 2020.

\bibitem{Information-Transfer}
W. Yan, X. Yuan and X. Kuai, “Passive beamforming and information transfer via large intelligent surface,” \emph{IEEE Wireless Communications Letters}, vol. 9, no. 4, pp. 533–537, Apr. 2020.

\bibitem{Index-Modulation}
E. Basar, “Reconfigurable intelligent surface-based index modulation: A new beyond MIMO paradigm for 6G,” \emph{IEEE Transactions on Communications}, vol. 68, no. 5, pp. 3187–3196, Feb. 2020.

\bibitem{Secure-Transmission}
M. Cui, G. Zhang and R. Zhang, “Secure wireless communication via intelligent reflecting surface,” \emph{IEEE Wireless Communications Letters}, vol. 8, no. 5, pp. 1410-1414, Oct. 2019.

\bibitem{Secure-Transmission-2}
L. Dong and H.-M. Wang, “Enhancing secure MIMO transmission via intelligent reflecting surface,” \emph{IEEE Transactions on Wireless Communications}, vol. 19, no. 11, pp. 7543-7556, Nov. 2020.

{\color{black}\bibitem{Secure-Transmission-3}
G. C. Alexandropoulos, \textit{et al.}, “Safeguarding MIMO communications with
reconfigurable metasurfaces and artificial noise,” in \emph{Proc. IEEE International Conference on Communications (ICC)}, Montreal, QC, Canada, Jun. 2021, pp. 1-6.}

\bibitem{N2}
Q. Wu and R. Zhang, "Beamforming optimization for wireless network aided by intelligent reflecting surface with discrete phase shifts," \emph{IEEE Transactions on Communications}, vol. 68, no. 3, pp. 1838-1851, Mar. 2020.

{\color{black}
\bibitem{R3-Comm3-2}
M. Rahal, \textit{et al.}, "Arbitrary beam pattern approximation via RISs with measured element responses," in \emph{Proc. Joint European Conference on Networks and Communications \& 6G Summit (EuCNC/6G Summit)}, Grenoble, France, Jun. 2022, pp. 506-511.

\bibitem{R3-Comm3-4}
G. C. Alexandropoulos, V. Jamali, R. Schober and H. V. Poor, "Near-field hierarchical beam management for RIS-enabled millimeter wave multi-antenna systems," \emph{arXiv:2203.15557}, Mar. 2022, [Online]. Available: https://arxiv.org/abs/2203.15557.}

\bibitem{ACR-Maximization-2}
C. Huang, A. Zappone, M. Debbah and C. Yuen, "Achievable rate maximization by passive intelligent mirrors," in \emph{Proc. IEEE International Conference on Acoustics, Speech and Signal Processing}, Calgary, AB, Canada, Apr. 2018, pp. 3714-3718.

\bibitem{Outage-Probability}
C. Guo, Y. Cui, F. Yang and L. Ding, "Outage probability analysis and minimization in intelligent reflecting surface-assisted MISO systems," \emph{IEEE Communications Letters}, vol. 24, no. 7, pp. 1563-1567, Jul. 2020.

\bibitem{Outage-Probability-2}
B. Lu, R. Wang and Y. Liu, "Outage probability of intelligent reflecting surface assisted full duplex two-way communications," \emph{IEEE Communications Letters}, vol. 26, no. 2, pp. 286-290, Feb. 2022.

\bibitem{BER}
J. Li, R. Wang and E. Liu, "Passive beamforming design for IRS communication system with few-bit ADCs," in \emph{Proc. 4th International Conference on Information Communication and Signal Processing (ICICSP)}, Shanghai, China, Sept. 2021, pp. 1-6.

{\color{black}
\bibitem{R3-Comm3-1}
K. Keykhosravi, \textit{et al}., “Leveraging RIS-enabled smart signal propagation for solving infeasible localization problems,” \emph{arXiv:2204.11538v1}, Apr. 2022, [Online]. Available: https://arxiv.org/abs/2204.11538.}

\bibitem{Localization-1}
J. He, \textit{et al}., "Adaptive beamforming design for mmwave RIS-Aided joint localization and communication," in \emph{Proc. IEEE Wireless Communications and Networking Conference Workshops}, Seoul, Korea (South), Apr. 2020, pp. 1-6.

{\color{black}
\bibitem{R3-Comm3-3}
Z. A.-Shaban, \textit{et al.}, "Near-field localization with a reconfigurable intelligent surface acting as lens," in \emph{Proc. IEEE International Conference on Communications (ICC)}, Montreal, QC, Canada, Jun. 2021, pp. 1-6.}

\bibitem{Localization-2}
A. Elzanaty, A. Guerra, F. Guidi and M.-S. Alouini, "Reconfigurable
intelligent surfaces for localization: Position and orientation error bounds," \emph{IEEE Transactions on Signal Processing}, vol. 69, pp. 5386-5402, Oct. 2021.

\bibitem{Localization-3}
H. Zhang, \textit{et al}., "MetaLocalization: Reconfigurable intelligent surface aided multi-user wireless indoor localization," \emph{IEEE Transactions on Wireless Communications}, vol. 20, no. 12, pp. 7743-7757, Dec. 2021.

\bibitem{RIS-Radar-1}
S. Buzzi, E. Grossi, M. Lops and L. Venturino, "Radar target detection aided by reconfigurable intelligent surfaces," \emph{IEEE Signal Processing Letters}, vol. 28, pp. 1315-1319, Jun. 2021.

\bibitem{RIS-Radar-2}
A. Aubry, A. D. Maio and M. Rosamilia, "Reconfigurable intelligent surfaces for N-LOS radar surveillance," \emph{IEEE Transactions on Vehicular Technology}, vol. 70, no. 10, pp. 10735-10749, Oct. 2021.

\bibitem{RIS-Radar-3}
H. Zhang, \textit{et al}., "MetaRadar: Multi-target detection for reconfigurable intelligent surface aided radar systems," \emph{IEEE Transactions on Wireless Communications}, Early Access, Mar. 2022, DOI: 10.1109/TWC.2022.3153792.

\bibitem{RIS-ISAC-1}
Z.-M. Jiang, \textit{et al.}, "Intelligent reflecting surface aided dual-function radar and communication system," \emph{IEEE Systems Journal}, vol. 16, no. 1, pp. 475-486, Feb. 2021.

\bibitem{RIS-ISAC-2}
X. Song, \textit{et al.}, "Joint transmit and reflective beamforming for IRS-assisted integrated sensing and communication," \emph{arXiv:2111.13511v1}, Nov. 2021, [Online]. Available: https://arxiv.org/abs/2111.13511.

\bibitem{RIS-ISAC-3}
R. Liu, M. Li and A. L. Swindlehurst, "Joint beamforming and reflection design for RIS-assisted ISAC systems," \emph{arXiv:2203.00265v1}, Mar. 2022, [Online]. Available: https://arxiv.org/abs/2203.00265.

\bibitem{RIS-ISAC-4}
X. Wang, Z. Fei, Z. Zheng and J. Guo, "Joint waveform design and passive beamforming for RIS-assisted dual-functional radar-communication system," \emph{IEEE Transactions on Vehicular Technology}, vol. 70, no. 5, pp. 5131-5136, May. 2021.

\bibitem{RIS-ISAC-Huang}
X. Tong, Z. Zhang, J. Wang, C. Huang and M. Debbah, "Joint multi-user communication and sensing exploiting both signal and environment sparsity," \emph{IEEE Journal of Selected Topics in Signal Processing}, vol. 15, no. 6, pp. 1409-1422, Nov. 2021.

\bibitem{RIS-ISAC-5}
R.S. Prasobh Sankar and S. P. Chepuri, "Beamforming in hybrid RIS assisted integrated sensing and communication systems," \emph{arXiv:2203.05902v1}, Mar. 2022, [Online]. Available: https://arxiv.org/abs/2203.05902.

\bibitem{RIS-ISAC-6}
X. Wang, \textit{et al.}, "Joint waveform and discrete phase shift design for RIS-assisted integrated sensing and communication system under Cramer-Rao bound constraint," \emph{IEEE Transactions on Vehicular Technology}, vol. 71, no. 1, pp. 1004-1009, Jan. 2022.

\bibitem{Secure_ISAC}
M. Hua, \textit{et al.}, "Secure intelligent reflecting surface aided integrated sensing and communication," \emph{arXiv:2207.09095}, Jul. 2022, [Online]. Available: https://arxiv.org/abs/2207.09095.

{\color{black}
\bibitem{R3-Comm4-1}
A. L. Moustakas, G. C. Alexandropoulos and M. Debbah, "Reconfigurable intelligent surfaces and capacity optimization: A large system analysis," \emph{arXiv:2208.09615}, Aug. 2022, [Online]. Available: https://arxiv.org/abs/2208.09615.}

\bibitem{Channel Estimation}
L. Wei \textit{et al}., “Channel estimation for RIS-empowered multi-user MISO wireless communications,” \emph{IEEE Transactions on Communications}, vol. 69, no. 6, pp. 4144-4157, Mar. 2021.

\bibitem{Physical-Model-TCOM}
M. Najafi, V. Jamali, R. Schober and H. V. Poor, "Physics-based modeling and scalable optimization of large intelligent reflecting surfaces," \emph{IEEE Transactions on Communications}, vol. 69, no. 4, pp. 2673-2691, Apr. 2021.

\bibitem{Physical-Model-Yuan}
M. Zhang and X. Yuan, "Intelligent reflecting surface aided MIMO
with cascaded line-of-sight links: Channel modelling and capacity analysis," \emph{arXiv:2109.08913}, Sept. 2021, [Online]. Available: https://arxiv.org/abs/2109.08913.

\bibitem{RCS}
M. C. Rezende, \textit{et al.}, "Radar cross section measurements (8-12 GHz) of magnetic and dielectric microwave absorbing thin sheets," \emph{Revista de Fisica Aplicada e Instrumentacao}, vol. 15, no. 1, pp. 24-29, Dec. 2002.

\bibitem{TWC-localization}
A. Shahmansoori, \textit{et al.}, "Position and orientation estimation through millimeter-wave MIMO in 5G systems," \emph{IEEE Transactions on Wireless Communications}, vol. 17, no. 3, pp. 1822-1835, Mar. 2018.

\bibitem{Book}
S. M. Kay, "Fundamentals of statistical signal processing: Volume II: Detection theory," Englewood Cliffs, NJ: Prentice-Hall, 1993.

{\color{black}\bibitem{Q-func}
I. M. Tanash and T. Riihonen, "Global minimax approximations and bounds for the Gaussian Q-function by sums of exponentials," \emph{IEEE Transactions on Communications}, vol. 68, no. 10, pp. 6514-6524, Oct. 2020.}

\bibitem{SDR-SPM}
Z.-Q. Luo, \textit{et al}., "Semidefinite relaxation of quadratic optimization problems," \emph{IEEE Signal Processing Magazine}, vol. 27, no. 3, pp. 20-34, May. 2010.

\bibitem{SDP-CCP}
M. Á. Vázquez, L. Blanco and A. I. Pérez-Neira, "Spectrum sharing backhaul satellite-terrestrial systems via analog beamforming," \emph{IEEE Journal of Selected Topics in Signal Processing}, vol. 12, no. 2, pp. 270-281, May 2018.




















































\end{thebibliography}
\end{document}